\documentclass[11pt]{article}

\usepackage{setspace}
\usepackage[T1]{fontenc}
\usepackage[utf8]{inputenc}
\usepackage[margin=1.1in]{geometry}
\usepackage{amsmath, amsfonts, amssymb, bm, mathtools}
\usepackage{bbm}
\usepackage{booktabs}
\usepackage{graphicx}
\usepackage{enumitem}
\usepackage{xr}
\usepackage{xcolor}
\usepackage{amsthm}
\usepackage{thmtools}
\usepackage{algorithm}
\usepackage{algpseudocode}

\usepackage{natbib}
\usepackage{hyperref}
\usepackage{cleveref}

\usepackage{longtable}

\usepackage{makecell,tabularx}
\usepackage{multirow}
\usepackage{threeparttable}

\usepackage{subcaption}

\usepackage{pifont} 

\usepackage{authblk} 

\hypersetup{colorlinks=true, linkcolor=blue, citecolor=blue, urlcolor=blue}

\definecolor{MySlate}{RGB}{70, 80, 90}
\definecolor{SoftBG}{RGB}{245,247,250}
\definecolor{MissDot}{RGB}{180,50,50}    

\usepackage{tikz}
\usepackage{adjustbox}
\usetikzlibrary{arrows.meta,positioning,calc,fit,decorations.pathreplacing}
\tikzset{
  patient/.style={font=\bfseries\Large, text=MySlate},
  timeaxis/.style={line width=1.2pt, MySlate},
  tickmark/.style={MySlate, line width=0.8pt},
  yobs/.style={circle, draw=black, fill=black, inner sep=0pt, minimum size=12pt},       
  ymiss/.style={circle, draw=MissDot, line width=2.0pt, fill=white, inner sep=0pt, minimum size=12pt},  
  covtag/.style={font=\Large, text=MySlate},
  tablebox/.style={draw=MySlate, rounded corners=3pt, line width=0.8pt, fill=SoftBG},
  headcell/.style={font=\Large\bfseries, text=MySlate},
  sepLine/.style={MySlate, line width=0.7pt},
  maparrow/.style={->, MySlate, line width=1.4pt},
  yNA/.style={font=\Large, text=MissDot, draw=MissDot, fill=MissDot!14,
              rounded corners=2pt, inner sep=2pt}
}

\usepackage{xcolor}
\usepackage{colortbl} 

\definecolor{MySlate}{RGB}{70,80,90}
\definecolor{SoftBG}{RGB}{245,247,250}
\definecolor{MissDot}{RGB}{180,50,50}    

\usepackage{tikz}
\usepackage{adjustbox}
\usetikzlibrary{arrows.meta,positioning,calc,fit,decorations.pathreplacing}
\tikzset{
  patient/.style={font=\bfseries\Large, text=MySlate},
  timeaxis/.style={line width=1.2pt, MySlate},
  tickmark/.style={MySlate, line width=0.8pt},
  yobs/.style={circle, draw=black, fill=black, inner sep=0pt, minimum size=12pt},       
  ymiss/.style={circle, draw=MissDot, line width=2.0pt, fill=white, inner sep=0pt, minimum size=12pt},  
  covtag/.style={font=\Large, text=MySlate},
  tablebox/.style={draw=MySlate, rounded corners=3pt, line width=0.8pt, fill=SoftBG},
  headcell/.style={font=\Large\bfseries, text=MySlate},
  sepLine/.style={MySlate, line width=0.7pt},
  maparrow/.style={->, MySlate, line width=1.4pt},
  yNA/.style={font=\Large, text=MissDot, draw=MissDot, fill=MissDot!14,
              rounded corners=2pt, inner sep=2pt}
}

\usepackage{tikz}
\usetikzlibrary{shapes, arrows.meta, positioning, calc, backgrounds, decorations.pathmorphing}
\usepackage{booktabs}

\definecolor{primaryBlue}{RGB}{0, 85, 150}      
\definecolor{accessGreen}{RGB}{46, 125, 50}     
\definecolor{barrierRed}{RGB}{198, 40, 40}      
\definecolor{waterBlue}{RGB}{227, 242, 253}     
\definecolor{oceanLine}{RGB}{33, 150, 243}      
\definecolor{hiddenGray}{RGB}{117, 117, 117}    

\tikzset{
    semitransparent/.style={opacity=0.8, text opacity=1},
    box/.style={draw, rectangle, rounded corners, minimum width=2.2cm, minimum height=0.8cm, align=center, font=\small\bfseries},
    cloudNode/.style={draw=gray, cloud, cloud puffs=11, fill=gray!10, align=center, font=\footnotesize},
    mainArrow/.style={->, >={Stealth[length=3mm]}, line width=1.5pt},
    waveDecoration/.style={decorate, decoration={snake, amplitude=1.5mm, segment length=5mm}}
}

\declaretheorem[name=Assumption,numberwithin=section]{assumption}
\declaretheorem[name=Proposition,numberwithin=section]{proposition}
\declaretheorem[name=Theorem,numberwithin=section]{theorem}
\declaretheorem[name=Remark,numberwithin=section]{remark}
\declaretheorem[name=Lemma,numberwithin=section]{lemma}
\declaretheorem[name=Corollary,numberwithin=section]{corollary}

\DeclareMathOperator{\Var}{var}
\DeclareMathOperator{\Cov}{cov}
\DeclareMathOperator{\E}{E}
\DeclareMathOperator{\pr}{pr}

\newcommand{\dint}{\mathrm{d}}

\title{\textbf{Joint Modeling of Longitudinal EHR Data with Shared Random Effects for Informative Visiting and Observation Processes}}

\author[1]{Cheng-Han Yang}
\author[2]{Xu Shi}
\author[1]{Bhramar Mukherjee\thanks{Corresponding author: \href{mailto:bhramar.mukherjee@yale.edu}{bhramar.mukherjee@yale.edu}}}

\affil[1]{Department of Biostatistics, Yale University}
\affil[2]{Department of Biostatistics, University of Michigan}

\date{}

\begin{document}
\maketitle

\begin{abstract}
Longitudinal electronic health record data offer unprecedented opportunities to study biomarker trajectories in association with genetic, environmental, and social determinants of health. However, association estimates from standard longitudinal models designed for regular observation times may be biased by a two-stage hierarchical missingness mechanism. The first stage is the visiting process, where encounters occur at irregular times largely driven by patient health status; it introduces a bias known as informative presence. The second stage is the observation process: conditional on a visit occurring, biomarkers---the longitudinal outcomes of interest---are selectively measured mostly based on clinical judgment and the underlying health trajectory of the patient. This observation process induces a second source of bias termed informative observation.
Two critical gaps remain in this field. First, the joint influence of these mechanisms on statistical inference has not been systematically evaluated. Second, models accommodating both mechanisms remain limited. To address these gaps, we conduct extensive benchmarking studies and propose a unified semiparametric joint modeling framework that simultaneously characterizes the visiting, biomarker observation, and longitudinal outcome processes. Central to this framework is a shared subject-specific Gaussian latent variable that captures unmeasured frailty and induces dependence across all three components. 
We also introduce a sequential procedure that imputes missing biomarkers prior to adjusting for irregular visiting, and examine its performance.
To ensure computational tractability for large-scale studies, we develop a three-stage estimation procedure and establish the consistency and asymptotic normality of our estimators.
Simulation results demonstrate that our method yields unbiased estimates under the two-stage hierarchical mechanism, whereas existing approaches can be substantially biased. In particular, methods that adjust only for irregular visiting may exhibit even greater bias than approaches that ignore both mechanisms.
We further apply our framework---particularly suited for outpatient visiting settings---to data from the All of Us Research Program to investigate the association between two neighborhood-level socioeconomic status indicators and the trajectories of six biomarkers measured through blood tests.
\end{abstract}


\section{Introduction}

Longitudinal data collected in routine clinical care pose inferential challenges that differ fundamentally from those in prospective cohort studies or clinical trials. In designed studies, a protocol governs both the schedule of patient encounters and the panel of measurements obtained at each encounter. In routine care, neither is predetermined; instead, both the timing of a visit and whether a particular biomarker is measured arise from a complex interplay of disease severity, clinical judgment, and patient characteristics and behavior \citep{hripcsak2013next, goldstein2016controlling}. When these clinically driven mechanisms depend on the latent health trajectory under study, analyses that treat the data collection process as fixed or exogenous yield systematically biased estimates of exposure--biomarker associations \citep{lin2004analysis, pullenayegum2016longitudinal}. These issues are of growing importance as electronic health record (EHR) data are now widely used in large-scale biomedical initiatives such as the All of Us Research Program in the United States \citep{all2019all} and the UK Biobank \citep{bycroft2018uk}, where routinely collected laboratory measurements serve as longitudinal outcomes for studying associations with clinical, demographic, genetic, environmental, behavioral, and social determinants of health. A prominent example is the laboratory-wide association study \citep[LabWAS;][]{goldstein2020labwas, dennis2021clinical}, in which repeated biomarker measurements are linked to genetic risk factors on a phenome-wide scale. Although LabWAS has yielded empirical insights, collapsing repeated measurements into simple summaries neglects within-subject variability and fails to account for the informative nature of the data collection process. Specifically, two clinically driven mechanisms---the timing of visits and the decision to measure a particular biomarker---distort the available data in ways that standard methods do not address.

These two mechanisms reflect a two-stage hierarchy inherent in EHR data collection (Figure~\ref{fig:IPIO}). The first stage is the visiting process, which generates patient encounters at irregular, clinically driven times. When patients with greater disease burden visit more frequently, the resulting overrepresentation of adverse health states introduces a bias known as informative presence \citep[IP;][]{lin2004analysis, sun2007regression}. 
However, the direction of bias under IP is not definitive: limited access to care may reduce visit frequency among the sickest patients \citep{obermeyer2019dissecting}, while heterogeneities in diagnostic testing rates may inflate encounter frequency \citep{ellenbogen2024race}, making the magnitude and sign of bias under IP difficult to determine \emph{a priori}. The second stage is the observation process: conditional on a visit occurring, a clinician decides whether to measure the biomarker of interest---which serves as the longitudinal outcome in our modeling framework---based on the patient's current symptoms, prior test results, and clinical judgment. When this decision depends on factors related to the health trajectory that are not fully captured by observed covariates, it gives rise to a second source of bias that we term informative observation \citep[IO;][]{wells2013strategies, haneuse2016general}. Consider patient 2 in the left panel of Figure~\ref{fig:IPIO} as an illustration of this hierarchy. At the first visit, denoted by an open circle, no laboratory measurement is taken, reflecting the absence of clinical indication. At the second visit, a measurement is observed (solid dot) in response to an adverse condition. Finally, at the third visit, a follow-up measurement is obtained to confirm recovery; notably, this decision is conditional on both the patient's current condition and clinical history. 
In the All of Us data, biomarker measurements are observed only at a fraction of outpatient visits; for example, even among patients with at least one glucose record, the median per-patient measurement rate is only 9.1\%. The extent of missingness varies substantially across biomarkers (Table~\ref{tab:biomarkers}).
Moreover, this missingness is unlikely to be ignorable: the unmeasured factors driving both the frequency of visits and the decision to measure---including latent health status, access to care, and heterogeneity in clinical decision-making---are often correlated with the biomarker trajectory under study, so that both stages constitute a missing not at random (MNAR) mechanism that cannot be resolved by addressing either stage in isolation.

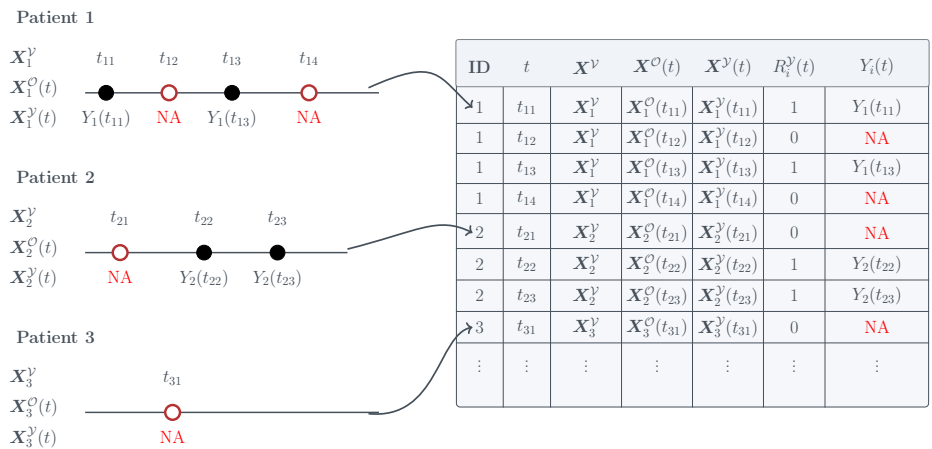
\begin{figure}[htbp]
\centering
\begin{adjustbox}{scale=0.465}
\newcommand{\TextScale}{2.567}
\begin{tikzpicture}[x=1cm,y=1cm]

\providecolor{MySlate}{RGB}{64,82,99}
\providecolor{MyGray}{RGB}{240,243,247}

\def\xL{2.6}
\def\xR{11.0}
\def\Oshift{2.3}            
\def\nameoffset{1.8}        

\def\panelTop{3.2}
\def\panelBot{-9.5}
\pgfmathsetmacro{\rowSep}{(\panelTop - \panelBot)/2}
\pgfmathsetmacro{\yA}{\panelTop}
\pgfmathsetmacro{\yB}{\panelTop - \rowSep}
\pgfmathsetmacro{\yC}{\panelTop - 2*\rowSep}

\def\covV{0.0}     
\def\covO{-1.2}    
\def\covY{-2.4}    

\pgfmathsetmacro{\Xcol}{\xL-\Oshift}

\pgfmathsetmacro{\yTimeA}{\yA + \covO}
\pgfmathsetmacro{\yTimeB}{\yB + \covO}
\pgfmathsetmacro{\yTimeC}{\yC + \covO}

\begin{scope}[yscale=0.72, every node/.style={transform shape=false}]

\node[patient, anchor=west] at (0.5,\yA+\nameoffset) {Patient 1};

\node[covtag, anchor=base west] at (\Xcol,\yA+\covV) {$\bm X_1^{\mathcal V}$};
\node[covtag, anchor=base west] at (\Xcol,\yA+\covO) {$\bm X_1^{\mathcal O}(t)$};
\node[covtag, anchor=base west] at (\Xcol,\yA+\covY) {$\bm X_1^{\mathcal Y}(t)$};

\draw[timeaxis] (\xL,\yTimeA) -- (\xR,\yTimeA);

\foreach \x/\lab in {3.2/$t_{11}$, 5.0/$t_{12}$, 6.8/$t_{13}$, 9.0/$t_{14}$}{
  \draw[tickmark] (\x,\yTimeA-0.14) -- (\x,\yTimeA+0.14);
  \node[covtag, anchor=base] at (\x,\yA+\covV) {\lab};
}

\foreach[count=\i from 1] \x in {3.2,5.0,6.8,9.0}{
  \ifnum\i=1 \node[yobs]  at (\x,\yTimeA) {}; \fi
  \ifnum\i=2 \node[ymiss] at (\x,\yTimeA) {}; \fi
  \ifnum\i=3 \node[yobs]  at (\x,\yTimeA) {}; \fi
  \ifnum\i=4 \node[ymiss] at (\x,\yTimeA) {}; \fi
  \node[covtag, anchor=base] at (\x,\yA+\covY) {%
    \ifnum\i=2
      \textcolor{red}{NA}%
    \else\ifnum\i=4
      \textcolor{red}{NA}%
    \else
      $Y_1(t_{1\i})$%
    \fi\fi
  };
}

\node[patient, anchor=west] at (0.5,\yB+\nameoffset) {Patient 2};

\node[covtag, anchor=base west] at (\Xcol,\yB+\covV) {$\bm X_2^{\mathcal V}$};
\node[covtag, anchor=base west] at (\Xcol,\yB+\covO) {$\bm X_2^{\mathcal O}(t)$};
\node[covtag, anchor=base west] at (\Xcol,\yB+\covY) {$\bm X_2^{\mathcal Y}(t)$};

\draw[timeaxis] (\xL,\yTimeB) -- (\xR-0.9,\yTimeB);

\foreach \x/\lab in {3.6/$t_{21}$, 6.0/$t_{22}$, 8.1/$t_{23}$}{
  \draw[tickmark] (\x,\yTimeB-0.14) -- (\x,\yTimeB+0.14);
  \node[covtag, anchor=base] at (\x,\yB+\covV) {\lab};
}

\foreach[count=\i from 1] \x in {3.6,6.0,8.1}{
  \ifnum\i=1 \node[ymiss] at (\x,\yTimeB) {}; \fi
  \ifnum\i=2 \node[yobs]  at (\x,\yTimeB) {}; \fi
  \ifnum\i=3 \node[yobs]  at (\x,\yTimeB) {}; \fi
  \node[covtag, anchor=base] at (\x,\yB+\covY) {%
    \ifnum\i=1
      \textcolor{red}{NA}%
    \else
      $Y_2(t_{2\i})$%
    \fi
  };
}

\node[patient, anchor=west] at (0.5,\yC+\nameoffset) {Patient 3};

\node[covtag, anchor=base west] at (\Xcol,\yC+\covV) {$\bm X_3^{\mathcal V}$};
\node[covtag, anchor=base west] at (\Xcol,\yC+\covO) {$\bm X_3^{\mathcal O}(t)$};
\node[covtag, anchor=base west] at (\Xcol,\yC+\covY) {$\bm X_3^{\mathcal Y}(t)$};

\draw[timeaxis] (\xL,\yTimeC) -- (\xR,\yTimeC);

\draw[tickmark] (5.1,\yTimeC-0.14) -- (5.1,\yTimeC+0.14);
\node[covtag, anchor=base] at (5.1,\yC+\covV) {$t_{31}$};

\node[ymiss] at (5.1,\yTimeC) {};
\node[covtag, anchor=base, text=red] at (5.1,\yC+\covY) {NA};

\coordinate (LeftAnchor1) at (10.7,\yTimeA+0.20);
\coordinate (LeftAnchor2) at (10.1,\yTimeB+0.14);
\coordinate (LeftAnchor3) at (10.7,\yTimeC-0.05);

\end{scope} 

\node[tablebox, minimum width=13.5cm, minimum height=10.5cm, anchor=west] (T) at (13.2,-2.3) {};
\coordinate (TL) at (T.north west);
\coordinate (TR) at (T.north east);
\coordinate (BL) at (T.south west);
\coordinate (BR) at (T.south east);

\foreach \f in {0.10,0.20,0.35,0.50,0.65,0.78}{
  \draw[sepLine] ($(TL)!\f!(TR)$) -- ($(BL)!\f!(BR)$);
}
\draw[fill=MyGray, draw=MySlate!70, line width=0.6pt, rounded corners=2pt]
  (TL) rectangle ($(TL)!1!(TR)+(0,-1.35)$);

\node[headcell] at ($(TL)!0.05!(TR)+(0,-0.78)$) {ID};
\node[headcell] at ($(TL)!0.15!(TR)+(0,-0.78)$) {$t$};
\node[headcell] at ($(TL)!0.275!(TR)+(0,-0.78)$) {$\bm X^{\mathcal V}$};
\node[headcell] at ($(TL)!0.425!(TR)+(0,-0.78)$) {$\bm X^{\mathcal O}(t)$};
\node[headcell] at ($(TL)!0.575!(TR)+(0,-0.78)$) {$\bm X^{\mathcal Y}(t)$};
\node[headcell] at ($(TL)!0.715!(TR)+(0,-0.78)$) {$R_i^{\mathcal Y}(t)$};
\node[headcell] at ($(TL)!0.89!(TR) +(0,-0.78)$) {$Y_i(t)$};

\def\ryA{-1.95}\def\sAD{-0.86}\def\biggap{-1.00}\def\sEL{-0.90}
\pgfmathsetmacro{\ryB}{\ryA+1*\sAD}
\pgfmathsetmacro{\ryC}{\ryA+2*\sAD}
\pgfmathsetmacro{\ryD}{\ryA+3*\sAD}

\pgfmathsetmacro{\ryE}{\ryD+\biggap}
\pgfmathsetmacro{\ryF}{\ryE+1*\sEL}
\pgfmathsetmacro{\ryG}{\ryE+2*\sEL}

\pgfmathsetmacro{\ryH}{\ryE+3*\sEL}

\pgfmathsetmacro{\ryM}{\ryH+\biggap}

\foreach \yy in {\ryA,\ryB,\ryC,\ryD,\ryE,\ryF,\ryG,\ryH}{
  \draw[sepLine] ($(TL)+(0,\yy-0.45)$) -- ($(TR)+(0,\yy-0.45)$);
}

\node[covtag] at ($(TL)!0.05!(TR)+(0,\ryA)$) {1};
\node[covtag] at ($(TL)!0.15!(TR)+(0,\ryA)$) {$t_{11}$};
\node[covtag] at ($(TL)!0.275!(TR)+(0,\ryA)$) {$\bm X_1^{\mathcal V}$};
\node[covtag] at ($(TL)!0.425!(TR)+(0,\ryA)$) {$\bm X_1^{\mathcal O}(t_{11})$};
\node[covtag] at ($(TL)!0.575!(TR)+(0,\ryA)$) {$\bm X_1^{\mathcal Y}(t_{11})$};
\node[covtag] at ($(TL)!0.715!(TR)+(0,\ryA)$) {1};
\node[covtag] at ($(TL)!0.89!(TR) +(0,\ryA)$) {$Y_1(t_{11})$};

\node[covtag] at ($(TL)!0.05!(TR)+(0,\ryB)$) {1};
\node[covtag] at ($(TL)!0.15!(TR)+(0,\ryB)$) {$t_{12}$};
\node[covtag] at ($(TL)!0.275!(TR)+(0,\ryB)$) {$\bm X_1^{\mathcal V}$};
\node[covtag] at ($(TL)!0.425!(TR)+(0,\ryB)$) {$\bm X_1^{\mathcal O}(t_{12})$};
\node[covtag] at ($(TL)!0.575!(TR)+(0,\ryB)$) {$\bm X_1^{\mathcal Y}(t_{12})$};
\node[covtag] at ($(TL)!0.715!(TR)+(0,\ryB)$) {0};
\node[covtag, text=red] at ($(TL)!0.89!(TR) +(0,\ryB)$) {NA};

\node[covtag] at ($(TL)!0.05!(TR)+(0,\ryC)$) {1};
\node[covtag] at ($(TL)!0.15!(TR)+(0,\ryC)$) {$t_{13}$};
\node[covtag] at ($(TL)!0.275!(TR)+(0,\ryC)$) {$\bm X_1^{\mathcal V}$};
\node[covtag] at ($(TL)!0.425!(TR)+(0,\ryC)$) {$\bm X_1^{\mathcal O}(t_{13})$};
\node[covtag] at ($(TL)!0.575!(TR)+(0,\ryC)$) {$\bm X_1^{\mathcal Y}(t_{13})$};
\node[covtag] at ($(TL)!0.715!(TR)+(0,\ryC)$) {1};
\node[covtag] at ($(TL)!0.89!(TR) +(0,\ryC)$) {$Y_1(t_{13})$};

\node[covtag] at ($(TL)!0.05!(TR)+(0,\ryD)$) {1};
\node[covtag] at ($(TL)!0.15!(TR)+(0,\ryD)$) {$t_{14}$};
\node[covtag] at ($(TL)!0.275!(TR)+(0,\ryD)$) {$\bm X_1^{\mathcal V}$};
\node[covtag] at ($(TL)!0.425!(TR)+(0,\ryD)$) {$\bm X_1^{\mathcal O}(t_{14})$};
\node[covtag] at ($(TL)!0.575!(TR)+(0,\ryD)$) {$\bm X_1^{\mathcal Y}(t_{14})$};
\node[covtag] at ($(TL)!0.715!(TR)+(0,\ryD)$) {0};
\node[covtag, text=red] at ($(TL)!0.89!(TR) +(0,\ryD)$) {NA};

\node[covtag] at ($(TL)!0.05!(TR)+(0,\ryE)$) {2};
\node[covtag] at ($(TL)!0.15!(TR)+(0,\ryE)$) {$t_{21}$};
\node[covtag] at ($(TL)!0.275!(TR)+(0,\ryE)$) {$\bm X_2^{\mathcal V}$};
\node[covtag] at ($(TL)!0.425!(TR)+(0,\ryE)$) {$\bm X_2^{\mathcal O}(t_{21})$};
\node[covtag] at ($(TL)!0.575!(TR)+(0,\ryE)$) {$\bm X_2^{\mathcal Y}(t_{21})$};
\node[covtag] at ($(TL)!0.715!(TR)+(0,\ryE)$) {0};
\node[covtag, text=red] at ($(TL)!0.89!(TR) +(0,\ryE)$) {NA};

\node[covtag] at ($(TL)!0.05!(TR)+(0,\ryF)$) {2};
\node[covtag] at ($(TL)!0.15!(TR)+(0,\ryF)$) {$t_{22}$};
\node[covtag] at ($(TL)!0.275!(TR)+(0,\ryF)$) {$\bm X_2^{\mathcal V}$};
\node[covtag] at ($(TL)!0.425!(TR)+(0,\ryF)$) {$\bm X_2^{\mathcal O}(t_{22})$};
\node[covtag] at ($(TL)!0.575!(TR)+(0,\ryF)$) {$\bm X_2^{\mathcal Y}(t_{22})$};
\node[covtag] at ($(TL)!0.715!(TR)+(0,\ryF)$) {1};
\node[covtag] at ($(TL)!0.89!(TR) +(0,\ryF)$) {$Y_2(t_{22})$};

\node[covtag] at ($(TL)!0.05!(TR)+(0,\ryG)$) {2};
\node[covtag] at ($(TL)!0.15!(TR)+(0,\ryG)$) {$t_{23}$};
\node[covtag] at ($(TL)!0.275!(TR)+(0,\ryG)$) {$\bm X_2^{\mathcal V}$};
\node[covtag] at ($(TL)!0.425!(TR)+(0,\ryG)$) {$\bm X_2^{\mathcal O}(t_{23})$};
\node[covtag] at ($(TL)!0.575!(TR)+(0,\ryG)$) {$\bm X_2^{\mathcal Y}(t_{23})$};
\node[covtag] at ($(TL)!0.715!(TR)+(0,\ryG)$) {1};
\node[covtag] at ($(TL)!0.89!(TR) +(0,\ryG)$) {$Y_2(t_{23})$};

\node[covtag] at ($(TL)!0.05!(TR)+(0,\ryH)$) {3};
\node[covtag] at ($(TL)!0.15!(TR)+(0,\ryH)$) {$t_{31}$};
\node[covtag] at ($(TL)!0.275!(TR)+(0,\ryH)$) {$\bm X_3^{\mathcal V}$};
\node[covtag] at ($(TL)!0.425!(TR)+(0,\ryH)$) {$\bm X_3^{\mathcal O}(t_{31})$};
\node[covtag] at ($(TL)!0.575!(TR)+(0,\ryH)$) {$\bm X_3^{\mathcal Y}(t_{31})$};
\node[covtag] at ($(TL)!0.715!(TR)+(0,\ryH)$) {0};
\node[covtag, text=red] at ($(TL)!0.89!(TR) +(0,\ryH)$) {NA};

\node[covtag] at ($(TL)!0.05!(TR)+(0,\ryM)$) {$\vdots$};
\node[covtag] at ($(TL)!0.15!(TR)+(0,\ryM)$) {$\vdots$};
\node[covtag] at ($(TL)!0.275!(TR)+(0,\ryM)$) {$\vdots$};
\node[covtag] at ($(TL)!0.425!(TR)+(0,\ryM)$) {$\vdots$};
\node[covtag] at ($(TL)!0.575!(TR)+(0,\ryM)$) {$\vdots$};
\node[covtag] at ($(TL)!0.715!(TR)+(0,\ryM)$) {$\vdots$};
\node[covtag] at ($(TL)!0.89!(TR) +(0,\ryM)$) {$\vdots$};

\coordinate (RowID1) at ($(TL)!0.05!(TR)+(0-0.2,\ryA)$);
\coordinate (RowID2) at ($(TL)!0.05!(TR)+(0-0.2,\ryE)$);
\coordinate (RowID3) at ($(TL)!0.05!(TR)+(0-0.2,\ryH)$);

\draw[maparrow] (LeftAnchor1) .. controls +(2.0,0.9) and +(-1.0,0.9) .. (RowID1);
\draw[maparrow] (LeftAnchor2) .. controls +(2.0,0.5) and +(-1.0,0.5) .. (RowID2);
\draw[maparrow] (LeftAnchor3) .. controls +(2.0,-0.2) and +(-1.0,-0.2) .. (RowID3);

\end{tikzpicture}
\end{adjustbox}
\caption{Illustration of the hierarchical data generation process involving Informative Presence (IP) and Informative Observation (IO). 
Left panel: Patient timelines where clinic visits (ticks) are generated by the visiting process driven by covariates $\bm{X}_i^{\mathcal{V}}$. 
At each visit, the observation process (driven by $\bm{X}_i^{\mathcal{O}}(t)$) determines whether the biomarker outcome $Y_i(t)$ is measured (solid dots, $R_i^{\mathcal Y}(t)=1$) or unmeasured (hollow circles, $R_i^{\mathcal Y}(t)=0$), while the underlying longitudinal biomarker trajectory is driven by $\bm{X}_i^{\mathcal{Y}}(t)$. 
Right panel: The resulting long-format dataset used for analysis, where ``NA'' in the $Y_i(t)$ column indicates an unmeasured outcome ($R_i^{\mathcal Y}(t)=0$) despite the patient's presence at the clinic.}
\label{fig:IPIO}
\end{figure}

The statistical literature on longitudinal EHR analysis \citep[e.g.,][]{gasparini2020mixed} has focused predominantly on the IP mechanism, the first stage of the hierarchy, when estimating outcome biomarker trajectories. 
Such approaches include inverse intensity weighting \citep{robins2000marginal, lin2001semiparametric, burvzkova2007longitudinal, yiu2025accommodating}, pairwise likelihood \citep{chen2015regression, shen2019regression}, and joint modeling of the visiting and longitudinal processes through shared frailties \citep{liang2009joint, dai2018joint, weaver2023functional}. While these approaches have substantially advanced the analysis of irregularly observed longitudinal data, they share a common limitation: each implicitly assumes that a clinic visit guarantees measurement of the biomarker of interest---that is, they address IP while treating the observation stage as nonexistent. Although the Bayesian framework proposed by \cite{anthopolos2021modeling} considers both IP and IO, the approach remains computationally intractable for large-scale EHR data.
A recent contribution, EHRJoint \citep{du2025newstatisticalapproachjoint}, acknowledges the two-stage structure by specifying separate models for the visiting, observation, and longitudinal outcome processes. 
However, its observation model is formulated as an independent regression on observed covariates. This formulation fails to capture unmeasured heterogeneity---such as patient health awareness, access to care, or unrecorded clinical intuition---that jointly determines the timing of encounters, the necessity of testing, and the biomarker trajectories. This restriction is equivalent to treating the observation mechanism as missing at random (MAR) conditional on observed covariates, limiting its ability to accommodate the MNAR mechanism identified above.

To address these challenges, we make two main contributions. First, we propose GIVEHR (Gaussian Informative Visiting and observation processes in Electronic Health Records), a semiparametric joint modeling framework. Particularly well-suited to outpatient settings, it explicitly accommodates the two-stage hierarchy of EHR data collection.
Unlike existing methods, GIVEHR links visit intensity, observation mechanisms, and biomarker trajectories through a shared Gaussian latent variable. This structure captures unmeasured heterogeneity while yielding closed-form marginal likelihoods, enabling a computationally efficient estimation procedure that scales to large EHR databases.
In contrast to the Bayesian framework of \citet{anthopolos2021modeling}, which requires time discretization, GIVEHR operates in continuous time, preserving the irregular visit structure characteristic of EHR data. 
In effect, GIVEHR enables valid inference by accommodating the MNAR mechanisms inherent in both the visiting and observation processes---a challenge that existing methods leave unresolved.
We establish the consistency and asymptotic normality of the proposed estimators under regularity conditions.

Second, we provide the first systematic evaluation of 20 existing approaches under joint IP and IO mechanisms. These approaches can be broadly categorized into four classes: outcome-only methods, which ignore both IP and IO; IP-only adjustments, which correct for irregular visiting but assume the biomarker is measured at every visit; imputation-based pipelines, which impute missing biomarkers before applying an IP correction; and joint IP+IO frameworks, which model both stages simultaneously. Since no prior work has combined imputation with IP correction as described above, we introduce and evaluate this pipeline against the other approaches.
A key finding of this evaluation is that partial corrections can be counterproductive: adjusting for IP while ignoring IO can amplify bias relative to na\"ive analyses that ignore both stages entirely, because IP-only methods implicitly assume that every clinic visit yields a measurement, forcing the model to misattribute unmodeled observation bias to the visiting process itself---thereby amplifying, rather than removing, the underlying distortion.
Simulations demonstrate that, when the joint MNAR mechanism is correctly specified, GIVEHR yields approximately unbiased estimates, whereas existing methods exhibit substantial bias. Under departures from this specification, GIVEHR still considerably reduces bias relative to all competing approaches.
We further apply GIVEHR to the All of Us Research Program to investigate the association between two neighborhood-level socioeconomic status indicators and longitudinal trajectories of six clinical biomarkers. Because the visiting process is shared across biomarkers for a given patient, this design allows us to contrast varying observation mechanisms---ranging from routine to highly targeted tests---under a common IP structure.
This application demonstrates the scalability of GIVEHR to large-scale EHR data and illustrates how accounting for the MNAR mechanisms in both the visiting and observation processes can substantially change the magnitude of estimated socioeconomic–biomarker associations.

The remainder of this paper is organized as follows: Sections~\ref{sec:model} and~\ref{sec:estimation} detail the GIVEHR model specification and estimation procedure, respectively; Section~\ref{sec:asy_properties} establishes the asymptotic theory; Section~\ref{sec:supp:existing_methods} reviews existing methods; Section~\ref{sec:simulation} presents the simulation study; Section~\ref{sec:data_analysis} describes the application to the All of Us data; and Section~\ref{sec:discussion} concludes with a discussion.

\section{Method} \label{sec:model}

\subsection{Model Specification}

In a study comprising $n$ patients, let $m_i$ denote the total number of visits for subject $i$, and let $Y_i(t)$ represent the longitudinal biomarker process observed over the window $t \in [0, \tau]$. The follow-up time for each patient is restricted to this interval and is further limited by a random censoring time $C_i$.
The collection of biomarker measurements in the EHR system can be characterized by three processes: the visiting process, the observation process, and the longitudinal process. To account for the MNAR mechanisms that arise when all three processes are driven by unmeasured clinical heterogeneity, we posit a shared, subject-level latent variable $U_i \sim \mathcal{N}(0, 1)$ that captures unmeasured patient tendencies and induces dependence among all three processes. The Gaussian specification places the latent variable on the same scale as the random components in all three submodels, yielding transparent cross-process dependence, closed-form marginal likelihoods, and a scalable estimation procedure, as we detail below.

Each process may depend on distinct but possibly overlapping observed covariate sets. 
Let $\bm{X}_i^{\mathcal{V}}$ denote baseline covariates for the visiting process, 
$\bm{X}_i^{\mathcal{O}}(t)$ time-varying covariates for the observation process, 
and $\bm{X}_i^{\mathcal{Y}}(t)$ covariates for the longitudinal outcome. 
In the observation and outcome models, subvectors 
$\bm{Z}_i^{\mathcal{O}}(t) \subseteq \bm{X}_i^{\mathcal{O}}(t)$ and 
$\bm{Z}_i^{\mathcal{Y}}(t) \subseteq \bm{X}_i^{\mathcal{Y}}(t)$ 
are associated with subject-specific random coefficients linked to $U_i$. The visiting process depends solely on a scalar frailty, as visit intensity is primarily driven by aggregate health burden rather than covariate-specific random effects.

The visiting process captures the temporal pattern of patients' clinical visits. Let $N_i(t)$ be the visit counting process with jumps $\dint N_i(t) \in \{0, 1\}$ at visit times. 
Conditional on $U_i$ and $\bm{X}_i^{\mathcal V}$, we assume a multiplicative intensity model for the visiting process:
\begin{equation}\label{eq:visit_intensity}
\E\left\{\dint N_i(t) | U_i, \bm{X}_i^{\mathcal V} \right\} = \mathbb{I} \left( t \leq C_i \right) \eta_i \exp(\bm{\gamma}^\top \bm{X}_i^{\mathcal V}) \dint \Lambda_0(t)
\end{equation}
where $\dint \Lambda_0(t) = \lambda_0(t)\,\mathrm{d}t$ is an unspecified baseline intensity function, $\bm{\gamma}$ denotes the vector of regression coefficients, and $\mathbb{I} \left( t \leq C_i \right)$ is an indicator function, taking the value 1 if the patient is at risk at time $t$, and 0 otherwise.
The term $\eta_i$ represents a subject-specific multiplicative frailty that accounts for unmeasured heterogeneity in visit intensity. 
Unlike \citet{liang2009joint} and \citet{du2025newstatisticalapproachjoint}, who adopt a gamma frailty for its conjugacy with the Poisson-type visit likelihood, we specify 
\begin{align*}
\eta_i = \exp(\mu_0 + \sigma U_i)
\end{align*}
as lognormal. The lognormal specification ensures that the full vector of latent variables is jointly Gaussian, preserving closed-form conditional distributions throughout the estimation procedure. A gamma frailty would break this joint normality, requiring numerical integration over the full random effects vector whose dimension grows with the number of random coefficients in the observation and outcome processes. We constrain 
$\mu_0 = -\sigma^2/2$ so that $\E (\eta_i) = 1$, 
ensuring identifiability \citep{cook2007statistical}.

Additionally, the observation process determines whether the biomarker of interest, $Y_i(t)$, is observed at a given visit time $t$. Let $R_i^{\mathcal Y}(t) \in \{0, 1\}$ be the observation indicator, conditional on a visit occurring (i.e., $\dint N_i(t) = 1$). We formulate the probability of observation using a probit mixed effects model, as the normal cumulative distribution function (CDF) composed with normally distributed random effects admits a closed-form marginal probability---a property not shared by the logit link---thereby maintaining the analytical tractability of the Gaussian latent structure:
\begin{equation}\label{eq:observation_probability}
\pr \left\{ R_i^{\mathcal Y}(t) = 1 \mid \dint N_i(t) = 1, U_i, \bm{X}_i^{\mathcal O}(t), \bm{Z}_i^{\mathcal O}(t) \right\} = \Phi \left\{ \bm{\alpha}^\top \bm{X}_i^{\mathcal O}(t) + \bm{q}_i^\top \bm{Z}_i^{\mathcal O}(t) \right\}
\end{equation}
where $\Phi(\cdot)$ denotes the standard normal CDF. 
The vector $\bm{q}_i$ represents subject-specific random coefficients that capture unmeasured factors driving the observation decision beyond what is explained by observed covariates.
To capture the dependency on the shared latent structure, 
we assume 
\begin{align*}
\bm{q}_i | U_i \sim \mathcal{N}(\bm{\delta} U_i, \bm{\Sigma}_q)   
\end{align*}
where $\bm{\delta}$ governs the dependence between the shared latent variable and the propensity to observe. When $\bm{\delta} \neq \bm{0}$, the observation process depends on the latent variable that also drives the biomarker trajectory, inducing the MNAR mechanism discussed above; setting $\bm{\delta} = \bm{0}$ reduces the observation model to MAR. For brevity, we define $\omega_i(t;\bm{\alpha},\bm{\delta},\bm{\Sigma}_q, U_i) = \pr \left\{ R_i^{\mathcal Y}(t) = 1 \mid \dint N_i(t) = 1, U_i, \bm{X}_i^{\mathcal O}(t), \bm{Z}_i^{\mathcal O}(t) \right\}$.

The longitudinal process characterizes the biomarker trajectory over time. The biomarker value, $Y_i(t)$, is observed only when a visit occurs and the biomarker is measured (i.e., $\dint N_i(t)=1$ and $R_i^{\mathcal Y}(t)=1$). We assume a semiparametric model with random effects for the trajectory:
\begin{equation}\label{eq:continuous_outcome}
Y_i(t) = \beta_0(t) + \bm{\beta}^\top \bm{X}_i^{\mathcal Y}(t) + \bm{b}_i^\top \bm{Z}_i^{\mathcal Y}(t) + \varepsilon_i(t)
\end{equation}
where $\bm{\beta}$ denotes the fixed-effect coefficients corresponding to $\bm{X}_i^{\mathcal Y}(t)$ and $\varepsilon_i(t) \sim \mathcal{N}(0, \sigma^2_\varepsilon)$ represents the random error. The baseline mean function $\beta_0(t)$ is left unspecified, avoiding the need to select a parametric form for the temporal trend and providing robustness against misspecification of the time effect.
The random effects $\bm{b}_i$ are linked to the 
shared factor $U_i$ via 
\begin{align*}
\bm{b}_i \mid U_i \sim \mathcal{N}(\bm{\theta} U_i, \bm{\Sigma}_b),
\end{align*}
where $\bm{\theta}$ plays an analogous role to $\bm{\delta}$
in the observation process: it governs how the shared latent variable shifts the biomarker trajectory. For instance, when $\sigma >0 $ and $\bm{\theta} > \bm{0}$, patients who visit more frequently also tend to have higher biomarker values, capturing the positive confounding structure commonly encountered in clinical data. On the other hand, setting $\bm{\theta} = \bm{0}$ removes the dependence between the longitudinal process and the latent variable, reducing the outcome model to a standard linear mixed model.

The shared latent variable $U_i$ induces a coherent dependence structure across the three submodels. Define $\zeta_i = \mu_0 + \sigma U_i$ and rewrite $\bm b_i = \bm\theta U_i + \bm{\widetilde{b}}_i$ and $\bm q_i = \bm\delta U_i + \bm{\widetilde{q}}_i$, where $(\bm{\widetilde{b}}_i, \bm{\widetilde{q}}_i)$ are independent of $U_i$. Then $(\zeta_i, \bm b_i^\top, \bm q_i^\top)^\top$ is multivariate normal with covariance matrix
\[
\Var\!\begin{pmatrix}
\zeta_i\\ \bm b_i\\ \bm q_i
\end{pmatrix}=
\begin{pmatrix}
\sigma^{2} & \sigma\,\bm\theta^\top & \sigma\,\bm\delta^\top\\[2pt]
\sigma\,\bm\theta & \bm\theta\bm\theta^\top + \bm\Sigma_b & \bm\theta\,\bm\delta^\top\\[2pt]
\sigma\,\bm\delta & \bm\delta\,\bm\theta^\top & \bm\delta\bm\delta^\top + \bm\Sigma_q
\end{pmatrix}.
\]
The off-diagonal blocks reveal the cross-process dependence implied by the model. For example, $\sigma \bm{\theta}^\top$ quantifies the covariance between visit intensity and the biomarker trajectory, while $\bm{\theta} \bm{\delta}^\top$ captures the association between the biomarker trajectory and the propensity to be observed. This explicit, interpretable dependence structure is the central advantage of the Gaussian latent specification---it is precisely this structure that allows GIVEHR to correct for the joint MNAR mechanism induced by IP and IO.

The parameters to be estimated are organized according to the three submodels. For the visiting process, these include the regression coefficients $\bm{\gamma}$, the frailty parameter $\sigma$, and the non-parametric baseline intensity $\Lambda_0(t)$. The observation process involves the fixed effects $\bm{\alpha}$, the dependence vector $\bm{\delta}$, and the covariance matrix $\bm{\Sigma}_q$. Finally, the longitudinal process encompasses the fixed effects $\bm{\beta}$, the dependence vector $\bm{\theta}$, and the covariance matrix $\bm{\Sigma}_b$. The comprehensive estimation procedure for these components is detailed in Section~\ref{sec:estimation}.

\subsection{Assumptions}

These assumptions serve distinct functions within the GIVEHR framework, bridging theoretical rigor with computational efficiency. Assumptions~\ref{asmp:censoring} and~\ref{asmp:cond_indep} establish the independence structures necessary for unbiased parameter estimation, while Assumption~\ref{asmp:distributions} specifies the parametric forms that enable closed-form representations of the empirical Bayes posterior and marginalized observation probabilities.

\begin{assumption}[Noninformative censoring]\label{asmp:censoring}
The censoring time $C_i$ is assumed to be noninformative. Specifically, conditional on the covariates $\{\bm{X}_i^{\mathcal V}, \bm{X}_i^{\mathcal O}(t), \bm{X}_i^{\mathcal Y}(t)\}$, $C_i$ is independent of the visiting process $N_i(\cdot)$, the observation process $R_i^{\mathcal Y}(\cdot)$, and the longitudinal outcome $Y_i(\cdot)$. This assumption ensures that the censoring mechanism does not introduce additional bias in model estimation.  
\end{assumption}

\begin{assumption}[Conditional independence of data collection and outcome errors]\label{asmp:cond_indep}
Given the shared latent variable $U_i$ and all observed covariates, the visiting, observation, and longitudinal outcome processes are conditionally independent:
\[
\left\{N_i(t), R_i^{\mathcal Y}(t) \right\} \perp \{\varepsilon_i(t)\} 
\; \big|\; \left\{ U_i, \bm{X}_i^{\mathcal V}, \bm{X}_i^{\mathcal O}(t), \bm{X}_i^{\mathcal Y}(t) \right\}.
\]  
This condition permits the factorization of the joint likelihood, justifying the separability of the estimation procedure.
\end{assumption}

\begin{assumption}[latent variable and random-effect distributions]\label{asmp:distributions}
Conditional on \(U_i\), the random slopes in the observation submodel follow a multivariate normal distribution,
\[
\bm{q}_i \mid U_i \sim \mathcal{N}(\bm{\delta} U_i, \bm{\Sigma}_q).
\]
In the longitudinal submodel, the random effects satisfy the conditional mean link
\[
\bm{b}_i \mid U_i \sim \mathcal{N}(\bm{\theta} U_i, \bm{\Sigma}_b).
\]
This Gaussian specification, combined with the probit link in~\eqref{eq:observation_probability}, enables a closed-form marginalization of the latent variable and avoids the need for high-dimensional numerical integration.
\end{assumption}

\section{Three-Stage Estimation}\label{sec:estimation}

\subsection{Overview}

We propose a three-stage estimation procedure summarized in Figure~\ref{fig:estimation_roadmap}. In the first stage, we estimate nuisance parameters governing the visiting process and derive an empirical Bayes posterior for the latent variable $U_i$. In the second stage, we estimate the observation process parameters and obtain the marginalized observation probability $\overline{\omega}_i(t)$. In the third stage, these quantities are used to construct a compensation covariate that corrects for bias in the weighted least squares (WLS) estimation of the target longitudinal outcome parameters $(\bm{\beta}, \bm{\theta})$. 

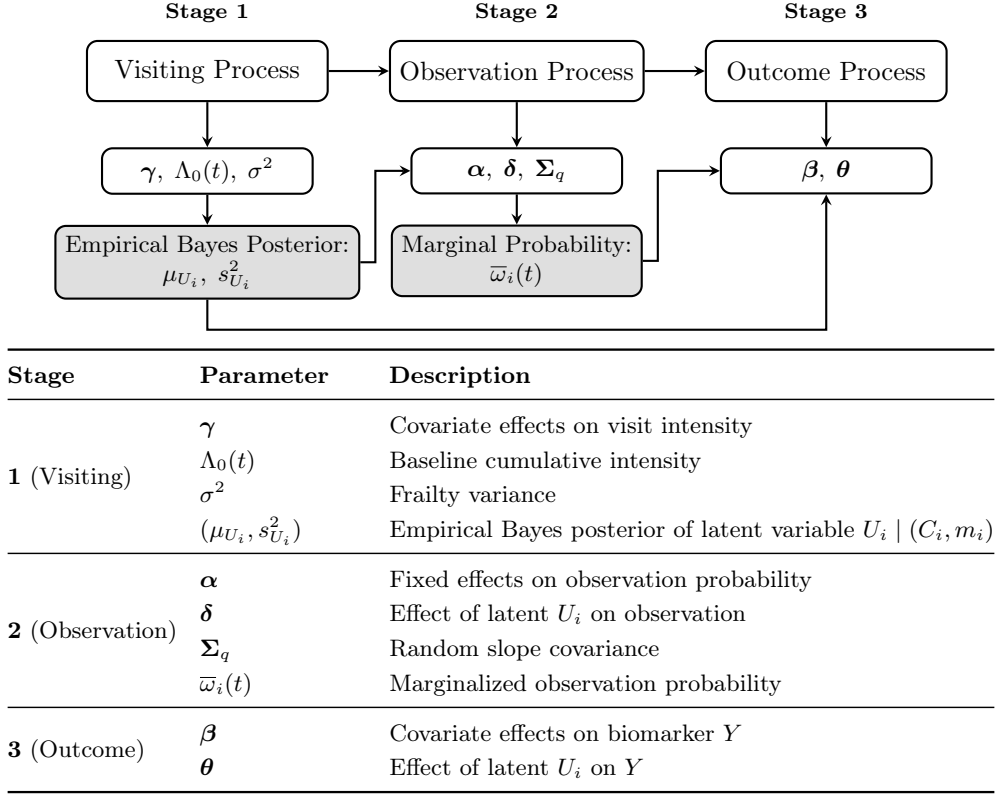
\begin{figure}[t]
\centering
\captionsetup{width=\textwidth} 

\begin{tikzpicture}[
    node distance=0.5cm and 0.8cm,
    stage/.style={rectangle, rounded corners, draw=black, 
                  thick, minimum height=0.8cm, minimum width=3.2cm, align=center, font=\small},
    param/.style={rectangle, rounded corners, draw=black, fill=white,
                    thick, minimum height=0.6cm, minimum width=2.8cm, align=center, font=\footnotesize},
    latent/.style={rectangle, rounded corners, draw=black, fill=gray!25,
                    thick, minimum height=0.6cm, minimum width=2.8cm, align=center, font=\footnotesize},
    arrow/.style={->, thick, >=stealth, draw=black}
]
\node[stage] (visit) {Visiting Process};
\node[stage, right=of visit] (obs) {Observation Process};
\node[stage, right=of obs] (outcome) {Outcome Process};
\node[above=0.1cm of visit, font=\scriptsize\bfseries] {Stage 1};
\node[above=0.1cm of obs, font=\scriptsize\bfseries] {Stage 2};
\node[above=0.1cm of outcome, font=\scriptsize\bfseries] {Stage 3};

\node[param, below=0.6cm of visit] (vp) {$\bm{\gamma},\; \Lambda_0(t),\; \sigma^2$};
\node[param, below=0.6cm of obs] (op) {$\bm{\alpha},\; \bm{\delta},\; \bm{\Sigma}_q$};
\node[param, below=0.6cm of outcome] (outp) {$\bm{\beta},\; \bm{\theta}$};

\node[latent, below=0.4cm of vp] (eb) {Empirical Bayes Posterior:\\ $\mu_{U_i},\; s_{U_i}^2$};
\node[latent, below=0.4cm of op] (omega) {Marginal Probability:\\ $\overline{\omega}_i(t)$};

\draw[arrow, black] (visit) -- (vp);
\draw[arrow, black] (obs) -- (op);
\draw[arrow, black] (outcome) -- (outp);
\draw[arrow, black] (vp) -- (eb);
\draw[arrow, black] (op) -- (omega);

\draw[arrow] (visit) -- (obs);
\draw[arrow] (obs) -- (outcome);

\draw[arrow] (eb.east) -- ++(0.2,0) |- (op.west);
\draw[arrow] (omega.east) -- ++(0.2,0) |- (outp.west);
\draw[arrow] (eb.south) -- ++(0,-0.4) -| (outp.south);
\end{tikzpicture}

\vspace{0.2cm}

\footnotesize
\renewcommand{\arraystretch}{1.2}
\begin{tabular}{@{}l@{\hspace{0.3cm}}l@{\hspace{0.8cm}}l@{}}
\toprule
\textbf{Stage} & \textbf{Parameter} & \textbf{Description} \\
\midrule
\multirow{4}{*}{\textbf{1} (Visiting)} 
& $\bm{\gamma}$ & Covariate effects on visit intensity \\
& $\Lambda_0(t)$ & Baseline cumulative intensity \\
& $\sigma^2$ & Frailty variance \\
& $(\mu_{U_i}, s_{U_i}^2)$ & Empirical Bayes posterior of latent variable $U_i \mid (C_i, m_i)$ \\
\midrule
\multirow{4}{*}{\textbf{2} (Observation)} 
& $\bm{\alpha}$ & Fixed effects on observation probability \\
& $\bm{\delta}$ & Effect of latent $U_i$ on observation \\
& $\bm{\Sigma}_q$ & Random slope covariance \\
& $\overline{\omega}_i(t)$ & Marginalized observation probability \\
\midrule
\multirow{2}{*}{\textbf{3} (Outcome)} 
& $\bm{\beta}$ & Covariate effects on biomarker $Y$ \\
& $\bm{\theta}$ & Effect of latent $U_i$ on $Y$ \\
\bottomrule
\end{tabular}
\caption{Three-stage estimation procedure. Gray shaded boxes indicate latent information transmitted across stages. Stages 1 and 2 estimate nuisance parameters to estimate the empirical Bayes posterior and marginalized observation probabilities, respectively; these quantities are subsequently incorporated into Stage 3 to correct for clinically informed bias.}
\label{fig:estimation_roadmap}
\end{figure}


\subsection{Stage 1: Visiting Process}

We estimate $\bm{\gamma}$ by maximizing the partial likelihood 
for the multiplicative intensity model \eqref{eq:visit_intensity}.
The corresponding score equation takes the form
\begin{equation} \label{eq:EE_gamma}
\sum_{i=1}^{n}\int_{0}^{\tau}\left\{\bm{X}_i^{\mathcal V} -\overline{\bm{X}}^{\mathcal V}(t) \right\}\,\dint N_{i}(t)=0,
\end{equation}
where $\overline{\bm{X}}^{\mathcal V}(t) = \sum_{j=1}^{n} \mathbb{I}(t \leq C_j) \exp\left(\bm{\gamma}^\top \bm{X}_j^{\mathcal V}\right) \bm{X}_j^{\mathcal V} \big/ \sum_{j=1}^{n} \mathbb{I}(t \leq C_j) \exp \left(\bm{\gamma}^\top \bm{X}_j^{\mathcal V}\right)$ 
is the risk-set average at time $t$.
This is the 
Andersen--Gill extension \citep{andersen1982cox} of 
the Cox partial likelihood score equation to recurrent 
events, treating each visit as a recurrent event within 
a counting process framework.
Given $\widehat{\bm{\gamma}}$, the baseline cumulative intensity is estimated via the Breslow estimator \citep{breslow1974covariance}:
\begin{equation} \label{eq:AalenBreslow}
\widehat{\Lambda}_{0}(t)=\sum_{i=1}^{n}\int_{0}^{t}\frac{\dint N_{i}(s)}{\sum_{j=1}^{n} \mathbb{I}(s \leq C_j) \exp(\widehat{\bm{\gamma}}^\top \bm{X}_j^{\mathcal V})}.
\end{equation}
Define the estimated expected cumulative intensity by
\[
\widehat{\nu}_i
:= \exp\!\big(\widehat{\bm{\gamma}}^\top \bm{X}_i^{\mathcal V}\big)\,\widehat{\Lambda}_0(C_i).
\]
Let $\sigma_\eta^2 := \Var(\eta_i)$ denote the common frailty variance.
This variance is estimated using the method of moments based on the observed visit counts $m_i$ for patient $i$:
\begin{equation} \label{eq:MoM_VarEta}
\widehat{\sigma}^2_\eta = \max\left\{0, \frac{\sum_{i=1}^n \left( m_i^2 - m_i - \widehat{\nu}_i^2 \right) }{\sum_{i=1}^n \widehat{\nu}_i^2}\right\}.
\end{equation}
Assuming $\log \eta_i \sim \mathcal{N}(\mu_0, \sigma^2)$ subject to $\E (\eta_i) =1$, we obtain $\widehat{\sigma}^2=\log\left\{1+\widehat{\Var}(\eta_i)\right\}$ and
$\widehat{\mu}_0=-\widehat{\sigma}^2/2$.
Subsequently, we approximate the empirical Bayes posterior by a normal distribution via a Laplace approximation
\citep{tierney1986accurate} (Lemma~\ref{lem:laplace}),
\[
U_i \mid (C_i,m_i) \approx \mathcal{N}(\mu_{U_i}, s_{U_i}^2),
\]
where $\mu_{U_i}$ is the posterior mode and $s_{U_i}^2$ is the negative inverse Hessian of the log-posterior of $U_i$ evaluated at $\mu_{U_i}$.
Lemma~\ref{lem:laplace} further demonstrates that the posterior is strictly log-concave in $U_i$, which guarantees a unique posterior mode
$\mu_{U_i}$ and a uniquely determined $s_{U_i}^2$.

\subsection{Stage 2: Observation Process}

To construct the likelihood for the observation indicators $R_i^{\mathcal Y}(t)$, we require the observation probability averaged over uncertainty in $U_i$. Using the empirical Bayes posterior from Stage~1,
we define the marginalized observation probability
\[
\overline{\omega}_i(t) := \E\!\left\{\omega_i(t; U_i) \mid C_i, m_i\right\}.
\]
By probit--normal conjugacy (Lemma~\ref{lem:probit}), this admits a closed-form representation:
\[
\overline{\omega}_i(t) = \Phi\! \left\{ k_i(t) \right\},
\]
where
\[
k_i(t) := \frac{\bm{\alpha}^\top \bm{X}_i^{\mathcal O}(t) + \left\{ \bm{\delta}^\top \bm{Z}_i^{\mathcal O}(t) \right\} \mu_{U_i}}
{\sqrt{1 + \left\{ \bm{Z}_i^{\mathcal O}(t) \right\}^\top \bm{\Sigma}_q \bm{Z}_i^{\mathcal O}(t) 
+ \left\{\bm{\delta}^\top \bm{Z}_i^{\mathcal O}(t)\right\}^2 s_{U_i}^2}}.
\]
With the marginalized observation probability $\overline{\omega}_i(t)$ defined for each visit, the parameters $(\bm{\alpha},\bm{\delta},\bm{\Sigma}_q)$ are estimated by maximizing the composite Bernoulli log-likelihood function, summed over all visit times in $[0, \tau]$ for all subjects:
\begin{equation}\label{eq:Likelihood_IO}
\ell_{\text{IO}}(\bm{\alpha},\bm{\delta},\bm{\Sigma}_q)
= \sum_{i=1}^n \int_0^{\tau}
\Bigl[
R_i^{\mathcal Y}(t)\log\!\bigl\{\overline{\omega}_i(t)\bigr\}
+ \bigl\{1-R_i^{\mathcal Y}(t)\bigr\}\log\!\bigl\{1-\overline{\omega}_i(t)\bigr\}
\Bigr]\, \dint N_i(t).
\end{equation}
The resulting estimators $(\widehat{\bm{\alpha}},\widehat{\bm{\delta}},\widehat{\bm{\Sigma}}_q)$ are the values that maximize this function, subject to $\widehat{\bm{\Sigma}}_q$ being positive semi-definite. For notational brevity, we denote the plug-in estimate as $\widehat{\overline{\omega}}_i(t)
:= \overline{\omega}_i\!\left(t;\widehat{\bm{\alpha}},\widehat{\bm{\delta}},\widehat{\bm{\Sigma}}_q\right)$.

\subsection{Stage 3: Longitudinal Outcome Model}

Conditional on the follow-up time $C_i$ and the number of visits $m_i$, estimation of the outcome parameters relies on a mean-zero process $M_i(t)$. Following \cite{liang2009joint}, the conditional expectation of the outcome increment is corrected for IP and IO mechanisms driven by $U_i$.

We define the observation-weighted posterior mean ratio as
\begin{align*}
\kappa_i(t) &:= \frac{\E\!\left\{U_i \, \omega_i(t; U_i) \mid C_i, m_i\right\}}{\overline{\omega}_i(t)} \\
&= \mu_{U_i} + \frac{ \left\{ \bm{\delta}^\top \bm{Z}_i^{\mathcal O}(t) \right\} s_{U_i}^2}
{\sqrt{1 + \left\{\bm{Z}_i^{\mathcal O}(t) \right\}^\top \bm{\Sigma}_q \bm{Z}_i^{\mathcal O}(t) 
+ \left\{ \bm{\delta}^\top \bm{Z}_i^{\mathcal O}(t)\right\}^2 s_{U_i}^2}}
\, \frac{\phi\{k_i(t)\}}{\Phi\{k_i(t)\}},
\end{align*}
where $\phi(\cdot)$ denotes the probability density function of the standard normal distribution. As demonstrated in Lemma~\ref{lem:probit}, $\kappa_i(t)$ admits a closed-form representation, which simplifies the subsequent computation of the marginal likelihood.
A compensation covariate is then constructed as $\bm{B}_i(t) := \bm{Z}_i^{\mathcal Y}(t) \, \kappa_i(t)$.

Define the weight function for subject $i$ at time $t$ as
\begin{equation}\label{eq:weight}
p_i(t) = \overline{\omega}_i(t)\, \mathbb{I} (t \leq C_i) \,\frac{m_i}{\Lambda_0(C_i)},
\end{equation}
which corrects for two distinct sources of clinically informative sampling.  The marginalized observation probability $\overline{\omega}_i(t)$ adjusts for the informative observation process, whereas the factor $\mathbb{I}(t \leq C_i) m_i/\Lambda_0(C_i)$ adjusts for the informative visiting process. 
This latter term arises from Lemma~\ref{lem:order}, which establishes that, conditional on the follow-up window $C_i$ and the total number of visits $m_i$, each visit is equally likely to occur at any point along the baseline intensity scale, i.e., $\E\{\dint N_i(t)\mid C_i,m_i\} = \mathbb{I}( t \leq C_i ) m_i\,\dint \Lambda_0(t)/\Lambda_0(C_i)$.

The estimating equations for the cumulative baseline function $\mathcal{A}(t) = \int_0^t \beta_0(s)\,\dint \Lambda_0(s)$ and the regression parameters $(\bm{\beta}, \bm{\theta})$ are derived from the compensated process:
\begin{equation}\label{eq:mean_zero_process}
M_i(t)
=\int_0^t \big\{Y_i(s)-\bm{\beta}^\top \bm{X}_i^{\mathcal Y}(s)-\bm{\theta}^\top \bm{B}_i(s)\big\} R_i^{\mathcal Y}(s) \,\dint N_i(s)
-\int_0^t p_i(s) \,\dint \mathcal{A}(s).
\end{equation}
By Lemma~\ref{lem:order} and the martingale compensation result in Lemma~\ref{lem:compensation}, $M_i(t)$ is a zero-mean martingale at the true parameter values.

We solve these equations using a profiling approach. Summing the martingale representation over subjects yields a Breslow-type estimator for the baseline increment:
\[
\dint\widehat{\mathcal{A}}(t) = \frac{\sum_{i=1}^n \big\{Y_i(t)-\bm{\beta}^\top \bm{X}_i^{\mathcal Y}(t)-\bm{\theta}^\top\widehat{\bm{B}}_i(t)\big\}\,R_i^{\mathcal Y}(t)\,\dint N_i(t)}
{\sum_{i=1}^n \widehat{p}_i(t)},
\]
where $\widehat{p}_i(t)$ is the plug-in estimate of~\eqref{eq:weight}. Substituting this back into the martingale representation and rearranging yields a risk-set centered estimating equation for $(\bm{\beta}, \bm{\theta})$:
\begin{equation}
\label{eq:EE_outcome}
\sum_{i=1}^n \int_0^\tau
\begin{pmatrix} \bm{X}_i^{\mathcal Y}(t)-\overline{\bm{X}}^{\mathcal Y}(t)\\ \widehat{\bm{B}}_i(t)- \widehat{\overline{\bm{B}}}(t) \end{pmatrix}
\big\{Y_i(t)-\bm{\beta}^\top \bm{X}_i^{\mathcal Y}(t)-\bm{\theta}^\top \widehat{\bm{B}}_i(t)\big\}\,R_i^{\mathcal Y}(t)\,\dint N_i(t)=\bm{0},
\end{equation}
where the centering terms are weighted averages over all subjects at each time $t$:
\[
\overline{\bm{X}}^{\mathcal Y}(t) = \frac{\sum_{i=1}^n \bm{X}_i^{\mathcal Y}(t) \, \widehat{p}_i(t)}{\sum_{i=1}^n \widehat{p}_i(t)}, \quad 
\widehat{\overline{\bm{B}}}(t) = \frac{\sum_{i=1}^n \widehat{\bm{B}}_i(t) \, \widehat{p}_i(t)}{\sum_{i=1}^n \widehat{p}_i(t)}.
\]
Solving this equation yields the estimators $(\widehat{\bm{\beta}}, \widehat{\bm{\theta}})$.

The complete three-stage procedure is summarized in Algorithm~\ref{alg:estimation}.

\begin{algorithm}[htbp]
\caption{Three-stage estimation procedure for GIVEHR.}
\label{alg:estimation}
\begin{algorithmic}[1] 

\State \textbf{Step 1. Visiting process estimation.}
\State \quad (a) Estimate $(\bm{\gamma}, \Lambda_0(t), \sigma^2)$ using \eqref{eq:EE_gamma}--\eqref{eq:MoM_VarEta}.
\State \quad (b) Compute the empirical Bayes posterior $(\widehat{\mu}_{U_i}, \widehat{s}_{U_i}^2)$
via Laplace approximation (Lemma~\ref{lem:laplace}).

\Statex

\State \textbf{Step 2. Observation process estimation.}
\State \quad Estimate $(\bm{\alpha}, \bm{\delta}, \bm{\Sigma}_q)$ by maximizing \eqref{eq:Likelihood_IO},
treating $(\widehat{\mu}_{U_i}, \widehat{s}_{U_i}^2)$ from Step~1 as known.

\Statex

\State \textbf{Step 3. Outcome process estimation.}
\State \quad (a) Construct the compensation covariate
$\widehat{\bm{B}}_i(t)=\bm{Z}_i^{\mathcal{Y}}(t)\,\widehat{\kappa}_i(t)$ and weights $\widehat{p}_i(t)$
using estimates from Steps~1--2.
\State \quad (b) Obtain $(\widehat{\bm{\beta}}, \widehat{\bm{\theta}})$ from the closed-form weighted least squares
solution \eqref{eq:EE_outcome}.

\end{algorithmic}
\end{algorithm}

\begin{remark}
The proposed framework can further accommodate dependence on the subject's own observation past history. For example, an abnormal biomarker value may prompt repeat testing. However, including history $R_i^{\mathcal{Y}}(t^-)$ in the observation model introduces endogeneity, since past observation indicators involve the same patient-specific latent variable that drives current observation. 
Moreover, conditioning on the prior history breaks the conditional independence across observation history, substantially increasing computational cost.
A modified estimation algorithm and additional simulation addressing this issue are provided in Sections~\ref{sec:tv_obs} and~\ref{sec:tv_obs_sim} of the Supplementary Material, respectively. 
\end{remark}

\section{Asymptotic Properties of \texorpdfstring{$(\widehat{\bm\beta},\widehat{\bm\theta})$}{(beta, theta)}} \label{sec:asy_properties}

We establish the large-sample behavior of the joint estimator $(\widehat{\bm\beta},\widehat{\bm\theta})$, defined as the solution to \eqref{eq:EE_outcome}.

\begin{theorem}[Consistency]\label{thm:consistency}
Under Assumptions~\ref{asmp:censoring}--\ref{asmp:distributions} and regularity conditions \textup{(C1)}--\textup{(C6)}, let $(\widehat{\bm\beta},\widehat{\bm\theta})$ be the solution to the estimating equations \eqref{eq:EE_outcome}, where the conditional expectations involving $U_i$ are evaluated using the Laplace approximation. Then, as $n \to \infty$,

\[
\begin{pmatrix}\widehat{\bm\beta}\\ \widehat{\bm\theta}\end{pmatrix}
\xrightarrow{p}
\begin{pmatrix}\bm\beta_0\\ \bm\theta_0\end{pmatrix}.
\]
\end{theorem}

Consistency holds provided that the approximation error introduced by the Laplace method is asymptotically negligible under standard regularity conditions. The proof is detailed in Supplementary Section~\ref{sec:proof-consistency}.

The asymptotic distribution of the estimator follows from standard martingale central limit theory, extending the results of \citet{liang2009joint} and \citet{du2025newstatisticalapproachjoint}.

\begin{theorem}[Asymptotic normality]\label{thm:normality}
Under Assumptions~\ref{asmp:censoring}--\ref{asmp:distributions} and regularity conditions \textup{(C1)}--\textup{(C6)}, if $\sigma^2 > 0$, then as $n \to \infty$,
\[
\sqrt{n} \left\{
\begin{pmatrix}
\widehat{\bm{\beta}} \\[2pt]
\widehat{\bm{\theta}}
\end{pmatrix}
-
\begin{pmatrix}
\bm \beta_0 \\[2pt]
\bm \theta_0
\end{pmatrix}
\right\}
\xrightarrow{d}\;
\mathcal{N}\!\left\{\mathbf{0},\, \bm{W}^{-1} \bm{\Gamma} \left(\bm{W}^{-1} \right) \right\},
\]
where the matrices $\bm W$ and $\bm \Gamma$ are defined in Supplementary Section \ref{sec:proof-normality}.
\end{theorem}

The asymptotic variance $\bm{W}^{-1} \bm{\Gamma} \bm{W}^{-1}$ can be estimated via the plug-in method using the influence function decomposition detailed in Section~\ref{sec:proof-normality} of the Supplementary Material. The nonparametric bootstrap offers an alternative that avoids analytic derivations \citep{liang2009joint, du2025newstatisticalapproachjoint}, though at the cost of additional computational burden from repeated estimation across resampled datasets.

\section{Existing Methods}\label{sec:supp:existing_methods}

We classify existing approaches for analyzing longitudinal EHR data into four broad categories based on the assumptions they make regarding the data generation mechanisms: (i) outcome-only methods that ignore the IP and IO mechanisms; (ii) IP-only methods that explicitly model the visiting process but assume complete observation at visits; (iii) imputation+IP methods that first impute missing outcomes conditional on the visiting process and then apply an IP-only model; and (iv) IP- and IO-aware methods (IP+IO) that address both informative visiting and observation. Table~\ref{tab:methods_summary} provides a summary of these approaches.

\begin{table}[h]
\centering
\caption{Comprehensive comparison of statistical methods for longitudinal EHR data. Methods are differentiated by whether they account for informative visiting (IP) and informative observation (IO), the specific scope of the shared latent structure linking these processes to the outcome, and the adjustment mechanism applied. 
Symbols: \ding{51} = Explicitly handled/modeled; \ding{55} = Ignored or not handled. 
Abbreviations: LMM, linear mixed model; VA/OA, visit/observation-adjusted; JMVL, joint modeling of visiting and longitudinal data; MAR, missing at random; MI, multiple imputation; LI, linear increment.}
\label{tab:methods_summary}
\begingroup\footnotesize
\setlength{\tabcolsep}{4pt}
\renewcommand{\arraystretch}{1.5} 
\begin{tabularx}{\linewidth}{@{}l c c l X@{}}
\toprule
\textbf{Method} & \textbf{IP} & \textbf{IO} & \textbf{Shared Latent Pattern} & \textbf{Adjustment Mechanism} \\
\midrule
\multicolumn{5}{l}{\textit{Outcome-only}}\\
\makecell[l]{Summary statistics} &
\ding{55} & \ding{55} &
None &
Reduces longitudinal series to a single scalar (e.g., min, mean, median, max) for regression. \\
\makecell[l]{Standard LMM\\ \scriptsize\citep{laird1982random}} &
\ding{55} & \ding{55} &
None &
Conditions on observed covariates (valid only under MAR). \\
\makecell[l]{VA LMM\\ \scriptsize\citep{goldstein2016controlling}} &
\ding{51} & \ding{55} &
None &
Includes \textit{visit count} as a fixed covariate proxy. \\
\makecell[l]{OA LMM} &
\ding{55} & \ding{51} &
None &
Includes \textit{observation count} as a fixed covariate proxy. \\
\midrule
\multicolumn{5}{l}{\textit{IP-only}}\\
\makecell[l]{Liang\\ \scriptsize\citep{liang2009joint}} &
\ding{51} & \ding{55} &
Visiting \& Outcome &
Joint likelihood estimation via shared frailty linking visiting and outcome. \\
\makecell[l]{JMVL--G\\ \scriptsize\citep{gasparini2020mixed}} &
\ding{51} & \ding{55} &
Visiting \& Outcome &
Joint likelihood modeling inter-visit times via shared random effects. \\
\makecell[l]{IIRR Weighting\\ \scriptsize\citep{burvzkova2007longitudinal}} &
\ding{51} & \ding{55} &
None &
Inverse intensity weighting based on visiting intensity. \\
\makecell[l]{IIRR Balanced\\ \scriptsize\citep{yiu2025accommodating}} &
\ding{51} & \ding{55} &
None &
Inverse intensity weighting adjusted for covariate balance. \\
\makecell[l]{Pairwise Likelihood\\ \scriptsize\citep{chen2015regression}} &
\ding{51} & \ding{55} &
None &
Composite likelihood exploiting visit time ordering. \\
\midrule
\multicolumn{5}{l}{\textit{Imputation + IP-only}}\\
\makecell[l]{Multiple Imputation\\ \scriptsize\citep{Rubin1987MI}} &
\ding{51} & \ding{51} &
Determined by IP-only &
Multiple imputation based on posterior draws, followed by IP-only modeling. \\
\makecell[l]{Linear Increment Methods\\ \scriptsize\citep{aalen2010dynamic}} &
\ding{51} & \ding{51} &
Determined by IP-only &
Deterministic linear interpolation between observed visits to impute missing values, followed by IP-only modeling. \\
\midrule
\multicolumn{5}{l}{\textit{IP+IO}}\\
\makecell[l]{Adapted-Liang\\ \scriptsize\citep{liang2009joint}} &
\ding{51} & \ding{51} &
Visiting \& Outcome &
Joint model structure defined, but observation coefficients are constrained to zero. \\
\makecell[l]{EHRJoint\\ \scriptsize\citep{du2025newstatisticalapproachjoint}} &
\ding{51} & \ding{51} &
Visiting \& Outcome &
Tripartite model; observation process modeled via covariates only (no latent link). \\
\makecell[l]{Proposed\\ \scriptsize(This paper)} &
\ding{51} & \ding{51} &
All three processes &
Tripartite model; full shared latent structure links Visiting, Observation, and Outcome. \\
\bottomrule
\end{tabularx}
\endgroup
\end{table}

\subsection{Outcome-Only Methods}

The simplest method, frequently employed in high-throughput phenotyping studies, involves reducing the longitudinal trajectory to a single cross-sectional summary statistic, such as the mean, median, or maximum of the observed measurements. This derived outcome is then regressed on covariates. While computationally convenient, this approach discards temporal information and yields biased estimates when the frequency of measurement is correlated with disease severity.

The other standard statistical framework for analyzing such longitudinal data is the linear mixed-effects model (LMM). This approach models the outcome trajectory $Y_i(t)$ as:
\begin{equation*}
Y_i(t) = \beta_0(t) + \bm{\beta}^\top \bm{X}_i^{\mathcal{Y}}(t) + \bm{b}_i^\top \bm{Z}_i^{\mathcal{Y}}(t) + \varepsilon_i(t),
\end{equation*}
where $\beta_0(t)$ is a baseline mean function, $\bm{b}_i \sim \mathcal{N}(\bm{0}, \bm{\Sigma}_b)$ are subject-specific random effects, and $\varepsilon_i(t) \sim \mathcal{N}(0, \sigma_\varepsilon^2)$. Estimation is typically performed via restricted maximum likelihood (REML) conditional on the realized observation times. 

To mitigate clinically informed bias within the LMM framework without fully modeling the visiting process, augmented specifications introduce constructed variables to serve as proxies for the unobserved health status. The Visit-Aware LMM (VA-LMM) assumes that the cumulative number of visits up to time $t$, denoted $H_i^N(t) = \int_0^t \dint N_i(s)$, acts as a surrogate for disease severity (or frailty) and includes it as an additional time-varying covariate \citep{goldstein2016controlling}:
\begin{equation*}
Y_i(t) = \beta_0(t) + \bm{\beta}^\top \bm{X}_i^{\mathcal{Y}}(t) + \gamma_N H_i^N(t) + \bm{b}_i^\top \bm{Z}_i^{\mathcal{Y}}(t) + \varepsilon_i(t).
\end{equation*}
Similarly, the Observation-Aware LMM (OA-LMM) employs the historical count of actual biomarker measurements, $H_i^R(t) = \int_0^t R_i^{\mathcal{Y}}(s) \dint N_i(s)$, as a proxy to capture the IP and IO mechanism.

\subsection{IP-Only Methods}

Methods in this class explicitly account for the visiting process $N_i(t)$ to correct for IP, but generally assume that the observation process is noninformative, implying that if a patient visits, the biomarker is observed with probability one (i.e., $R_i^{\mathcal{Y}}(t) \equiv 1$).

One widely used approach involves Inverse Intensity Rate Ratio (IIRR-Weighting) \citep{lin2001semiparametric, burvzkova2007longitudinal}. These methods adopt a marginal modeling perspective, specifying a proportional rates model for the visiting process, $\lambda_i(t) = \lambda_0(t) \exp\left\{\bm{\gamma}^\top X_i^{\mathcal{V}}(t) \right\}$, and estimating regression parameters $\bm{\beta}$ by solving a weighted estimating equation:
\[
\sum_{i=1}^n \int_0^\tau \frac{1}{\widehat{V}_i(t)} \bm{D}_i(t) \left\{ Y_i(t) - \bm{\beta}^\top \bm{X}_i^{\mathcal{Y}}(t) \right\} \dint N_i(t) = 0,
\]
where $\widehat{V}_i(t)$ is the estimated relative visit intensity. We refer to this approach as IIRR weighting. The validity of IIRR hinges on the Visiting-at-Random (VAR) assumption: $\dint N_i(t) \perp Y_i(t) \mid \mathcal{H}_i(t^-)$, where $\mathcal{H}_i$ represents observed history. Extensions using balancing weights \citep{yiu2025accommodating} have been proposed to robustify estimation against misspecification of the visiting intensity $\lambda_i(t)$; we refer to this approach as IIRR-Balanced.

Alternatively, pairwise composite likelihood methods \citep{chen2015regression, shen2019regression}, denoted as PairLik, avoid full specification of the visiting process by constructing an objective function based on pairs of observations conditioned on their order. Under the assumption that the visiting process satisfies a separability condition (factorizing into outcome-dependent and covariate-dependent components), nuisance parameters governing the visiting process cancel out. The estimator for $\bm{\beta}$ is obtained by maximizing the pairwise objective function:
\[
\ell_p(\bm{\beta}) = \sum_{i<j} \int_0^\tau \int_0^\tau - \log \left\{ 1 + \frac{f(Y_i(s)| \bm{X}_j(t); \bm{\beta}) f(Y_j(t)| \bm{X}_i(s); \bm{\beta})}{f(Y_i(s)| \bm{X}_i(s); \bm{\beta}) f(Y_j(t)| \bm{X}_j(t); \bm{\beta})} \right\} \dint N_i(s) \dint N_j(t),
\]
where $f(\cdot)$ denotes the working conditional density of the outcome. 

A pivotal framework for addressing IP via random effects is the semiparametric joint model by \citet{liang2009joint}, which we refer to as Liang. This approach explicitly links the visiting and outcome processes via a shared scalar frailty $\eta_i$. The visiting process is specified as
\[
\lambda_i(t) = \eta_i \lambda_0(t) \exp(\bm{\gamma}^\top \bm{X}_i^{\mathcal{V}}),
\]
where $\eta_i$ is a subject-specific latent variable acting multiplicatively on the risk. Crucially, $\eta_i$ is assumed to follow a Gamma distribution, $\eta_i \sim \text{Gamma}(\phi^{-1}, \phi^{-1})$, with $\E(\eta_i)=1$ and $\Var(\eta_i)=\phi$. Dependence is induced by linking $\eta_i$ to the outcome random effects $\bm{b}_i$ via a linear conditional expectation:
\[
\E[\bm{b}_i \mid \eta_i] = \bm{\theta}(\eta_i - 1).
\]

Another related IP-aware joint model links the visiting and longitudinal processes through a shared Gaussian random effect, commonly referred to as JMVL-G
\citep{gasparini2020mixed}. Specifically, let $S_{ij}=T_{ij}-T_{i(j-1)}$ denote the inter-visit time.
JMVL-G specifies a proportional hazards model for $S_{ij}$,
\[
r\!\left(s_{ij}\mid \bm{W}_{ij},u_i\right)=r_0(s_{ij})\exp\!\left(\bm{\gamma}^\top \bm{W}_{ij}+u_i\right),
\]
and a longitudinal model in which the same latent variable enters the outcome,
\[
Y_{ij} = \bm{\beta}^\top X_{ij} + \beta_u u_i + o_i + \varepsilon_{ij},
\]
where $(u_i,o_i)^\top$ are jointly Gaussian random effects and $\varepsilon_{ij}$ is mean-zero noise.
By construction, JMVL-G addresses IP via the shared latent effect in the visiting model.

\subsection{Imputation+IP Methods}

A sequential strategy is to handle IO through imputation and then apply various IP-only methods. Assuming that the observation decision is Missing at Random (MAR) given the observed history, i.e.,
\[
\pr\{R_i^{\mathcal{Y}}(t)=1 \mid Y_i(t), \mathcal{H}_i(t)\}
=
\pr\{R_i^{\mathcal{Y}}(t)=1 \mid \mathcal{H}_i(t)\},
\]
missing values at observed visits can be imputed. Multiple imputation (MI) \citep{Rubin1987MI} can then be used to generate completed datasets, each of which is analyzed using an IP-only method, with the resulting estimates combined using standard MI rules.

Alternatively, \citet{aalen2010dynamic} proposed the Linear Increments (LI) method, which models the increment $\Delta Y_i(t) = Y_i(t) - Y_i(t^-)$ rather than the value itself. Under the Discrete-Time Independent Censoring assumption, unobserved increments are imputed to reconstruct the trajectory. The completed datasets are then analyzed using IP-only methods such as IIRR or Liang.

\subsection{IP+IO Methods}


The EHRJoint framework \citep{du2025newstatisticalapproachjoint} extends the Liang framework \citep{liang2009joint} to explicitly model the observation process. The visiting and outcome processes follow the specification of \citet{liang2009joint} with shared Gamma frailty. The observation indicator $R_i^{\mathcal{Y}}(t)$ is modeled via a logistic regression
\[
\text{logit}\,\left[ \pr\left\{R_i^{\mathcal{Y}}(t)=1 \mid \mathcal{H}_i(t)\right\} \right] = \bm{\alpha}^\top \bm{X}_i^{\mathcal{O}}(t).
\]
Unlike our proposed unified framework, EHRJoint treats the observation process as a nuisance mechanism that is conditionally independent of the visiting and outcome processes given observed covariates; the shared frailty $\eta_i$ does not enter the observation model. 

Following the benchmarking setup
in \citet{du2025newstatisticalapproachjoint}, we refer to this approach as Adapted-Liang:
it augments the Liang method with an observation model (e.g., logistic regression for $R_i^{\mathcal{Y}}(t)$) but fixes the observation coefficients at zero, i.e., $\bm{\alpha}\equiv\bm{0}$. Consequently, Adapted-Liang does not correct for IO and reduces to fitting Liang on the observed measurements, serving as a comparator rather than a fully IP+IO method.

\section{Simulation Study}\label{sec:simulation}

\subsection{Data Generation}

We conducted comprehensive simulations to evaluate the performance of the proposed GIVEHR method under various data-generating mechanisms. For each configuration, we generated $1,000$ replicated datasets with sample size $n = 1,000$ over a follow-up period $[0, \tau]$ with $\tau = 60$.

In all scenarios, we generated baseline covariates $F_i \sim \text{Bernoulli}(0.5)$ and $X_i \sim \mathcal{N}(0,1)$ for $i = 1, \ldots, n$. The longitudinal outcome $Y_i(t)$ was generated according to the linear mixed-effects model:
\begin{equation*}
Y_i(t) = \beta_0 + \beta_t t + \beta_F F_i + \beta_X X_i + b_{0i} + b_{1i} F_i + \varepsilon_i(t),
\end{equation*}
with fixed effects $\beta_0 = -2$, $\beta_t = 0.1$, $\beta_F = -0.5$, and $\beta_X = 0.5$. The random effects $\mathbf{b}_i = (b_{0i}, b_{1i})^\top$ followed a normal distribution with mean $\boldsymbol{\mu}_b$ and covariance $\Sigma_b = \text{diag}(1, 2)$. The measurement error was $\varepsilon_i(t) \sim \mathcal{N}(0, 1)$. 

The visiting and observation processes varied across settings, as detailed below and summarized in Table~\ref{tab:setting_abc}. Setting~A evaluated performance under correct model specification with increasing levels of complexity. Setting~B examined robustness to visiting process misspecification under MNAR observation mechanisms, employing six distinct visiting processes that deviated from the assumed model (Scenarios~B.1--B.6). Setting~C utilized the same six visiting mechanisms as Setting~B but incorporated an MAR observation process conditional solely on measured covariates:
\[
\Pr\left\{R_i^{\mathcal{Y}}(t)=1 \mid \mathrm{d} N_i(t)=1\right\} = \text{logit}^{-1}(0 + 1.0\, F_i + 1.0\, X_i).
\]
This setting further evaluated robustness to link function misspecification by comparing logistic and probit models. 

\begin{table}[H]
\centering
\caption{Simulation Scenarios for Settings A, B, and C}
\label{tab:setting_abc}
\small
\begin{tabular}{@{}lp{6.8cm}p{6.5cm}@{}}
\toprule
\textbf{Scenario} & \textbf{Visiting Process} & \textbf{Observation Process} \\
\midrule
\multicolumn{3}{l}{\textit{Setting A: Correct Model Specification with Varying Complexity}} \\
\addlinespace[0.05cm]
A.1 (Non-informative) & Fixed regular intervals & 0.5 \\
\addlinespace[0.1cm]
A.2 (Measured) & $\exp(-3.5 + F_i + X_i)$ & $\text{logit}^{-1}(0 - 0.5 F_i - 0.5 X_i)$ \\
\addlinespace[0.1cm]
A.3 (Partial Latent) & $\eta_i \exp(-3.5 + F_i + X_i)$ & $\text{logit}^{-1}(0 - 0.5 F_i - 0.5 X_i)$ \\
 & $\eta_i \mid b_{1i} \sim \text{Gamma}\{1/\exp(0.3 b_{1i}), 1\}$ & \\
\addlinespace[0.1cm]
A.4 (Full Latent) & $\eta_i \exp(-3.5 + F_i + X_i)$ & $\Phi\{2 + F_i - X_i + (-0.2 - 0.6 F_i) U_i\}$ \\
 & $\eta_i = \exp(-0.5 + U_i)$, $U_i \sim \mathcal{N}(0,1)$  \\
\midrule
\multicolumn{3}{l}{\textit{Setting B: Robustness to Visiting Misspecification where $E_i \sim \mathcal{N}(0,1)$}} \\
\addlinespace[0.05cm]
B.1 (Latent-only) & $\exp(-2.5 + 0.8 E_i)$  & Same for all B scenarios: \\
\addlinespace[0.1cm]
B.2 (Measured and Latent) & $\exp(-2.5 + 0.3 F_i + 0.3 X_i + 0.8 E_i)$ & $\Phi(0.2 + 0.4 F_i + 0.3 X_i$ \\
\addlinespace[0.1cm]
B.3 (Gamma Frailty) & $\eta_i \exp(-2.5 + 0.3 F_i + 0.3 X_i)$ & $\quad + 0.8 E_i + 0.5 F_i E_i)$ \\
 & $\eta_i \sim \text{Gamma}(\phi_i, \phi_i)$, $\phi_i = \exp(1.0 - 0.6 E_i)$  \\
\addlinespace[0.1cm]
B.4 (Previous Outcome) & $\exp\{-3.0 + 0.8 Y_i(t_{j-1})\}$ & \\
\addlinespace[0.1cm]
B.5 (Threshold) & Visits when $Y_i(t)$ exceeds 60th percentile & \\
\addlinespace[0.1cm]
B.6 (Mixture) & Each subject assigned to one of B.1--B.5 & \\
\bottomrule
\end{tabular}

\vspace{0.3em}
\raggedright
\footnotesize
\textit{Note:} Setting C replicates all Setting B scenarios but replaces the probit observation model with a logistic specification: $\Pr\{R_i^{\mathcal{Y}}(t)=1 \mid \mathrm{d} N_i(t)=1\} = \text{logit}^{-1}(0 + 1.0\, F_i + 1.0\, X_i)$.
\end{table}

\subsection{Benchmark Methods and Evaluation Metrics}

To benchmark performance, we compared the proposed GIVEHR method against 20 existing approaches spanning four methodological categories. The outcome-only category comprises seven methods: cross-section regressions on summary statistics (mean, median, minimum and maximum), denoted Summary-Mean, Summary-Median, Summary-Min and Summary-Max, and linear mixed models using standard, outcome-aligned \citep{goldstein2016controlling}, and visit-aligned frameworks, denoted standard LMM, OA-LMM, and VA-LMM. The IP-only category includes five methods: the estimator of \citet{liang2009joint}, pairwise likelihood \citep{chen2015regression, shen2019regression}, generalized joint models \citep{gasparini2020mixed}, and inverse intensity rate ratio weighting \citep{burvzkova2007longitudinal} with original and balanced \citep{yiu2025accommodating} formulations, denoted Liang, PairLik, JMVL-G, IIRR-Weighting, and IIRR-Balanced. The imputation+IP category combines two imputation strategies---multiple imputation (MI) and linear increments (LI)---with three IP-only models (Liang, PairLik, and JMVL-G), yielding six method variants. The IP+IO category contains three methods: Adapted-Liang, EHRJoint, and our proposed GIVEHR. The systematic comparison is presented in Table~\ref{tab:methods_summary}.

Across all settings, we focused on estimation of $\beta_F$ because $F_i$ enters the outcome model through both a fixed effect and a random effect, introducing additional heterogeneity that makes accurate estimation more challenging. Performance was assessed using two metrics computed across $1,000$ replications. The empirical bias was calculated as the average deviation of the estimates from the true parameter value $\beta_F = -0.5$:
\[
\text{Bias}(\widehat{\beta}_F) = \frac{1}{1,000}\sum_{r=1}^{1,000} \left\{\widehat{\beta}_F^{(r)} - \beta_F\right\}.
\]
The root mean squared error (RMSE) quantifies total estimation error combining bias and variance:
\[
\text{RMSE}(\widehat{\beta}_F) = \sqrt{\frac{1}{1,000}\sum_{r=1}^{1,000} \left\{\widehat{\beta}_F^{(r)} - \beta_F\right\}^2}.
\]
Results for Settings A, B, and C are presented in Figs.~\ref{fig:settingA_bias_rmse}, \ref{fig:settingB_bias_rmse}, and \ref{fig:settingC_bias_rmse}, respectively.

\begin{figure}[htbp]
  \centering
  \includegraphics[width=\textwidth]{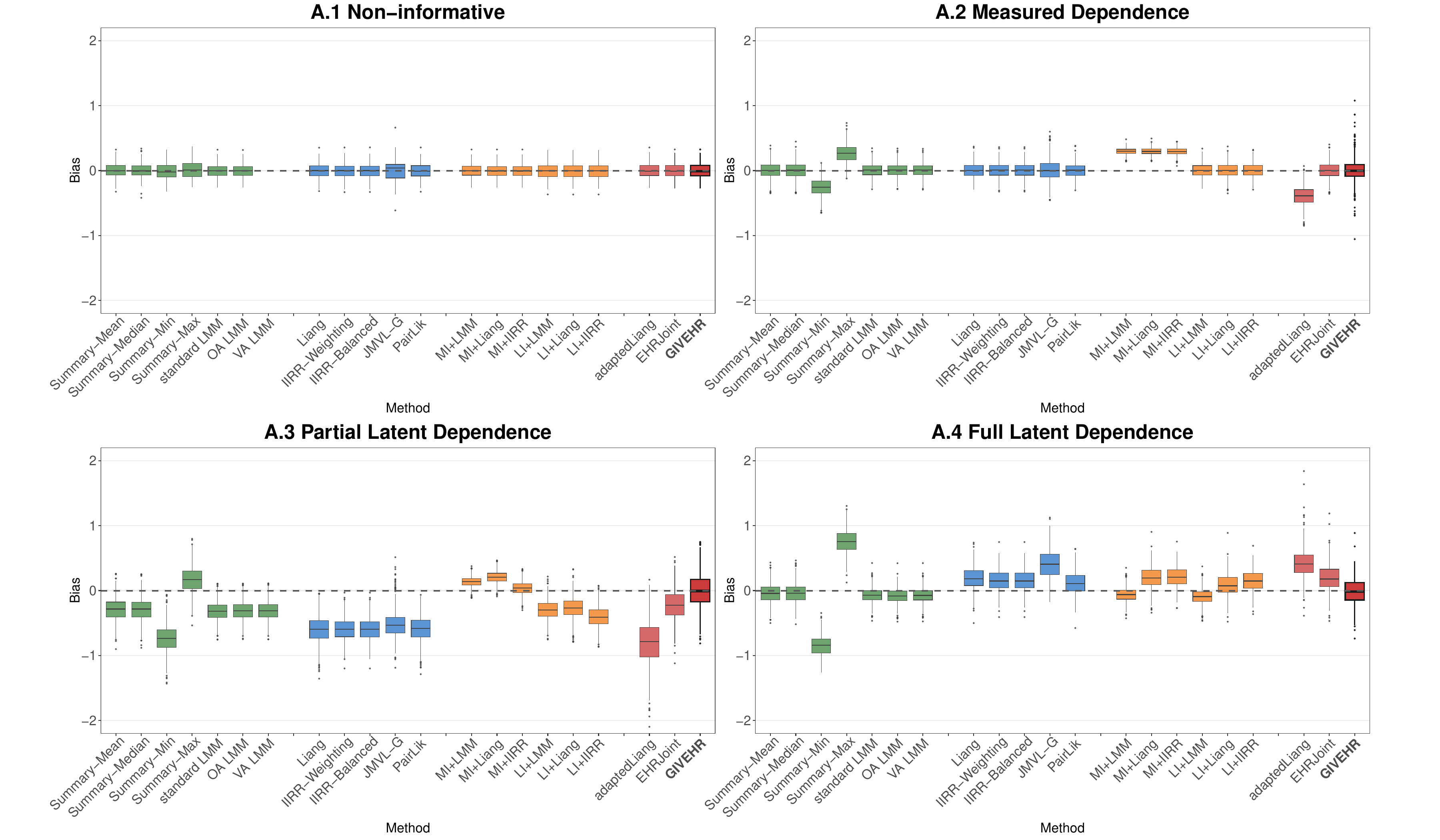}
  \vspace{0.6em}
  \includegraphics[width=\textwidth]{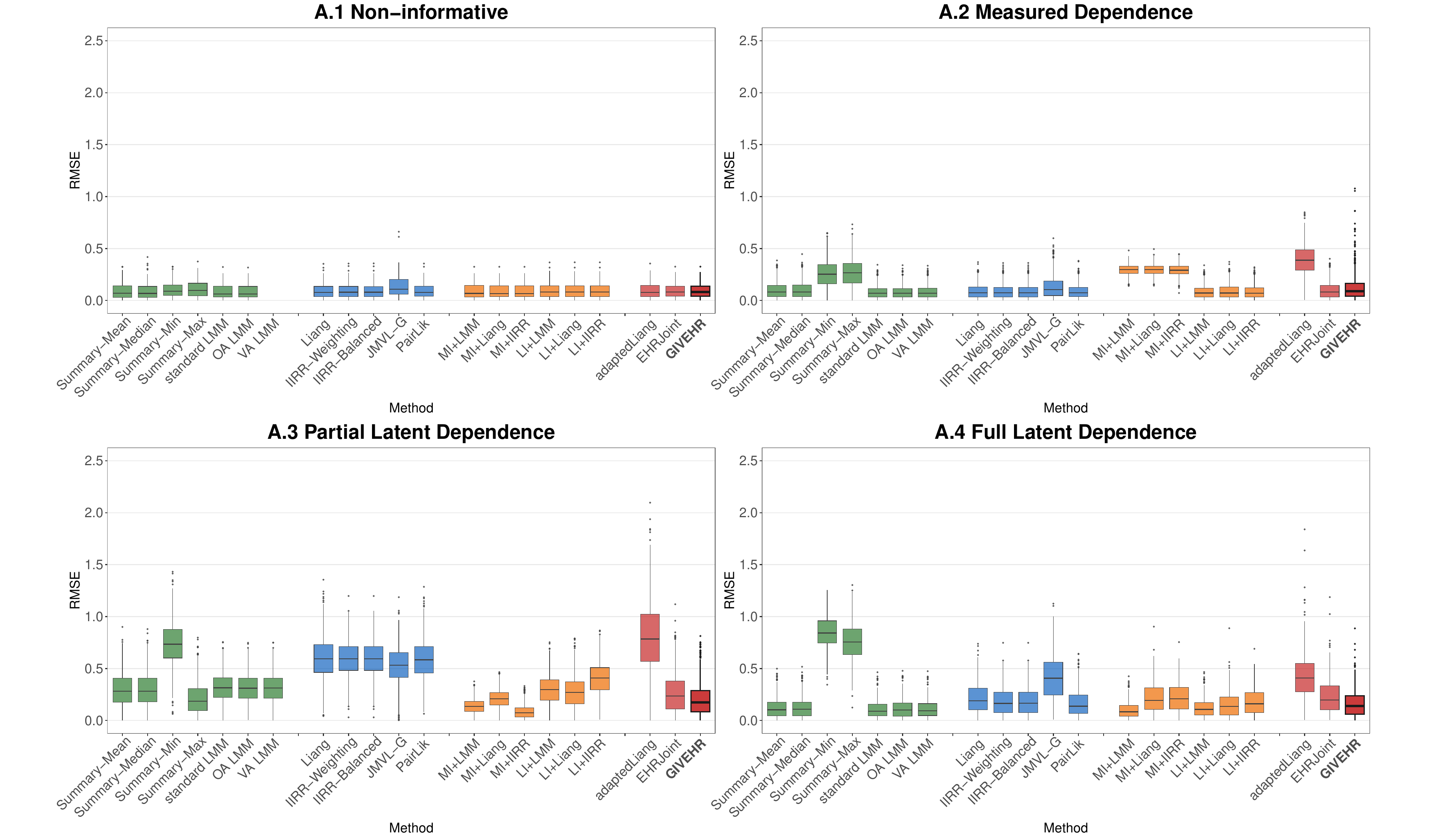}
  \caption{Evaluation of $\beta_F$ estimator performance across Setting~A (Scenarios A.1--A.4). Top: Empirical bias of $\widehat\beta_F$ (dashed line at 0). Bottom: RMSE of $\widehat\beta_F$. Boxplots summarize the distributions across replicates. Estimators are grouped by modeling approach and distinguished by color: Outcome-only (green), IP-only (blue), imputation+IP (orange), and IP+IO (red).}
  \label{fig:settingA_bias_rmse}
\end{figure}

\begin{figure}[htbp]
  \centering
  \includegraphics[width=\textwidth]{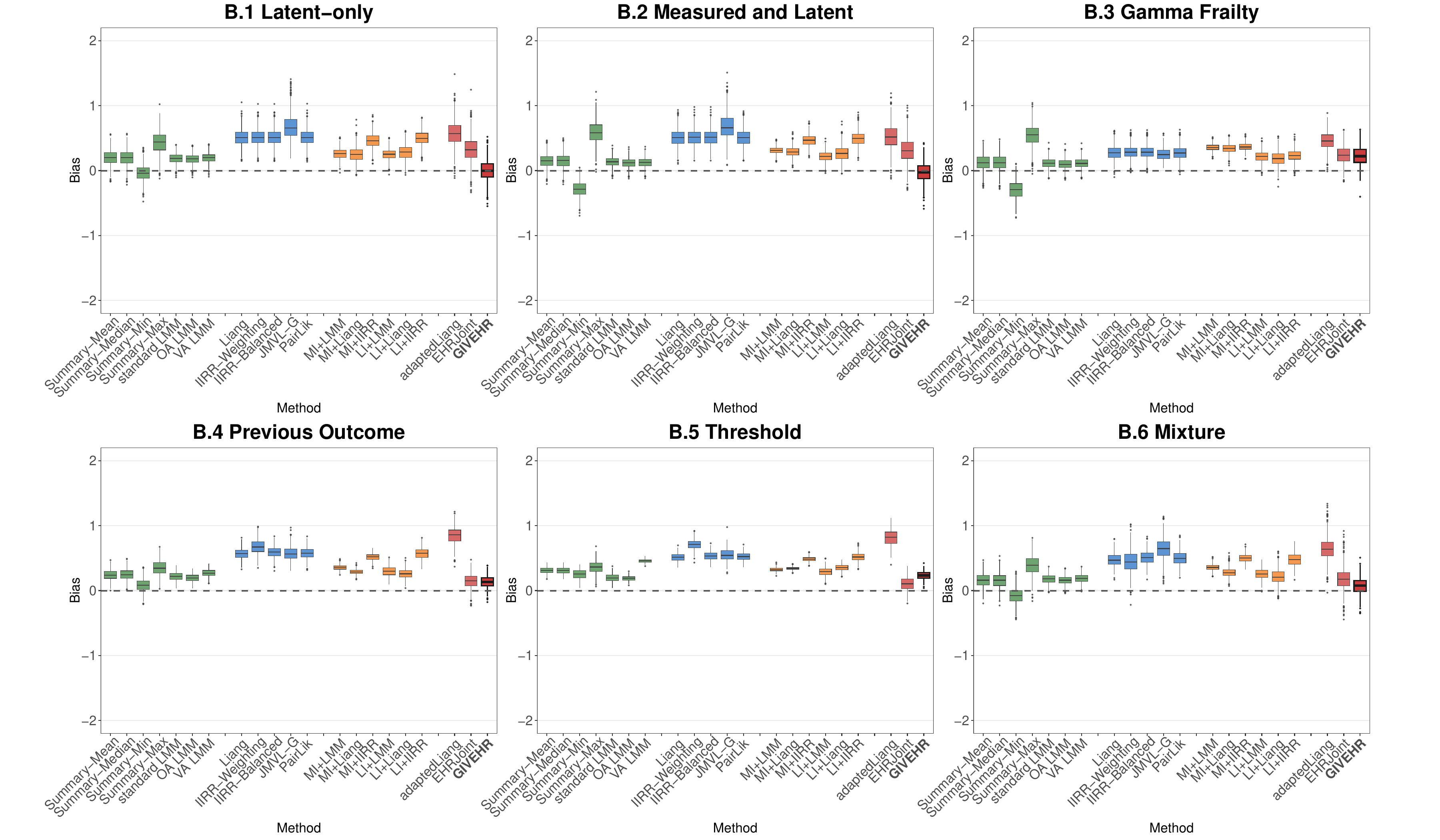}
  \vspace{0.6em}
  \includegraphics[width=\textwidth]{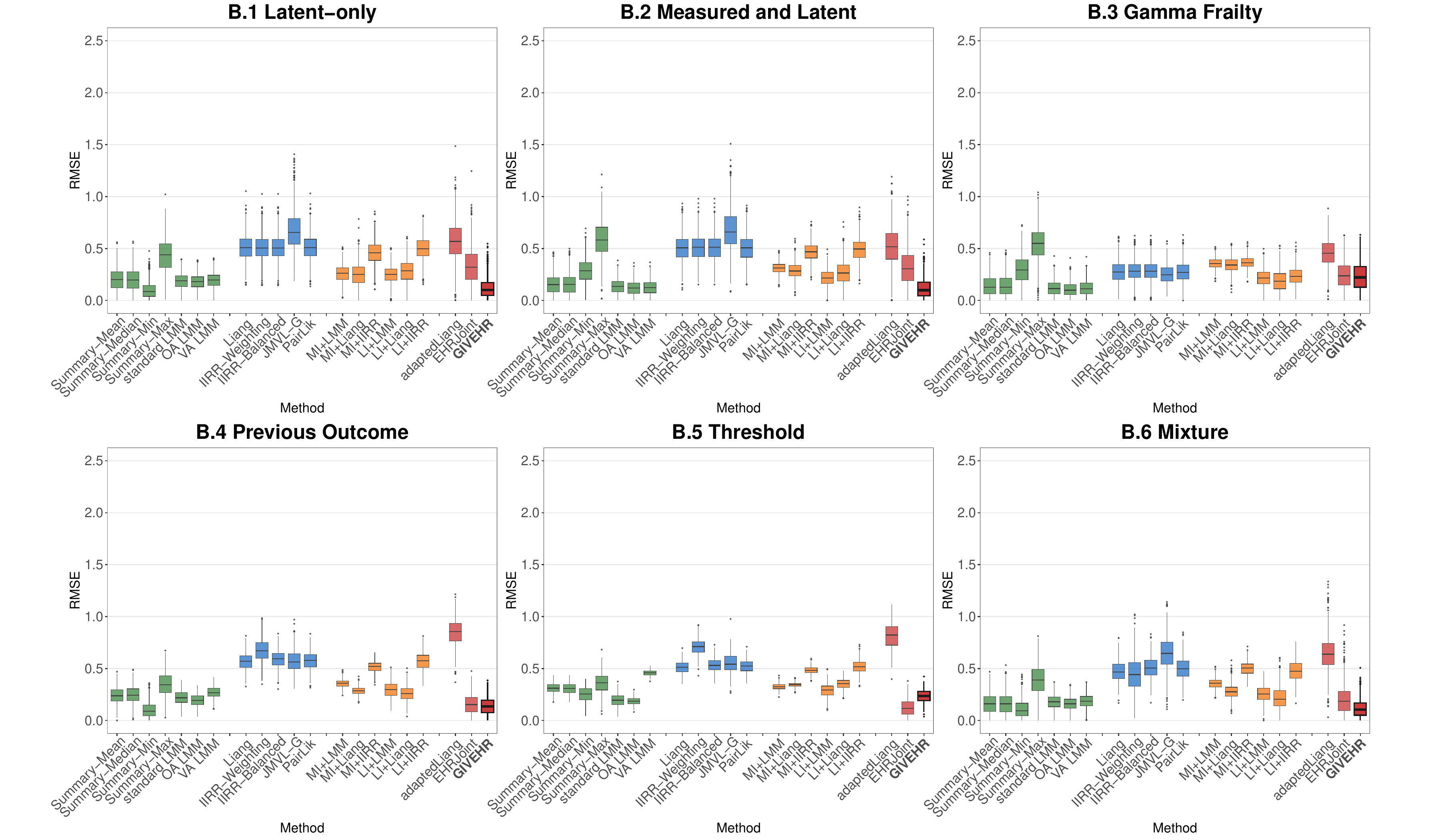}
  \caption{Evaluation of $\beta_F$ estimator performance across Setting~B (Scenarios B.1--B.6). Top: Empirical bias of $\widehat\beta_F$ (dashed line at 0). Bottom: RMSE of $\widehat\beta_F$. Boxplots summarize the distributions across replicates. Estimators are grouped by modeling approach and distinguished by color: Outcome-only (green), IP-only (blue), imputation+IP (orange), and IP+IO (red).}
  \label{fig:settingB_bias_rmse}
\end{figure}

\begin{figure}[htbp]
  \centering
  \includegraphics[width=\textwidth]{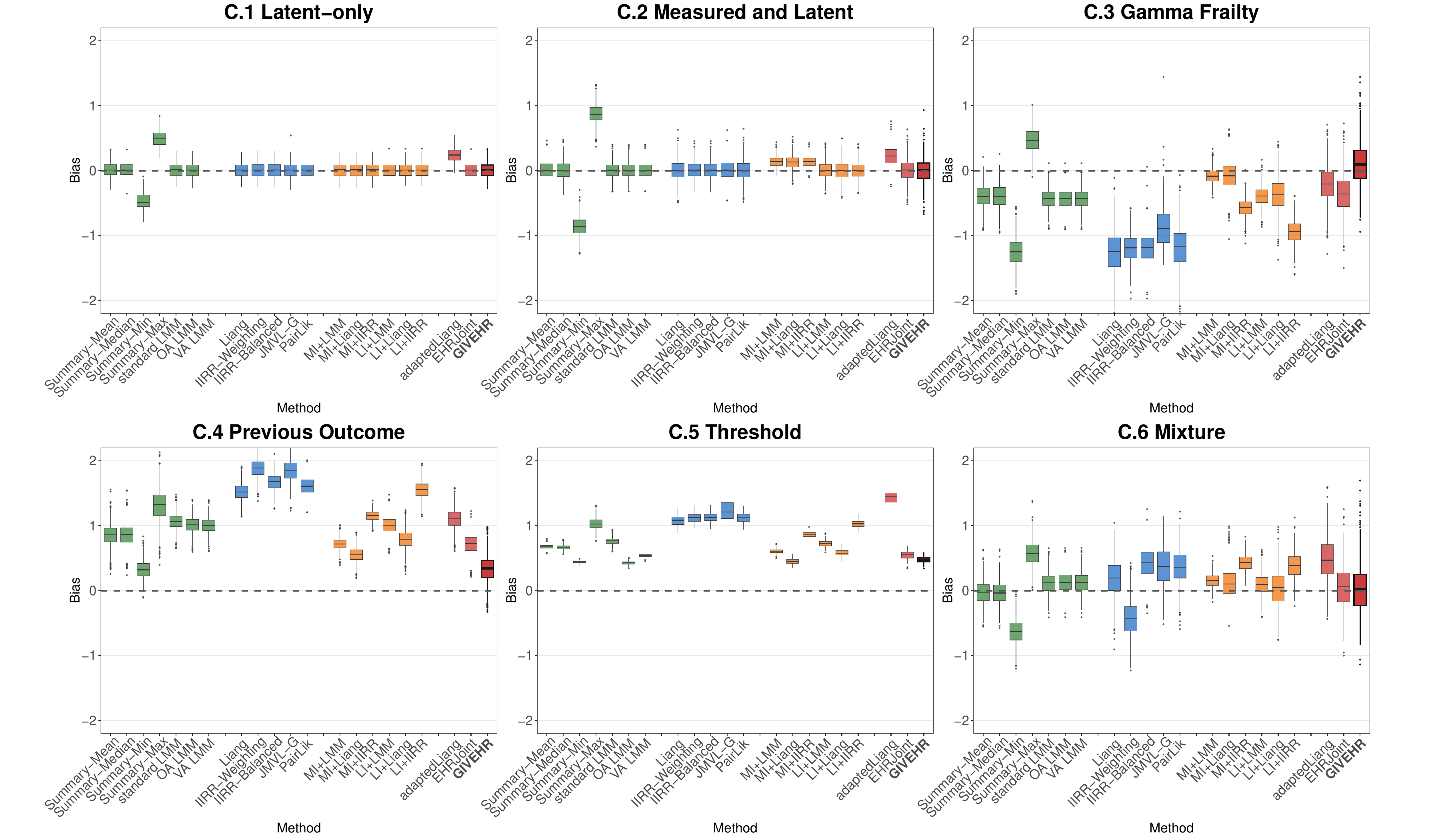}
  \vspace{0.6em}
  \includegraphics[width=\textwidth]{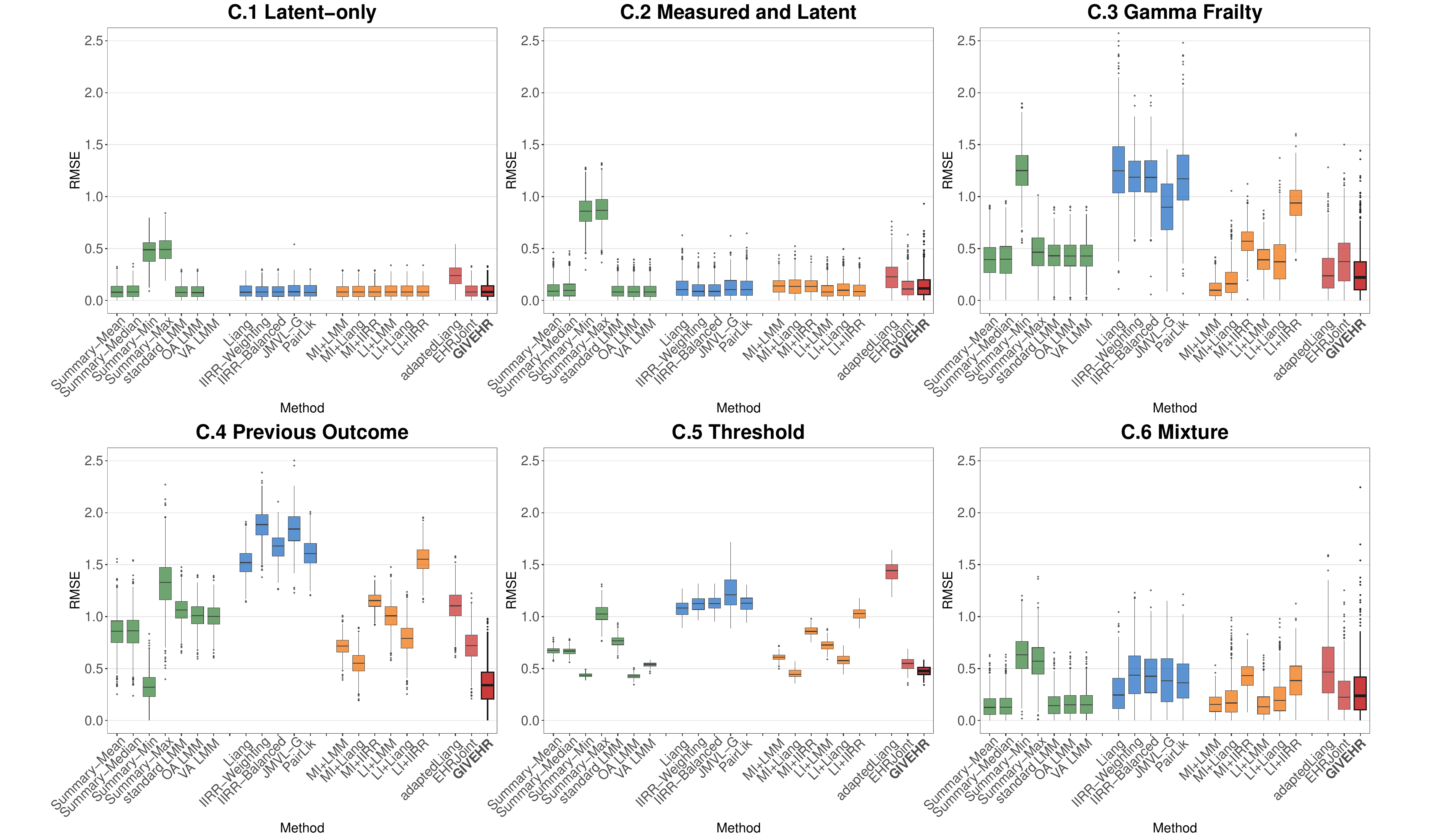}
  \caption{Evaluation of $\beta_F$ estimator performance across Setting~C (Scenarios C.1--C.6). Top: Empirical bias of $\widehat\beta_F$ (dashed line at 0). Bottom: RMSE of $\widehat\beta_F$. Boxplots summarize the distributions across replicates. Estimators are grouped by modeling approach and distinguished by color: Outcome-only (green), IP-only (blue), imputation+IP (orange), and IP+IO (red).}
  \label{fig:settingC_bias_rmse}
\end{figure}

\subsection{Simulation Results for Setting A under Correct Specification of Visiting and Observation Processes with Increasing Complexity}

\textbf{Bias.}
Under the non-informative benchmark (Scenario~A.1), all 21 methods were empirically unbiased with absolute biases below $0.001$. Differences emerged as informative mechanisms were introduced. In Scenario~A.2, where both visiting and observation depended only on observed covariates, outcome-only methods such as the standard LMM (bias $0.005$) performed well, consistent with the validity of maximum likelihood under MAR conditions. IP-only methods showed mixed results: Liang, IIRR, and PairLik exhibited low biases of approximately $0.004$, whereas JMVL-G showed substantial bias. 
Imputation+IP methods varied: LI+Liang was robust (bias $0.004$), while MI+Liang exhibited a bias of $0.295$, suggesting misspecification of the parametric imputation model. 
Among IP+IO methods, Adapted-Liang---which restricts observation coefficients to zero---exhibited a bias of $-0.389$; EHRJoint (bias $0.003$) and GIVEHR (bias $-0.003$) correctly modeled the measured dependence. As latent dependence increased from A.3 to A.4, a clear separation emerged. IP-only methods exhibited the most severe deterioration, with biases reaching $-0.600$ in A.3; in A.4, JMVL-G showed a positive bias of $0.414$. Within the IP+IO category, Adapted-Liang showed the largest bias of $-0.799$ in A.3. EHRJoint reduced these biases considerably ($-0.228$ in A.3, $0.197$ in A.4) but retained residual error because its assumed dependence structure, $\E(\bm{b}_i \mid \eta_i) = \bm{\theta}(\eta_i - 1)$, does not match the true generating mechanism. GIVEHR maintained near-zero bias across all four scenarios (maximum absolute bias $0.011$ in A.4).

\textbf{RMSE.}
The RMSE results mirrored the bias patterns. In Scenario~A.1, GIVEHR exhibited slightly larger RMSE than outcome-only methods due to estimation of additional nuisance parameters. In A.2, the standard LMM achieved the lowest RMSE of $0.100$, while Adapted-Liang showed $0.416$. Under latent dependence, GIVEHR achieved the lowest RMSE in both A.3 ($0.260$ versus $0.322$ for EHRJoint and $0.562$--$0.632$ for IP-only methods) and A.4 ($0.211$ versus $0.288$ for EHRJoint), with the advantage increasing as the complexity of the latent structure grew.

\textbf{Variance estimation.}
We additionally evaluated the cluster-robust sandwich variance estimator from Section~\ref{sec:asy_properties} under Scenario~A.4. The empirical standard errors, computed as the standard deviation of the point estimates across 1,000 simulation runs, were $0.212$ for $F_i$ and $0.128$ for $X_i$. The ratios of the average sandwich standard error to these empirical values were $1.08$ and $1.04$, respectively, indicating that the model-based standard errors slightly overestimate the true sampling variability. These results suggest that the sandwich estimator provides adequate inference at moderate sample sizes, even though it conditions on upstream nuisance as fixed.

\subsection{Simulation Results for Setting B under Misspecification of Visiting Process}

\textbf{Bias.}
When the visiting process deviated from the assumed multiplicative frailty structure while observation remained MNAR (Figure~\ref{fig:settingB_bias_rmse}), IP-only methods exhibited the most severe bias among all categories. In Scenario~B.1, Liang, IIRR, and PairLik clustered around a bias of $0.510$ with JMVL-G reaching $0.670$, whereas the standard LMM---an outcome-only method that ignores both mechanisms---yielded a smaller bias of $0.186$. This pattern illustrates a key finding: adjusting for visiting while ignoring observation can amplify bias beyond that of na\"ive approaches. Imputation+IP pipelines offered moderate improvement (MI: $0.261$; LI: $0.251$). Within the IP+IO category, Adapted-Liang exhibited a bias of $0.568$ and EHRJoint $0.322$, while GIVEHR achieved near-zero bias of $0.001$. This performance gap persisted under the mixture mechanism of Scenario~B.6: IP-only biases ranged from $0.450$ to $0.650$, EHRJoint reduced its bias to $0.177$, and GIVEHR achieved $0.070$.

\textbf{RMSE.}
The RMSE results reinforced these findings. In Scenario~B.1, GIVEHR attained an RMSE of $0.155$, substantially below EHRJoint ($0.374$), Adapted-Liang ($0.600$), and JMVL-G ($0.700$). Under the mixture mechanism (B.6), GIVEHR maintained the lowest RMSE of $0.146$ compared with $0.250$ for EHRJoint and $0.666$ for Adapted-Liang. This robustness stems from the shared Gaussian latent variable, which captures the dependence among the visiting, observation, and outcome processes.

\subsection{Simulation Results for Setting C under Misspecification of Both Visiting and Observation Processes}

\textbf{Bias.}
Setting~C adopted an MAR observation process with a logistic link, introducing link function misspecification relative to the probit model assumed by GIVEHR. Under mild visiting misspecification (Scenarios~C.1 and C.2), nearly all methods were unbiased; for instance, Summary-Mean had a bias of $0.010$ in C.2, and EHRJoint and GIVEHR showed biases of $0.008$ and $0.006$, respectively. Performance diverged sharply in the complex Scenarios~C.3--C.6. IP-only methods again demonstrated vulnerability: in Scenario~C.3, Liang exhibited a bias of $-1.251$, considerably worse than the standard LMM's $-0.434$, confirming that partial IP correction can be counterproductive even when observation is MAR. This arises because IP-only methods attribute observation process effects to the visiting process when the two mechanisms are not distinguished. Imputation+IP methods showed scenario-dependent performance, performing well in C.3 (MI+Liang bias $-0.079$) but poorly under previous outcome-driven visiting in C.4 (bias $0.548$). In Scenario~C.4, GIVEHR achieved a bias of $0.334$, compared with $0.718$ for EHRJoint and $1.066$ for the standard LMM. Under the mixture mechanism (C.6), GIVEHR attained the lowest bias of $0.029$ versus $0.059$ for EHRJoint.

\textbf{RMSE.}
Outcome-only methods achieved the lowest RMSE under mild misspecification due to model parsimony: Summary-Mean attained $0.114$ in C.2 compared with $0.179$ for GIVEHR, reflecting the cost of estimating nuisance parameters that are unnecessary under MAR observation. However, this efficiency advantage vanished in complex scenarios. In C.4, GIVEHR achieved an RMSE of $0.394$ compared with $0.735$ for EHRJoint and $1.073$ for the standard LMM. Similar patterns held across C.3--C.6, with GIVEHR consistently achieving the lowest or near-lowest RMSE.

\subsection{Summary Takeaways}

Across all 16 scenarios, the proposed GIVEHR method consistently achieved the lowest or near-lowest bias and RMSE among the 21 methods evaluated (7 outcome-only, 5 IP-only, 6 imputation+IP, and 3 IP+IO). This advantage stems from its flexible semiparametric specification and the use of shared latent structures to capture dependence across the visiting, observation, and longitudinal outcome processes.

EHRJoint serves as a competitive alternative; while it generally outperformed imputation+IP and outcome-only models, it exhibited residual error in complex latent scenarios, particularly under visiting process misspecification. IP-only methods exhibited the largest biases in most scenarios, with performance deteriorating substantially under MNAR observation mechanisms. Caution is warranted when applying IP-only approaches where informative observation is suspected. The results indicate that partially adjusting for the visiting process while neglecting the observation process can yield larger biases than unadjusted outcome-only methods, as failure to distinguish between visiting and observation leads to distorted estimation of the joint informative mechanism.

\section{Data Analysis}\label{sec:data_analysis}

\subsection{Data Description and Preprocessing}

The GIVEHR method was applied to data from the All of Us Research Program to assess the relationship between neighborhood-level socioeconomic status (NSES) and longitudinal biomarker trajectories, and was compared with six alternative methods: Summary-Median, standard LMM, IIRR, Liang, EHRJoint, and LI+Liang.
The analytic cohort comprised participants aged 18--65 years at enrollment with a minimum follow-up of 5 years. The observation window spanned from 1 January 2018 to 1 October 2023, restricted to outpatient visits.

Baseline characteristics of the study population, overall and by visit status, are summarized in Table~\ref{tab:baseline_chrc}. In total, 206{,}013 participants met the inclusion criteria, of whom 192{,}608 (93.5\%) had at least one outpatient visit during the observation window. Compared with participants who had no outpatient visits, those with at least one visit were slightly older and had a higher prevalence of cancer and chronic disease at baseline.

\begin{table}[H]
\centering
\caption{Baseline Characteristics of Study Population}
\label{tab:baseline_chrc}
\small
\begin{tabular}{@{}l c c c@{}}
\toprule
\textbf{Characteristic} & 
\textbf{Overall} & 
\textbf{$\geq$1 Visit} & 
\textbf{No Visits} \\
& (N=206,013) & (n=192,608) & (n=13,405) \\
\midrule
\multicolumn{4}{@{}l}{\textit{Demographics}} \\
\hspace{3mm}Age, mean (years)        & 46.0 & 46.1 & 44.4 \\
\hspace{3mm}Male sex, \%             & 34.1 & 34.4 & 30.2 \\
\hspace{3mm}Married, \%              & 41.8 & 41.5 & 46.5 \\
\addlinespace[0.5em]
\multicolumn{4}{@{}l}{\textit{Race and Ethnicity}} \\
\hspace{3mm}White, \%                & 71.6 & 70.4 & 88.1 \\
\hspace{3mm}Black, \%                & 28.4 & 29.6 & 11.9 \\
\hspace{3mm}Hispanic, \%             & 22.9 & 22.9 & 23.3 \\
\addlinespace[0.5em]
\multicolumn{4}{@{}l}{\textit{Clinical Characteristics}} \\
\hspace{3mm}Cancer diagnosis, \%     & 12.1 & 12.7 & 4.0 \\
\hspace{3mm}Chronic disease, \%      & 49.2 & 51.1 & 22.6 \\
\bottomrule
\end{tabular}
\end{table}

Our analysis focused on six biomarkers: glucose, glycated hemoglobin (HbA1c), C-reactive protein (CRP), white blood cell count (WBC), prostate-specific antigen (PSA), and cancer antigen 125 (CA125). These biomarkers were selected to represent a spectrum of clinical specificity, ranging from routine markers ordered as part of standard care (glucose, WBC) to highly targeted markers ordered only under specific clinical suspicion (PSA, CA125). This contrast is informative because, for a given patient, the visiting process (IP) remains constant regardless of the biomarker, whereas the observation process (IO) varies substantially; specifically, routine markers are measured at a significantly higher fraction of visits than specialized ones. Analyzing biomarkers across this spectrum allows us to assess the impact of the observation mechanism under a shared visiting process. We performed stratified analyses for glucose by diabetic status and restricted the analyses of CA125 to females and PSA to males with prostate cancer. To investigate associations with NSES, we utilized two indicators derived from the American Community Survey at the three-digit ZIP code level: median annual household income and the uninsured rate. All continuous covariates were standardized to have a mean of zero and unit variance.

Table~\ref{tab:biomarkers} summarizes the measurement patterns across the six biomarkers, confirming the anticipated heterogeneity in the observation process. Among participants with at least one outpatient visit, routine markers such as glucose and white blood cell count were observed for approximately half, whereas targeted markers such as PSA and CA125 were observed in fewer than 6\%. Notably, patients with at least one measurement of a targeted marker exhibited substantially higher healthcare utilization, with median visit counts exceeding 55 compared with approximately 35 for routine markers, consistent with the clinical conditions prompting those tests. Even among patients with available data, measurements occurred at only 6--20\% of visits, underscoring the pervasiveness of the observation mechanism across the entire specificity spectrum.

\begin{table}[H]
\centering
\caption{Biomarker measurement patterns among 206{,}013 patients
(192{,}608 with $\geq$1 clinical visit). Rows are sorted by measurement frequency.
\textbf{Left group}: $N$ is the count of patients with $\geq$1 measurement; Prop.\ is the percentage among patients with $\geq$1 visit.
\textbf{Right group}: Statistics calculated only among patients with $\geq$1 measurement.
``\#\ Visits'' and ``\#\ Meas.'' represent the median [IQR] counts per patient.
``Prop.'' is the median percentage of a patient's visits where the biomarker was observed.}
\label{tab:biomarkers}
\begin{tabular}{lrrrrr}
\toprule
 & \multicolumn{2}{c}{\shortstack[c]{\textbf{Number of Patients with}\\\textbf{at least one measurement}}}
 & \multicolumn{3}{c}{\shortstack[c]{\textbf{Summary of repeated}\\\textbf{measurement per person}}} \\
\cmidrule(lr){2-3} \cmidrule(lr){4-6}
\textbf{Biomarker} 
  & \textbf{$N$} 
  & \textbf{Prop.\ (\%)} 
  & \textbf{\#\ Visits} 
  & \textbf{\#\ Meas.} 
  & \textbf{Prop.\ (\%)} \\ 
\midrule
Glucose & 98{,}966 & 51.4 & 38.0 [16, 78]  & 3.0 [1, 6] &  9.1 \\
WBC     & 94{,}467 & 49.0 & 35.0 [14, 75]  & 2.0 [1, 5] &  4.4 \\
HbA1c   & 45{,}371 & 23.6 & 52.0 [26, 103] & 2.0 [1, 4] &  2.8 \\
CRP     & 13{,}957 &  7.2 & 57.0 [25, 110] & 1.0 [1, 2] &  8.7 \\
PSA     & 10{,}412 &  5.4 & 56.0 [28, 110] & 2.0 [1, 3] &  3.8 \\
CA125   &    797   &  0.4 & 65.0 [35, 120] & 1.0 [1, 3] &  2.6 \\
\bottomrule
\end{tabular}%
\end{table}

\subsection{Estimation Results: Visiting Process}\label{subsec:visit_data}

We modeled outpatient visit intensity as a function of patient characteristics and NSES. For each biomarker-specific cohort, we fitted a visiting intensity model of the form
\begin{equation*}
\begin{split}
    \lambda_i(t) 
    &= \lambda_0(t)\,\eta_i
      \exp\Bigl(
        \gamma_{\mathrm{Age}}\text{Age}_i
      + \gamma_{\mathrm{Male}}\text{Male}_i
      + \gamma_{\mathrm{Mar}}\text{Married}_i
      + \gamma_{\mathrm{Chron}}\text{Chronic}_i \\
    &\qquad\qquad\quad
      + \gamma_{\mathrm{Canc}}\text{Cancer}_i
      + \gamma_{\mathrm{Blk}}\text{Black}_i
      + \gamma_{\mathrm{Hisp}}\text{Hispanic}_i
      + \gamma_{\mathrm{NSES}}\text{NSES}_i
      \Bigr),
\end{split}
\end{equation*}
where $\lambda_0(t)$ is an unspecified baseline intensity function and $\eta_i$ is a subject-level lognormal frailty. The covariates included standardized age, and indicators for male sex, married status, chronic disease, cancer diagnosis, Black race and Hispanic ethnicity; White race and non-Hispanic ethnicity served as reference categories. To examine the role of neighborhood-level context, we fitted separate models for each NSES indicator: standardized median household income and the standardized uninsured rate. The sex indicator was omitted for the female-only CA125 analysis; similarly, the sex indicator and cancer diagnosis were excluded for the PSA analysis, which was restricted to males with prostate cancer.

Table~\ref{tab:visit_gamma} reports the estimated log-intensity ratios, $\widehat{\gamma}$, with Panels~1 and 2 corresponding to the income and uninsured rate specifications, respectively. Across all biomarker cohorts, age, chronic disease and cancer diagnosis were positively associated with visit intensity. For example, in the non-diabetic glucose cohort, chronic disease was associated with a 2.6-fold increase in the visit rate, corresponding to $\widehat{\gamma}_{\mathrm{Chron}} = 0.968$. A similar pattern was observed for cancer diagnosis, with $\widehat{\gamma}_{\mathrm{Canc}}$ ranging from 0.38 to 0.50, indicating substantially higher utilization among patients with a history of cancer.

\begin{table}[H]
\centering
\caption{Estimated log-intensity ratios ($\widehat{\gamma}$) for the visiting process. Panel~1 uses standardized neighborhood-level median household income as the NSES covariate; Panel~2 uses standardized neighborhood-level uninsured rate. Bold values indicate statistical significance at $p\text{-value} < 0.05$. NA indicates covariates not applicable to gender- or disease-restricted cohorts.}
\label{tab:visit_gamma}
\resizebox{\textwidth}{!}{%
\begin{tabular}{@{}lrrrrrrrr@{}}
\toprule
\textbf{Biomarker Cohort} & \textbf{Age} & \textbf{Male} & \textbf{Married} & \textbf{Chronic} & \textbf{Cancer} & \textbf{Black} & \textbf{Hispanic} & \textbf{NSES} \\ 
\midrule
\multicolumn{9}{@{}l}{\textit{Panel 1: NSES = Standardized Median Household Income}} \\
\addlinespace[2pt]
Glucose (Diabetic)            & \textbf{0.079} & \textbf{$-$0.066} & \textbf{$-$0.049} & \textbf{0.890} & \textbf{0.379} & \textbf{0.138} & \textbf{$-$0.222} & \textbf{0.140} \\
Glucose (Non-diabetic)        & \textbf{0.123} & \textbf{$-$0.139} & \textbf{$-$0.075} & \textbf{0.968} & \textbf{0.485} & \textbf{0.102} & \textbf{$-$0.190} & \textbf{0.148} \\
HbA1c, CRP, WBC             & \textbf{0.147} & \textbf{$-$0.109} & \textbf{$-$0.079} & \textbf{1.018} & \textbf{0.459} & \textbf{0.169} & \textbf{$-$0.146} & \textbf{0.141} \\
CA125 (Female)                & \textbf{0.117} & NA                 & \textbf{$-$0.077} & \textbf{0.996} & \textbf{0.463} & \textbf{0.179} & \textbf{$-$0.150} & \textbf{0.132} \\
PSA (Male w/ Prostate Cancer) & \textbf{0.142} & NA                 & $-$0.082          & \textbf{0.672} & NA             & 0.107          & $-$0.075          & \textbf{0.125} \\
\midrule
\multicolumn{9}{@{}l}{\textit{Panel 2: NSES = Standardized Uninsured Rate}} \\
\addlinespace[2pt]
Glucose (Diabetic)            & \textbf{0.087} & \textbf{$-$0.051} & \textbf{$-$0.044} & \textbf{0.845} & \textbf{0.382} & \textbf{0.169} & \textbf{$-$0.126} & \textbf{$-$0.160} \\
Glucose (Non-diabetic)        & \textbf{0.132} & \textbf{$-$0.123} & \textbf{$-$0.067} & \textbf{0.936} & \textbf{0.502} & \textbf{0.108} & \textbf{$-$0.124} & \textbf{$-$0.112} \\
HbA1c, CRP, WBC             & \textbf{0.156} & \textbf{$-$0.094} & \textbf{$-$0.073} & \textbf{0.982} & \textbf{0.471} & \textbf{0.183} & \textbf{$-$0.072} & \textbf{$-$0.124} \\
CA125 (Female)                & \textbf{0.125} & NA                 & \textbf{$-$0.073} & \textbf{0.961} & \textbf{0.474} & \textbf{0.203} & \textbf{$-$0.069} & \textbf{$-$0.128} \\
PSA (Male w/ Prostate Cancer) & \textbf{0.141} & NA                 & $-$0.071          & \textbf{0.662} & NA             & 0.062          & $-$0.059          & $-$0.063 \\
\bottomrule
\end{tabular}%
}
\begin{minipage}{\textwidth}
\small
NSES, neighborhood-level socioeconomic status; PCa, prostate cancer; HbA1c, glycated hemoglobin; CRP, C-reactive protein; WBC, white blood cell count; CA125, cancer antigen 125; PSA, prostate-specific antigen.
\end{minipage}
\end{table}

The estimated NSES effects consistently indicated that higher neighborhood socioeconomic status was associated with increased visit intensity. Specifically, visit intensity was positively associated with median household income (Panel 1: $\widehat{\gamma}_{\mathrm{NSES}} = 0.125$ to $0.148$) and negatively associated with the uninsured rate (Panel 2: $\widehat{\gamma}_{\mathrm{NSES}} = -0.160$ to $-0.063$).

\subsection{Estimation Results: Observation Process}\label{subsec:obs_data}

We modeled the probability of biomarker observation conditional on a clinic visit using the probit mixed effects model described in Section~\ref{sec:estimation}. Table~\ref{tab:obs_process} details the covariates included in each specification. Common covariates across all biomarkers included age, race (Black vs.\ White), ethnicity (Hispanic vs.\ non-Hispanic), baseline body mass index (BMI), chronic disease status, and the NSES indicator. Biomarker-specific covariates were selected to reflect clinical indications for ordering each test. For instance, the glucose observation models included diabetes-related comorbidities---such as type 1 and type 2 diabetes, dyslipidemia, atherosclerotic cardiovascular disease (ASCVD), kidney disease, and pancreatic disease---as these conditions typically prompt clinicians to monitor glycemic status. For inflammatory markers (CRP, WBC), we included conditions associated with inflammation or immune dysfunction. Full covariate lists are reported in Table~\ref{tab:obs_process}.

\begin{table}[H]
\centering
\caption{Covariates included in the probit mixed effects models. Two model specifications were fitted: one using standardized median household income as the NSES covariate, and another using the standardized uninsured rate. Bold text indicates significance ($p\text{-value} < 0.05$) in the income specification; $^{\dagger}$ indicates significance in the uninsured rate specification.}
\label{tab:obs_process}
\resizebox{\textwidth}{!}{%
\begin{tabular}{lll}
\hline
Biomarker & Common covariates & Biomarker-specific covariates \\ \hline
Glucose (Diabetic) & \textbf{Age}$^\dagger$, \textbf{Race}$^\dagger$, \textbf{Ethnicity}, BMI, & \textbf{Sex}$^\dagger$, Dyslipidemia, \textbf{ASCVD}$^\dagger$, \\
 & \textbf{Chronic disease}$^\dagger$, \textbf{NSES}$^\dagger$ & \textbf{Kidney disease}$^\dagger$, \textbf{Pancreatic disease}$^\dagger$, \\
 & & Type 1 diabetes, \textbf{Type 2 diabetes}$^\dagger$ \\ \hline
Glucose (Non-diabetic) & Age$^\dagger$, \textbf{Race}$^\dagger$, Ethnicity, BMI, & \textbf{Sex}$^\dagger$, Dyslipidemia, \textbf{ASCVD}$^\dagger$, \\
 & \textbf{Chronic disease}$^\dagger$, \textbf{NSES}$^\dagger$ & \textbf{Kidney disease}$^\dagger$, \textbf{Pancreatic disease}$^\dagger$ \\ \hline
HbA1c & \textbf{Age}$^\dagger$, \textbf{Race}$^\dagger$, \textbf{Ethnicity}$^\dagger$, BMI, & \textbf{Sex}$^\dagger$, \textbf{Obesity}$^\dagger$, \textbf{Type 1 diabetes}$^\dagger$, \\
 & \textbf{Chronic disease}$^\dagger$, \textbf{NSES}$^\dagger$ & \textbf{Type 2 diabetes}$^\dagger$, \textbf{Dyslipidemia}$^\dagger$, \\
 & & \textbf{ASCVD}$^\dagger$, \textbf{Kidney disease}$^\dagger$ \\ \hline
CRP & \textbf{Age}$^\dagger$, \textbf{Race}$^\dagger$, \textbf{Ethnicity}$^\dagger$, & \textbf{Sex}, \textbf{Autoimmune rheumatic disease}$^\dagger$ \\ 
 & \textbf{BMI}$^\dagger$, \textbf{Chronic disease}$^\dagger$, \textbf{NSES}$^\dagger$ & \\ \hline
WBC & \textbf{Age}$^\dagger$, \textbf{Race}, \textbf{Ethnicity}$^\dagger$, Married$^\dagger$, & \textbf{Sex}$^\dagger$, \textbf{Anemia}$^\dagger$, \textbf{Kidney disease}$^\dagger$, \\ 
 & BMI, \textbf{Chronic disease}$^\dagger$, \textbf{Cancer}$^\dagger$, \textbf{NSES}$^\dagger$ & \textbf{Cancer}$^\dagger$, \textbf{Liver disease}$^\dagger$ \\ \hline
CA125 (Female) & \textbf{Age}$^\dagger$, \textbf{Race}$^\dagger$, Ethnicity, \textbf{BMI}, & \textbf{Ovarian condition}$^\dagger$, Uterine fibroids, \\ 
 & \textbf{Chronic disease}$^\dagger$, NSES$^\dagger$ & \textbf{Pregnancy}$^\dagger$ \\ \hline
PSA (Male w/ PCa) & Age, Race, Ethnicity, BMI, & \textbf{BPH}$^\dagger$, \textbf{Prostatitis}$^\dagger$, LUTS \\
 & \textbf{Chronic disease}$^\dagger$, NSES$^\dagger$ & \\ \hline
\end{tabular}%
}
\vspace{5pt}
\begin{minipage}{\textwidth}
\small
NSES, neighborhood-level socioeconomic status; PCa, prostate cancer; HbA1c, glycated hemoglobin; CRP, C-reactive protein; WBC, white blood cell count; CA125, cancer antigen 125; PSA, prostate-specific antigen.
\end{minipage}
\end{table}

Estimation of the observation model revealed several consistent patterns. Age was significantly associated with observation probability for most biomarkers. Chronic disease served as a robust predictor across all biomarkers ($p\text{-value} < 0.05$). Neighborhood median income exhibited significant positive associations with observation probability for glucose, CRP, and WBC, suggesting that socioeconomic factors influence not only visit frequency but also the decision to order tests.

Associations with disease-specific covariates aligned with established clinical practice. For glucose monitoring, diabetes status (particularly type 2 diabetes), ASCVD, and kidney disease were strongly associated with increased observation probability. For HbA1c, type 1 and type 2 diabetes, dyslipidemia, ASCVD, and kidney disease were significant predictors. The CRP model showed significant associations with autoimmune rheumatic disease, consistent with its utility as an inflammatory marker. For WBC, anemia, kidney disease, cancer, and liver disease were all significantly associated with observation. Among sex-specific markers, ovarian conditions and pregnancy predicted CA125 recording, while benign prostatic hyperplasia (BPH) and prostatitis predicted PSA observation.

Collectively, these results provide strong evidence that the observation process is informative. Furthermore, the estimated parameter $\widehat{\bm{\delta}}$ differs significantly from zero, suggesting that the observation process is MNAR.

\subsection{Estimation Results: Longitudinal Outcomes}\label{subsec:outcome_results}

Having established that both the visiting and observation processes depend on patient covariates and exhibit significant latent heterogeneity (Sections~\ref{subsec:visit_data}--\ref{subsec:obs_data}), we now turn to the longitudinal outcome model. Using the semiparametric framework described in Section~\ref{sec:estimation}, the outcome process for each biomarker was specified as:
\begin{equation*}
\begin{split}
    \log Y_i(t) &= \beta_0(t) + \beta_{\text{Age}} \text{Age}_i + \beta_{\text{Male}} \text{Male}_i + \beta_{\text{Mar}} \text{Married}_i + \beta_{\text{Blk}} \text{Black}_i + \beta_{\text{Hisp}} \text{Hispanic}_i \\
           &\quad + \beta_{\text{NSES}} \text{NSES}_i + b_{0i} + b_{1i} \text{NSES}_i + \varepsilon_i(t),
\end{split}
\end{equation*}
where $\beta_0(t)$ is an unspecified baseline function, NSES represents either standardized neighborhood-level median household income or standardized neighborhood-level uninsured rate, and the random effects $\bm{b}_i = (b_{0i}, b_{1i})^\top$ are linked to the shared latent variable via $\E(\bm{b}_i \mid U_i) = \bm{\theta} U_i$. For sex-specific biomarkers (CA125 restricted to females, PSA restricted to males with prostate cancer), the Male covariate was omitted. Figure~\ref{fig:forest_income} presents the estimated coefficients and 95\% confidence intervals for the association between biomarker levels and standardized neighborhood-level median household income across the seven methods.

\begin{figure}
    \centering
    \includegraphics[width=1\linewidth]{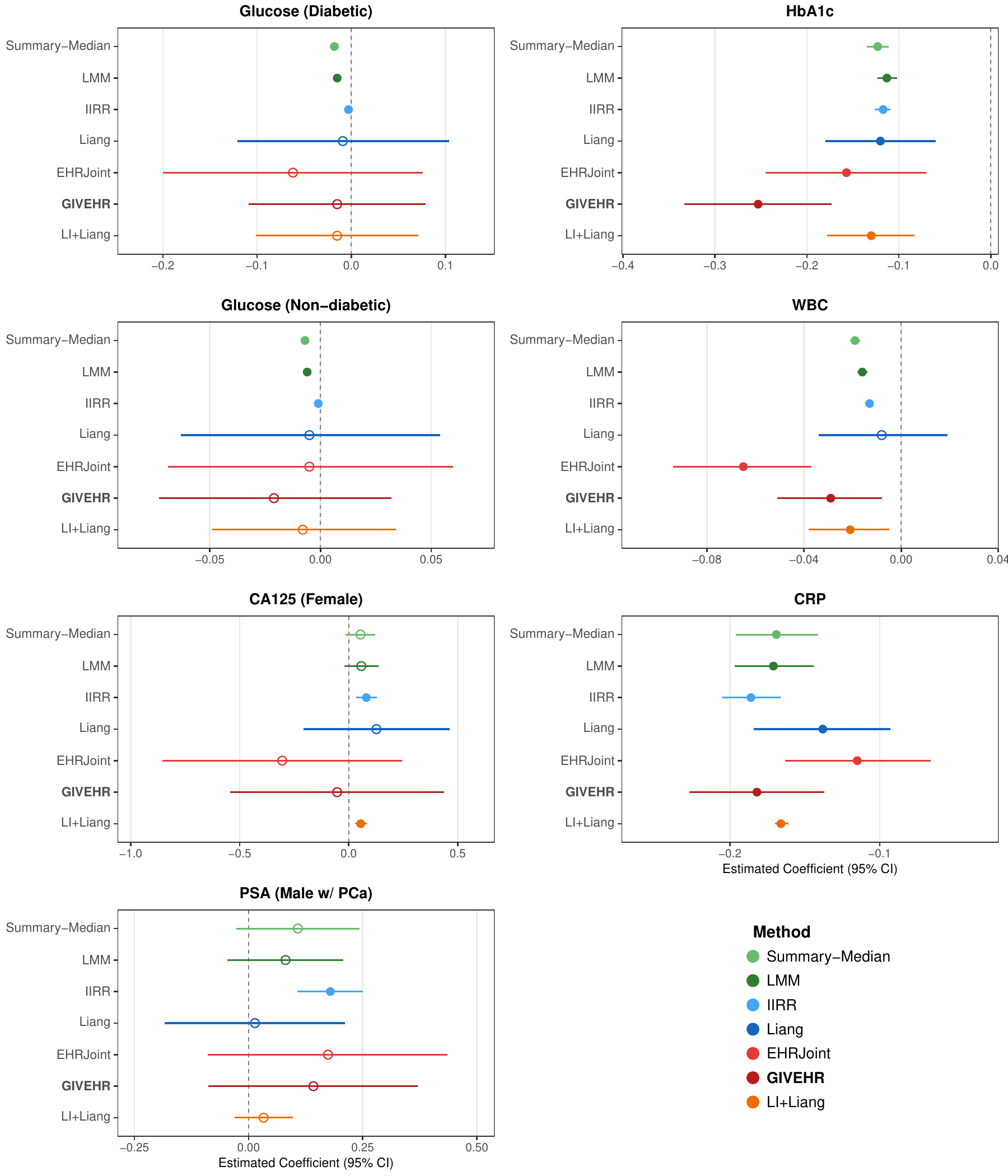}
    \caption{Estimated regression coefficients and 95\% confidence intervals for the association between log-transformed biomarker levels and standardized neighborhood-level median household income, comparing seven estimation methods. The left column displays results for glucose (diabetic and non-diabetic subgroups), CA125 (females), and PSA (males with prostate cancer); the right column displays results for HbA1c, WBC, and CRP. Filled circles denote statistical significance at the 0.05 level; open circles denote non-significance.}
\label{fig:forest_income}
\end{figure}

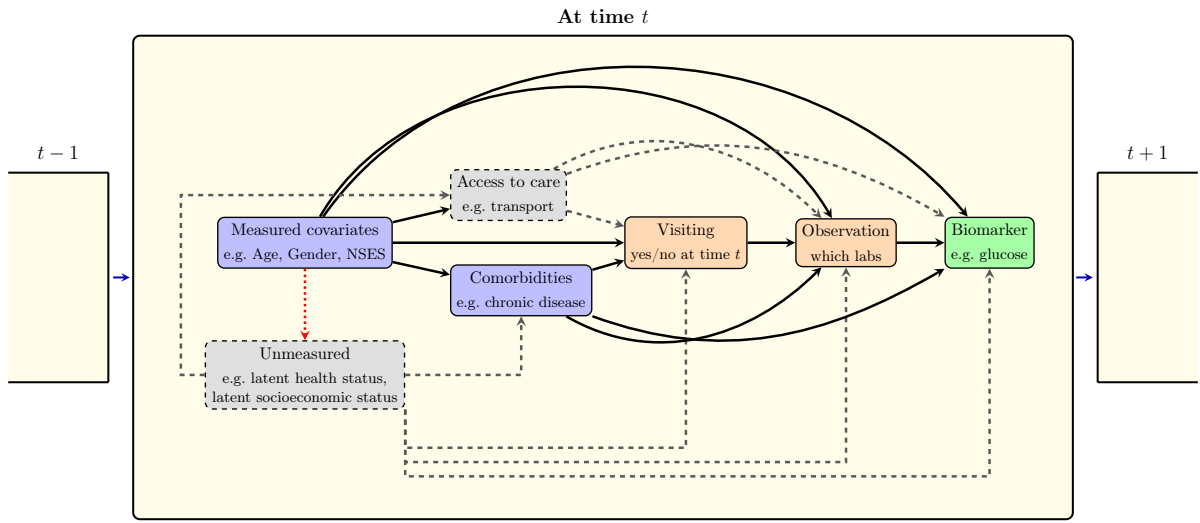
\begin{figure}[H]
\begin{center}
\resizebox{\textwidth}{!}{
\begin{tikzpicture}[
    node distance=0.8cm and 1.1cm,
    every node/.style={align=center, font=\small},
    obs/.style={draw, thick, rounded corners, minimum width=1.8cm, minimum height=0.7cm, fill=blue!25},
    latent/.style={draw, thick, rounded corners, minimum width=1.8cm, minimum height=0.7cm, fill=gray!25, dashed},
    process/.style={draw, thick, rounded corners, minimum width=1.8cm, minimum height=0.7cm, fill=orange!30},
    outcome/.style={draw, thick, rounded corners, minimum width=1.8cm, minimum height=0.7cm, fill=green!35},
    arr/.style={->, >=stealth, thick, line width=1.5pt},
    timelabel/.style={font=\large\bfseries}
]

\begin{scope}[xshift=4cm, local bounding box=timet]
    \node[obs] (Cov) {Measured covariates\\{\footnotesize \textnormal{e.g.\ Age, Gender, NSES}}};
    
    \node[latent, below=1.5cm of Cov, inner sep=4pt] (U) {Unmeasured\\{\footnotesize \textnormal{e.g.\ latent health status,}}\\[-3pt]{\footnotesize \textnormal{latent socioeconomic status}}};
    
    \node[latent, right=1.2cm of Cov, yshift=1cm] (Access) {Access to care\\{\footnotesize \textnormal{e.g.\ transport}}};
    \node[obs, right=1.2cm of Cov, yshift=-1cm] (Comorb) {Comorbidities\\{\footnotesize \textnormal{e.g.\ chronic disease}}};
    
    \node[process, right=1.2cm of Access, yshift=-1cm] (Visit) {Visiting\\{\footnotesize \textnormal{yes/no at time $t$}}};
    
    \node[process, right=1cm of Visit] (Obs) {Observation\\{\footnotesize \textnormal{which labs}}};
    
    \node[outcome, right=1cm of Obs] (Y) {Biomarker\\{\footnotesize \textnormal{e.g.\ glucose}}};
    
    \draw[arr] (Cov) -- (Access);
    \draw[arr] (Cov) -- (Comorb);
    \draw[arr] (Cov) -- (Visit);
    \draw[arr, dotted, red] (Cov) -- (U);  
    \draw[arr] (Cov) to[out=60, in=120] (Obs);
    \draw[arr] (Cov) to[out=55, in=130] (Y);
    
    \coordinate (bottomleft) at ($(U.south west)+(-0.3, -0.8)$);
    \coordinate (bottomright) at ($(Y.south east)+(0.3, -0.8)$);
    
    \draw[arr, dashed, gray!70!black] (U.west) -- ++(-0.5, 0) |- (Access.west);
    
    \draw[arr, dashed, gray!70!black] (U.east) -| (Comorb.south);
    
    \draw[arr, dashed, gray!70!black] (U.south east) -- ++(0, -0.8) -| (Visit.south);
    
    \draw[arr, dashed, gray!70!black] (U.south east) -- ++(0, -1.1) -| (Obs.south);
    
    \draw[arr, dashed, gray!70!black] (U.south east) -- ++(0, -1.4) -| (Y.south);
    
     \draw[arr, dashed, gray!70!black] (Access) to[out=30, in=135] (Obs);
    \draw[arr, dashed, gray!70!black] (Access) to[out=20, in=150] (Y);
    \draw[arr, dashed, gray!70!black] (Access) -- (Visit);
    \draw[arr] (Comorb) -- (Visit);
    \draw[arr] (Comorb) to[out=-30, in=-135] (Obs);
    \draw[arr] (Comorb) to[out=-20, in=-150] (Y);
    
    \draw[arr] (Visit) -- (Obs);
    
    \draw[arr] (Obs) -- (Y);
    
    \coordinate (topbound) at ($(Cov.north)+(0, 3.5)$);
    \coordinate (bottombound) at ($(U.south)+(0, -2.0)$);
    \coordinate (leftbound) at ($(U.west)+(-1.2, 0)$);
    \coordinate (rightbound) at ($(Y.east)+(0.5, 0)$);
\end{scope}

\begin{scope}[on background layer]
    \node[draw, very thick, rounded corners, inner sep=0.3cm, fit=(topbound)(bottombound)(leftbound)(rightbound), fill=yellow!10] (mainbox) {};
\end{scope}
\node[timelabel] at ($(mainbox.north)+(0, 0.4)$) {At time $t$};

\fill[yellow!10] ($(mainbox.west)+(-2.6, -2.2)$) rectangle ($(mainbox.west)+(-0.5, 2.2)$);
\draw[very thick] ($(mainbox.west)+(-0.5, -2.2)$) -- ($(mainbox.west)+(-0.5, 2.2)$);
\draw[very thick] ($(mainbox.west)+(-2.6, 2.2)$) -- ($(mainbox.west)+(-0.5, 2.2)$);
\draw[very thick] ($(mainbox.west)+(-2.6, -2.2)$) -- ($(mainbox.west)+(-0.5, -2.2)$);
\node[timelabel] at ($(mainbox.west)+(-1.55, 2.6)$) {$t-1$};

\fill[yellow!10] ($(mainbox.east)+(0.5, -2.2)$) rectangle ($(mainbox.east)+(2.6, 2.2)$);
\draw[very thick] ($(mainbox.east)+(0.5, -2.2)$) -- ($(mainbox.east)+(0.5, 2.2)$);
\draw[very thick] ($(mainbox.east)+(0.5, 2.2)$) -- ($(mainbox.east)+(2.6, 2.2)$);
\draw[very thick] ($(mainbox.east)+(0.5, -2.2)$) -- ($(mainbox.east)+(2.6, -2.2)$);
\node[timelabel] at ($(mainbox.east)+(1.55, 2.6)$) {$t+1$};

\draw[->, >=stealth, very thick, blue!70!black] ($(mainbox.west)+(-0.4,0)$) -- ($(mainbox.west)-(0.05,0)$);
\draw[->, >=stealth, very thick, blue!70!black] ($(mainbox.east)+(0.05,0)$) -- ($(mainbox.east)+(0.4,0)$);

\end{tikzpicture}
}
\end{center}
\caption{Directed acyclic graph (DAG) illustrating the informative visiting and observation processes in electronic health record data. Blue solid boxes represent measured covariates; gray dashed boxes represent unmeasured factors; orange boxes denote the visiting and observation processes; and the green box denotes the biomarker outcome. Solid arrows indicate pathways through observed variables; dashed arrows indicate pathways through unmeasured factors. The red dotted arrow represents a possible but non-identifiable dependence between measured covariates and the unmeasured factor. The panels at adjacent time points indicate that the same structure repeats across clinic visits.}
\label{fig:DAG}
\end{figure}

For metabolic and inflammatory markers, we observed an interpretable gradient across methods, most striking for HbA1c. Methods that ignore the observation process (Summary-Median, LMM) produced the weakest associations ($-0.123$ and $-0.113$); the Liang estimator \citep{liang2009joint}, which accounts only for the visiting process, yielded a negligible increase ($-0.120$); EHRJoint produced a moderate estimate ($-0.157$); and GIVEHR yielded the largest estimate ($\widehat{\beta}_{\mathrm{NSES}} = -0.253$; 95\% CI: $-0.333$, $-0.173$), corresponding to an approximate 22\% decrease in HbA1c per standard deviation increase in neighborhood income. This progression is consistent with the directed acyclic graph (DAG) in Figure~\ref{fig:DAG}: because the observation process is both informative and MNAR (Section~\ref{subsec:obs_data}), and because NSES jointly influences visiting, test ordering, and health status through the shared latent variable $U_i$, methods that ignore either process yield estimates that conflate the true NSES--biomarker association with the data collection mechanism. By jointly modeling the visiting and observation processes through the shared latent variable $U_i$, GIVEHR estimates the NSES--biomarker association while accounting for the data collection mechanism, whereas methods that omit one or both processes cannot. 
Significant negative associations were also identified by the proposed method for CRP. For WBC, the Liang estimator was the only approach that failed to detect a significant association ($-0.008$; 95\% CI: $-0.034$, $0.019$), likely because ignoring the observation mechanism distorts the visiting weights. In contrast, for the sex-specific cancer markers CA125 and PSA, no method---including GIVEHR---yielded a significant association. This is expected: these markers are ordered under specific clinical suspicion, so the observation mechanism is driven by disease status rather than NSES, leaving little distortion from the data collection process for any method to correct.

Figure~\ref{fig:forest_uninsured} presents results using the standardized neighborhood-level uninsured rate. As expected, coefficients were of opposite sign. The proposed method identified significant positive associations for HbA1c ($\widehat{\beta}_{\mathrm{NSES}} = 0.146$; 95\% CI: $0.037$, $0.254$), CRP ($0.187$; 95\% CI: $0.135$, $0.239$), and WBC ($0.042$; 95\% CI: $0.014$, $0.070$), reinforcing that socioeconomic disadvantage---here captured through the access pathway---is associated with elevated metabolic and inflammatory biomarker levels after accounting for informative visiting and observation. Notably, EHRJoint failed to detect significant associations for these biomarkers under this specification, and in some instances produced estimates whose signs were inconsistent with the expected direction based on the income specification, suggesting that modeling the observation process without a shared latent structure is insufficient to correct the bias.

\begin{figure}
    \centering
    \includegraphics[width=1.0\linewidth]{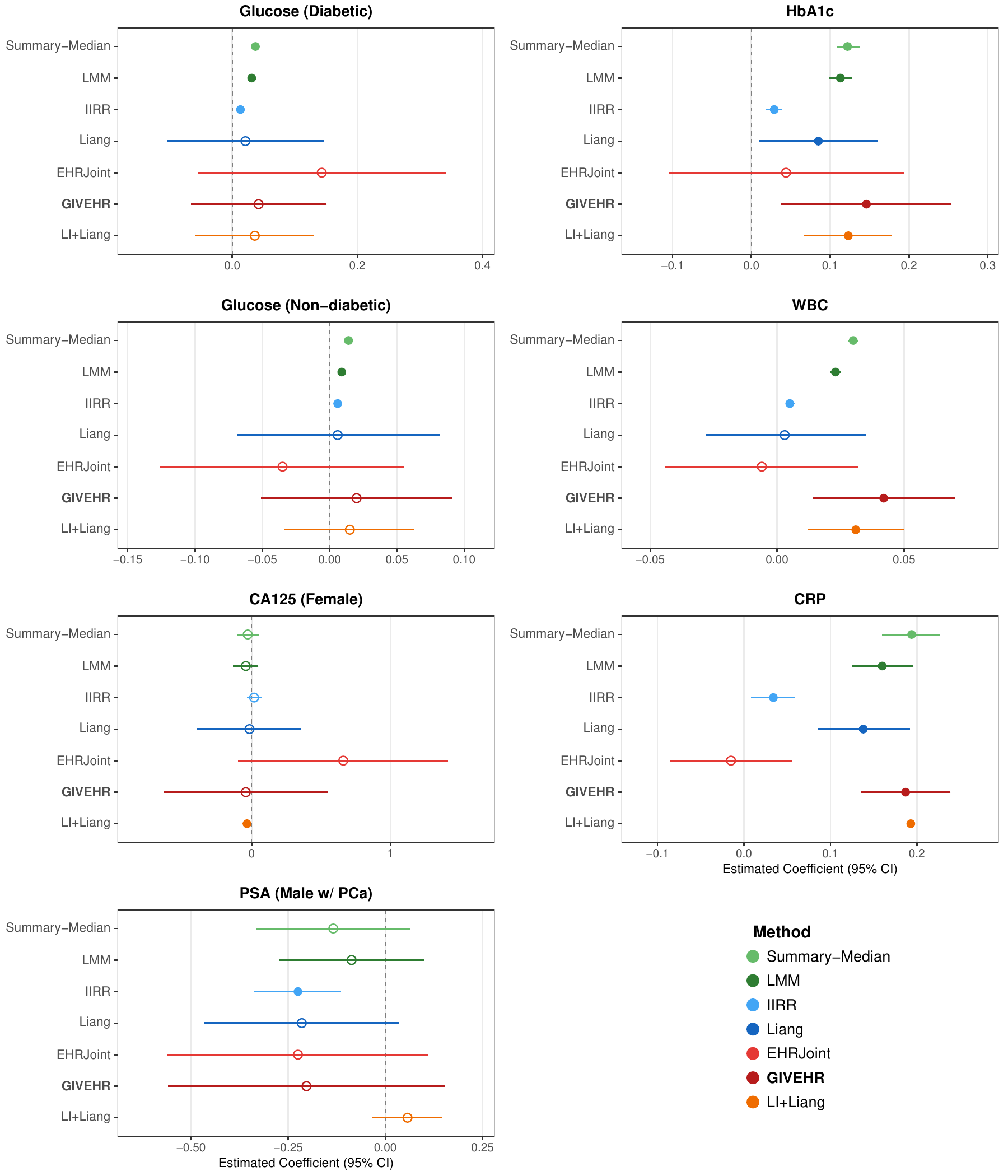}
    \caption{Estimated regression coefficients and 95\% confidence intervals for the association between log-transformed biomarker levels and standardized neighborhood-level uninsured rate, comparing seven estimation methods. The left column displays results for glucose (diabetic and non-diabetic subgroups), CA125 (females), and PSA (males with prostate cancer); the right column displays results for HbA1c, WBC, and CRP. Filled circles denote statistical significance at the 0.05 level; open circles denote non-significance.
    }
    \label{fig:forest_uninsured}
\end{figure}

In summary, neighborhood socioeconomic disadvantage is associated with meaningfully worse metabolic and inflammatory biomarker profiles, but this association is obscured in raw EHR data because NSES simultaneously influences which patients visit and which tests are ordered, confounding the biomarker trajectory with the data collection mechanism (Figure~\ref{fig:DAG}). 
Moreover, because measured covariates may influence the latent variable (the red dotted arrow in Figure~\ref{fig:DAG}), this dependence is non-identifiable but is absorbed into the fixed-effect coefficients, so the estimated coefficients capture the total association between NSES and the biomarker (see Section~\ref{sec:supp:invariance} in Supplementary Materials ). 
Correcting for both informative processes approximately doubles the estimated NSES effect for HbA1c relative to na\"ive approaches. The correction matters most when the observation process is socioeconomically patterned; for cancer markers driven by clinical suspicion rather than NSES, no method detects a significant association, providing a built-in specificity check, though the smaller sample sizes for these markers limit statistical power. These findings indicate that jointly accounting for informative visiting and observation through a shared latent variable can yield substantially different inferences about population-level associations from EHR data.

\section{Discussion}\label{sec:discussion}

Two extensions warrant future investigation. First, existing methods for modeling longitudinal biomarkers typically treat irregular visit times as a single phenomenon without distinguishing between visiting and observation. In the deep learning literature, strategies such as regularizing trajectories onto a fixed time grid \citep{dwarampudi2019effects, xie2022deep} or incorporating time intervals directly into network architectures \citep{che2018recurrent, wang2019mcpl} implicitly assume that biomarkers are observed at every encounter or restrict analysis to visits with observed values. Extending these approaches to explicitly incorporate both processes, for instance by embedding a shared latent structure within predictive models, represents an important direction. More broadly, while our framework can accommodate dependence on past observation history (Section~\ref{sec:tv_obs}), the associated computational cost limits its scalability to large EHR databases, and efficient algorithms for this extension remain an open problem.

A second direction is to extend the joint modeling framework to incorporate time-to-event outcomes, allowing analyses of disease progression or mortality that properly account for informative visiting and observation. Such extensions align naturally with the broader joint longitudinal--survival literature \citep{wang2025joint}. Recent methodological work has begun to jointly model longitudinal markers, visiting processes, and time-to-event outcomes, including competing risks \citep{thomadakis2025shared}, but typically does not include a distinct observation process conditional on a visit. Extending the proposed framework to incorporate survival outcomes while retaining an explicit observation layer would fill an important gap for applications where both visit timing and observation decisions are informative.

\bibliographystyle{chicago}
\bibliography{EHR}

@book{andersen2012statistical,
  title={Statistical models based on counting processes},
  author={Andersen, Per K and Borgan, Ornulf and Gill, Richard D and Keiding, Niels},
  year={2012},
  publisher={Springer Science \& Business Media}
}

@article{anthopolos2021modeling,
  title={Modeling heterogeneity and missing data of multiple longitudinal outcomes in electronic health records},
  author={Anthopolos, Rebecca and Wei, Ying and Chen, Qixuan},
  journal={arXiv preprint arXiv:2103.11170},
  year={2021}
}

@article{goldstein2016controlling,
  title={Controlling for informed presence bias due to the number of health encounters in an electronic health record},
  author={Goldstein, Benjamin A and Bhavsar, Nrupen A and Phelan, Matthew and Pencina, Michael J},
  journal={American journal of epidemiology},
  volume={184},
  number={11},
  pages={847--855},
  year={2016},
  publisher={Oxford University Press}
}

@article{hripcsak2013next,
  title={Next-generation phenotyping of electronic health records},
  author={Hripcsak, George and Albers, David J},
  journal={Journal of the American Medical Informatics Association},
  volume={20},
  number={1},
  pages={117--121},
  year={2013},
  publisher={Oxford Academic}
}

@article{all2019all,
  author  = {{All of Us Research Program Investigators}},
  title   = {The ``All of Us'' Research Program},
  journal = {New England Journal of Medicine},
  volume  = {381},
  number  = {7},
  pages   = {668--676},
  year    = {2019}
}

@article{bycroft2018uk,
  title={The UK Biobank resource with deep phenotyping and genomic data},
  author={Bycroft, Clare and Freeman, Colin and Petkova, Desislava and Band, Gavin and Elliott, Lloyd T and Sharp, Kevin and Motyer, Allan and Vukcevic, Damjan and Delaneau, Olivier and O’Connell, Jared and others},
  journal={Nature},
  volume={562},
  number={7726},
  pages={203--209},
  year={2018},
  publisher={Nature Publishing Group UK London}
}

@article{lin2004analysis,
  title={Analysis of longitudinal data with irregular, outcome-dependent follow-up},
  author={Lin, Haiqun and Scharfstein, Daniel O and Rosenheck, Robert A},
  journal={Journal of the Royal Statistical Society Series B: Statistical Methodology},
  volume={66},
  number={3},
  pages={791--813},
  year={2004},
  publisher={Oxford University Press}
}

@article{pullenayegum2016longitudinal,
  title={Longitudinal data subject to irregular observation: A review of methods with a focus on visit processes, assumptions, and study design},
  author={Pullenayegum, Eleanor M and Lim, Lily SH},
  journal={Statistical methods in medical research},
  volume={25},
  number={6},
  pages={2992--3014},
  year={2016},
  publisher={SAGE Publications Sage UK: London, England}
}

@article{sun2007regression,
  title={Regression analysis of longitudinal data in the presence of informative observation and censoring times},
  author={Sun, Jianguo and Sun, Liuquan and Liu, Dandan},
  journal={Journal of the American Statistical Association},
  volume={102},
  number={480},
  pages={1397--1406},
  year={2007},
  publisher={Taylor \& Francis}
}

@misc{robins2000marginal,
  title={Marginal structural models and causal inference in epidemiology},
  author={Robins, James M and Hernan, Miguel Angel and Brumback, Babette},
  journal={Epidemiology},
  volume={11},
  number={5},
  pages={550--560},
  year={2000},
  publisher={Lww}
}

@article{weaver2023functional,
  title={Functional data analysis for longitudinal data with informative observation times},
  author={Weaver, Caleb and Xiao, Luo and Lu, Wenbin},
  journal={Biometrics},
  volume={79},
  number={2},
  pages={722--733},
  year={2023},
  publisher={Wiley Online Library}
}

@article{wells2013strategies,
  title={Strategies for handling missing data in electronic health record derived data},
  author={Wells, Brian J and Chagin, Kevin M and Nowacki, Amy S and Kattan, Michael W},
  journal={Egems},
  volume={1},
  number={3},
  pages={1035},
  year={2013}
}

@book{cook2007statistical,
  title={The statistical analysis of recurrent events},
  author={Cook, Richard J and Lawless, Jerald F and others},
  year={2007},
  publisher={Springer}
}

@article{andersen1982cox,
  title={Cox's regression model for counting processes: a large sample study},
  author={Andersen, Per Kragh and Gill, Richard D},
  journal={The annals of statistics},
  pages={1100--1120},
  year={1982},
  publisher={JSTOR}
}

@article{haneuse2016general,
  title={A general framework for considering selection bias in EHR-based studies: what data are observed and why?},
  author={Haneuse, Sebastien and Daniels, Michael},
  journal={EGEMs},
  volume={4},
  number={1},
  pages={1203},
  year={2016}
}

@book{aalen2008survival,
  title={Survival and event history analysis: a process point of view},
  author={Aalen, Odd O and Borgan, {\O}rnulf and Gjessing, H{\aa}kon K},
  year={2008},
  publisher={Springer}
}

@article{ogden2021error,
  title={On the error in Laplace approximations of high-dimensional integrals},
  author={Ogden, Helen},
  journal={Stat},
  volume={10},
  number={1},
  pages={e380},
  year={2021},
  publisher={Wiley Online Library}
}

@incollection{van1996weak,
  title={Weak convergence},
  author={Van Der Vaart, Aad W and Wellner, Jon A},
  booktitle={Weak convergence and empirical processes: with applications to statistics},
  pages={16--28},
  year={1996},
  publisher={Springer}
}

@article{lin2001semiparametric,
  title={Semiparametric and nonparametric regression analysis of longitudinal data},
  author={Lin, Danyu Y and Ying, Zhiliang},
  journal={Journal of the American Statistical Association},
  volume={96},
  number={453},
  pages={103--126},
  year={2001},
  publisher={Taylor \& Francis}
}

@article{liang2009joint,
  title={Joint modeling and analysis of longitudinal data with informative observation times},
  author={Liang, Yu and Lu, Wenbin and Ying, Zhiliang},
  journal={Biometrics},
  volume={65},
  number={2},
  pages={377--384},
  year={2009},
  publisher={Oxford University Press}
}

@article{dai2018joint,
  title={Joint modelling of survival and longitudinal data with informative observation times},
  author={Dai, Hongsheng and Pan, Jianxin},
  journal={Scandinavian Journal of Statistics},
  volume={45},
  number={3},
  pages={571--589},
  year={2018},
  publisher={Wiley Online Library}
}

@article{burvzkova2007longitudinal,
  title={Longitudinal data analysis for generalized linear models with follow-up dependent on outcome-related variables},
  author={Burvzkova, Petra and Lumley, Thomas},
  journal={Canadian Journal of Statistics},
  volume={35},
  number={4},
  pages={485--500},
  year={2007},
  publisher={Wiley Online Library}
}

@article{yiu2025accommodating,
  title={Accommodating informative visit times for analysing irregular longitudinal data: a sensitivity analysis approach with balancing weights estimators},
  author={Yiu, Sean and Su, Li},
  journal={Journal of the Royal Statistical Society Series C: Applied Statistics},
  volume={74},
  number={3},
  pages={824--843},
  year={2025},
  publisher={Oxford University Press UK}
}

@article{chen2015regression,
  title={Regression analysis of longitudinal data with irregular and informative observation times},
  author={Chen, Yong and Ning, Jing and Cai, Chunyan},
  journal={Biostatistics},
  volume={16},
  number={4},
  pages={727--739},
  year={2015},
  publisher={Oxford University Press}
}

@article{shen2019regression,
  title={Regression analysis of longitudinal data with outcome-dependent sampling and informative censoring},
  author={Shen, Weining and Liu, Suyu and Chen, Yong and Ning, Jing},
  journal={Scandinavian Journal of Statistics},
  volume={46},
  number={3},
  pages={831--847},
  year={2019},
  publisher={Wiley Online Library}
}

@misc{du2025newstatisticalapproachjoint,
      title={A new statistical approach for joint modeling of longitudinal outcomes measured in electronic health records with clinically informative presence and observation processes}, 
      author={Jiacong Du and Xu Shi and Bhramar Mukherjee},
      year={2025},
      eprint={2410.13113},
      archivePrefix={arXiv},
      primaryClass={stat.ME},
      url={https://arxiv.org/abs/2410.13113}, 
}

@article{aalen2010dynamic,
  title={A dynamic approach for reconstructing missing longitudinal data using the linear increments model},
  author={Aalen, Odd O and Gunnes, Nina},
  journal={Biostatistics},
  volume={11},
  number={3},
  pages={453--472},
  year={2010},
  publisher={Oxford University Press}
}

@article{laird1982random,
  title={Random-effects models for longitudinal data},
  author={Laird, Nan M and Ware, James H},
  journal={Biometrics},
  pages={963--974},
  year={1982},
  publisher={JSTOR}
}

@article{thomadakis2025shared,
  title={Shared parameter modeling of longitudinal data allowing for possibly informative visiting process and terminal event},
  author={Thomadakis, Christos and Meligkotsidou, Loukia and Pantazis, Nikos and Touloumi, Giota},
  journal={Biostatistics},
  volume={26},
  number={1},
  pages={kxae041},
  year={2025},
  publisher={Oxford University Press}
}

@article{wang2025joint,
  title={Joint modeling of longitudinal and survival data},
  author={Wang, Jane-Ling and Zhong, Qixian},
  journal={Annual Review of Statistics and Its Application},
  volume={12},
  number={1},
  pages={449--476},
  year={2025},
  publisher={Annual Reviews}
}

@article{xie2022deep,
  title={Deep learning for temporal data representation in electronic health records: A systematic review of challenges and methodologies},
  author={Xie, Feng and Yuan, Han and Ning, Yilin and Ong, Marcus Eng Hock and Feng, Mengling and Hsu, Wynne and Chakraborty, Bibhas and Liu, Nan},
  journal={Journal of biomedical informatics},
  volume={126},
  pages={103980},
  year={2022},
  publisher={Elsevier}
}

@article{dwarampudi2019effects,
  title={Effects of padding on LSTMs and CNNs},
  author={Dwarampudi, Mahidhar and Reddy, NV},
  journal={arXiv preprint arXiv:1903.07288},
  year={2019}
}

@article{che2018recurrent,
  title={Recurrent neural networks for multivariate time series with missing values},
  author={Che, Zhengping and Purushotham, Sanjay and Cho, Kyunghyun and Sontag, David and Liu, Yan},
  journal={Scientific reports},
  volume={8},
  number={1},
  pages={6085},
  year={2018},
  publisher={Nature Publishing Group UK London}
}

@article{wang2019mcpl,
  title={MCPL-Based FT-LSTM: medical representation learning-based clinical prediction model for time series events},
  author={Wang, Lutong and Wang, Hong and Song, Yongqiang and Wang, Qian},
  journal={IEEE Access},
  volume={7},
  pages={70253--70264},
  year={2019},
  publisher={IEEE}
}

@book{Rubin1987MI,
author = {Rubin, D.},
title= {Multiple imputation for nonresponse in surveys},
year = {1987},
doi = {10.1002/9780470316696},
address = {New York},
publisher = {John Wiley \& Sons},
}

@article{gasparini2020mixed,
  title={Mixed-effects models for health care longitudinal data with an informative visiting process: A Monte Carlo simulation study},
  author={Gasparini, Alessandro and Abrams, Keith R and Barrett, Jessica K and Major, Rupert W and Sweeting, Michael J and Brunskill, Nigel J and Crowther, Michael J},
  journal={Statistica Neerlandica},
  volume={74},
  number={1},
  pages={5--23},
  year={2020},
  publisher={Wiley Online Library}
}

@article{breslow1974covariance,
  title={Covariance analysis of censored survival data},
  author={Breslow, Norman},
  journal={Biometrics},
  pages={89--99},
  year={1974},
  publisher={JSTOR}
}

@article{tierney1986accurate,
  title={Accurate approximations for posterior moments and marginal densities},
  author={Tierney, Luke and Kadane, Joseph B},
  journal={Journal of the american statistical association},
  volume={81},
  number={393},
  pages={82--86},
  year={1986},
  publisher={Taylor \& Francis}
}

@article{obermeyer2019dissecting,
  title={Dissecting racial bias in an algorithm used to manage the health of populations},
  author={Obermeyer, Ziad and Powers, Brian and Vogeli, Christine and Mullainathan, Sendhil},
  journal={Science},
  volume={366},
  number={6464},
  pages={447--453},
  year={2019},
  publisher={American Association for the Advancement of Science}
}

@article{ellenbogen2024race,
  title={Race and ethnicity and diagnostic testing for common conditions in the acute care setting},
  author={Ellenbogen, Michael I and Weygandt, P Logan and Newman-Toker, David E and Anderson, Andrew and Rim, Nayoung and Brotman, Daniel J},
  journal={JAMA Network Open},
  volume={7},
  number={8},
  pages={e2430306--e2430306},
  year={2024},
  publisher={American Medical Association}
}

@article{goldstein2020labwas,
  title={LabWAS: Novel findings and study design recommendations from a meta-analysis of clinical labs in two independent biobanks},
  author={Goldstein, Jeffery A and Weinstock, Joshua S and Bastarache, Lisa A and Larach, Daniel B and Fritsche, Lars G and Schmidt, Ellen M and Brummett, Chad M and Kheterpal, Sachin and Abecasis, Goncalo R and Denny, Joshua C and others},
  journal={PLoS genetics},
  volume={16},
  number={11},
  pages={e1009077},
  year={2020},
  publisher={Public Library of Science San Francisco, CA USA}
}

@article{dennis2021clinical,
  title={Clinical laboratory test-wide association scan of polygenic scores identifies biomarkers of complex disease},
  author={Dennis, Jessica K and Sealock, Julia M and Straub, Peter and Lee, Younga H and Hucks, Donald and Actkins, Ky’Era and Faucon, Annika and Feng, Yen-Chen Anne and Ge, Tian and Goleva, Slavina B and others},
  journal={Genome medicine},
  volume={13},
  number={1},
  pages={6},
  year={2021},
  publisher={Springer}
}


\section*{Supplementary Material}

\renewcommand{\thesection}{S\arabic{section}}
\renewcommand{\thetable}{S\arabic{table}}
\renewcommand{\thefigure}{S\arabic{figure}}
\setcounter{section}{0}
\setcounter{table}{0}
\setcounter{figure}{0}

\setcounter{page}{1}
\renewcommand{\thepage}{S\arabic{page}}

\setcounter{equation}{0}
\renewcommand{\theequation}{S\arabic{equation}}

\renewcommand{\thelemma}{S\arabic{lemma}}
\setcounter{lemma}{0}

\section{Notation and Model Recap}\label{sec:notation and model recap}

To facilitate the reading of the proofs, we summarize the main notation used in this Supplementary Material.

\begin{longtable}{l p{0.70\textwidth}}
\caption{Summary of notation used in the Supplementary Material.}\label{tab:notation}\\
\hline
\textbf{Symbol} & \textbf{Description} \\
\hline
\endfirsthead

\caption[]{(continued)}\\
\hline
\textbf{Symbol} & \textbf{Description} \\
\hline
\endhead

\hline
\endfoot

\multicolumn{2}{l}{\textit{Visiting Process}} \\
$N_i(t)$ & Counting process for visits of subject $i$ up to time $t$. \\
$C_i$ & Censoring time for subject $i$. \\
$m_i = N_i(C_i)$ & Total number of visits for subject $i$ before censoring. \\
$\lambda_0(t)$, $\Lambda_0(t)$ & Baseline intensity and cumulative baseline intensity. \\
$U_i$ & Shared latent factor, $U_i \sim \mathcal{N}(0,1)$. \\
$\eta_i$ & Multiplicative frailty, $\eta_i = \exp(\mu_0 + \sigma U_i)$ with $\mu_0 = -\sigma^2/2$. \\
$\bm\gamma$ & Regression coefficients for the visiting process. \\
$\nu_i$ & Expected cumulative intensity, $\nu_i = \exp(\bm\gamma^\top \bm{X}_i^{\mathcal V}) \Lambda_0(C_i)$. \\
$(\mu_{U_i}, s_{U_i}^2)$ & EB posterior mean and variance of $U_i \mid (C_i, m_i)$ via Laplace approximation. \\
\hline
\multicolumn{2}{l}{\textit{Observation Process}} \\
$R_i^{\mathcal Y}(t)$ & Indicator that outcome is recorded at visit time $t$. \\
$\omega_i(t; u)$ & Conditional probability of observation given $U_i = U$. \\
$\overline{\omega}_i(t)$ & Marginal observation probability, $\E\{\omega_i(t; U_i) \mid C_i, m_i\}$. \\
$k_i(t)$ & Aggregated probit argument for $\overline{\omega}_i(t) = \Phi(k_i(t))$. \\
$\kappa_i(t)$ & Observation-weighted posterior mean ratio, $\E\{U_i \omega_i(t; U_i) \mid C_i, m_i\}/\overline{\omega}_i(t)$. \\
$(\bm\alpha, \bm\delta, \bm\Sigma_q)$ & Fixed effects, latent factor loadings, and random slope covariance for observation. \\
\hline
\multicolumn{2}{l}{\textit{Outcome Process}} \\
$Y_i(t)$ & Longitudinal outcome for subject $i$ at time $t$. \\
$\bm{X}_i^{\mathcal V}$, $\bm{X}_i^{\mathcal O}(t)$, $\bm{X}_i^{\mathcal Y}(t)$ & Covariates for visiting, observation, and outcome processes. \\
$\bm{Z}_i^{\mathcal O}(t)$, $\bm{Z}_i^{\mathcal Y}(t)$ & Covariates for random effects in observation and outcome models. \\
$\bm{b}_i$ & Random effects in the outcome model, $\bm{b}_i \mid U_i \sim \mathcal{N}(\bm\theta U_i, \bm\Sigma_b)$. \\
$\widetilde{\bm{b}}_i$ & Residual random effect, $\bm{b}_i = \bm\theta U_i + \widetilde{\bm{b}}_i$ with $\widetilde{\bm{b}}_i \sim \mathcal{N}(\bm{0}, \bm\Sigma_b)$. \\
$\beta_0(t)$ & Time-varying baseline function in the outcome model. \\
$\varepsilon_i(t)$ & Measurement error, $\varepsilon_i(t) \sim \mathcal{N}(0, \sigma_\varepsilon^2)$. \\
$\mathcal{A}(t)$ & Cumulative baseline function, $\mathcal{A}(t) = \int_0^t \beta_0(s)\,\dint\Lambda_0(s)$. \\
$\psi = (\bm\beta^\top, \bm\theta^\top)^\top$ & Parameters of interest in the outcome model. \\
\hline
\multicolumn{2}{l}{\textit{Nuisance Parameters}} \\
$\Omega_{\mathcal{V}}$ & Visiting process parameters, $\Omega_{\mathcal{V}} = (\bm\gamma, \Lambda_0(\cdot), \sigma^2)$. \\
$\Omega_{\mathcal{O}}$ & Observation process parameters, $\Omega_{\mathcal{O}} = (\bm\alpha, \bm\delta, \mathrm{vech}(\bm\Sigma_q))$. \\
$\Omega_{\mathrm{EB}}$ & EB posterior parameters, $\Omega_{\mathrm{EB}} = \{(\mu_{U_i}, s_{U_i}^2)\}_{i=1}^n$. \\
$\Omega$ & All nuisance parameters, $\Omega = (\Omega_{\mathcal{V}}, \Omega_{\mathcal{O}}, \Omega_{\mathrm{EB}})$. \\
$\mathcal{G}_{\mathcal{V}}$ & Compact parameter space for $\Omega_{\mathcal{V}}$. \\
\hline
\multicolumn{2}{l}{\textit{Estimating Equation Components}} \\
$\bm{B}_i(t)$ & Compensation covariate, $\bm{B}_i(t) = \kappa_i(t) \bm{Z}_i^{\mathcal Y}(t)$. \\
$p_i(t)$ & Weight function, $p_i(t) = \overline{\omega}_i(t)\,\mathbf{1}(t \leq C_i)\,m_i/\Lambda_0(C_i)$. \\
$\bm{V}_i(t)$ & Uncentered design vector, $\bm{V}_i(t) = (\bm{X}_i^{\mathcal Y}(t)^\top, \bm{B}_i(t)^\top)^\top$. \\
$\overline{\bm{V}}(t)$ & Weighted average, $\overline{\bm{V}}(t) = \sum_j \bm{V}_j(t) p_j(t) / \sum_j p_j(t)$. \\
$\bm{H}_i(t)$ & Centered design vector, $\bm{H}_i(t) = \bm{V}_i(t) - \overline{\bm{V}}(t)$. \\
$M_i(t)$ & Compensated process (mean-zero given $(C_i, m_i)$). \\
\hline
\end{longtable}

\paragraph{Model recap.}
For completeness, we restate the three-component model:

\begin{enumerate}[label=(\roman*), leftmargin=2em]
\item \textbf{Visiting process:} The intensity for subject $i$ is
\[
\lambda_i(t) = \eta_i \exp(\bm\gamma^\top X_i^{\mathcal V}) \lambda_0(t), \quad \text{where } \eta_i = \exp\!\big(\mu_0 + \sigma U_i\big).
\]
Conditional on $(U_i, X_i^{\mathcal V}, C_i)$, the counting process $N_i(\cdot)$ is a non-homogeneous Poisson process (NHPP) with cumulative intensity $\eta_i \exp(\bm\gamma^\top X_i^{\mathcal V}) \Lambda_0(t)$.

\item \textbf{Observation process:} At each visit, the probability that the outcome is recorded is
\[
\omega_i(t; u) = \Phi\!\left\{ \frac{\bm\alpha^\top \bm{X}_i^{\mathcal O} + (\bm\delta^\top \bm{Z}_i^{\mathcal O}) u}{\sqrt{1 + (\bm{Z}_i^{\mathcal O})^\top \bm\Sigma_q \bm{Z}_i^{\mathcal O}}}\right\}.
\]
Here, the explicit random slopes $\bm{q}_i$ (where $\bm{q}_i \mid U_i=u \sim \mathcal{N}(\bm\delta u, \bm\Sigma_q)$) have been integrated out to obtain the conditional probability given $U_i=u$.

\item \textbf{Outcome model:} The longitudinal outcome follows
\[
Y_i(t) = \beta_0(t) + \bm\beta^\top \bm{X}_i^{\mathcal Y}(t) + \bm{b}_i^\top \bm{Z}_i^{\mathcal Y}(t) + \varepsilon_i(t),
\]
where $\bm{b}_i \mid U_i \sim \mathcal{N}(\bm\theta U_i, \bm\Sigma_b)$. We decompose $\bm{b}_i$ as $\bm{b}_i = \bm\theta U_i + \widetilde{\bm{b}}_i$, where $\widetilde{\bm{b}}_i \sim \mathcal{N}(\bm{0}, \bm\Sigma_b)$ represents the residual outcome random effect, which is independent of the observation process random effects given $U_i$ (by Assumption 2).
\end{enumerate}

\section{Preliminary Lemmas}\label{sec:lemmas}

This section establishes the key probabilistic identities used throughout the proofs.

\begin{lemma}[NHPP order-statistics identity]\label{lem:order}
Let $N_i(t)$ be a non-homogeneous Poisson process on $[0,\tau]$ with intensity $\lambda_0(t)$ and cumulative intensity $\Lambda_0(t) = \int_0^t \lambda_0(s)\,ds$. Let $C_i \le \tau$ be a (possibly random) censoring time, and let $m_i = N_i(C_i)$ be the total number of jumps before censoring. Conditionally on $(C_i, m_i)$, the jump times of $N_i$ in $[0, C_i]$ have the same joint distribution as the order statistics of $m_i$ i.i.d.\ random variables with density $\lambda_0(t)/\Lambda_0(C_i)$ on $[0, C_i]$. In particular,
\begin{equation}\label{eq:order-stats}
\E\!\left\{ \dint N_i(t) \mid C_i, m_i\right\} = \mathbf{1}(t \le C_i)\,m_i\,\frac{\dint\Lambda_0(t)}{\Lambda_0(C_i)} .
\end{equation}
\end{lemma}

\begin{proof}
This is a standard property of non-homogeneous Poisson processes; see, for example, \citet[Chapter~5]{aalen2008survival}. Conditionally on $(C_i, m_i)$, the joint density of the $m_i$ jump times $(T_{i1}, \ldots, T_{im_i})$ is proportional to
\[
\prod_{j=1}^{m_i} \lambda_0(T_{ij})\,\exp\!\big\{-\Lambda_0(C_i)\big\} \propto \prod_{j=1}^{m_i} \lambda_0(T_{ij}),
\]
restricted to $0 < T_{i1} < \cdots < T_{im_i} < C_i$. This is exactly the joint density of the order statistics of $m_i$ i.i.d.\ draws from the density $\lambda_0(t)/\Lambda_0(C_i)$ on $[0, C_i]$. Integrating out all but one time gives the marginal expectation:
\[
\E\!\left\{ \dint N_i(t) \mid C_i, m_i\right\} = m_i \cdot \frac{\lambda_0(t)\,dt}{\Lambda_0(C_i)} \cdot \mathbf{1}\{t \le C_i\} = \mathbf{1}(t \leq C_i)\,m_i\,\frac{\dint\Lambda_0(t)}{\Lambda_0(C_i)}. 
\]
\end{proof}

\begin{remark}[Independence from the frailty]\label{rem:frailty-cancel}
An important consequence of Lemma~\ref{lem:order} is that, conditional on $(C_i, m_i)$, the distribution of visit times does \emph{not} depend on the frailty $\eta_i$ or the covariates $X_i^{\mathcal V}$. Although these factors affect the \emph{number} of visits $m_i$, once $m_i$ is observed, the \emph{timing} of visits follows a universal distribution determined only by the baseline intensity $\lambda_0(t)$.
\end{remark}

\begin{lemma}[Martingale compensation]\label{lem:compensation}
Let $M_i(t)$ be a (local) martingale with respect to a filtration $\{\mathcal{F}_t\}$, and let $H_i(t)$ be a predictable and integrable process. Then the stochastic integral $\int_0^\cdot H_i(s)\,\dint M_i(s)$ is also a (local) martingale and, in particular,
\[
\E\!\left\{ \int_0^\tau H_i(t)\,\dint M_i(t)\right\} = 0,
\]
whenever the expectation exists.
\end{lemma}

\begin{proof}
This is a standard result in martingale theory \citep{lin2001semiparametric}. Since $H_i$ is predictable and integrable with respect to $M_i$, the stochastic integral is a (local) martingale with mean zero at each fixed time $t$. Taking $t = \tau$ yields the stated equality.
\end{proof}

\begin{lemma}[Probit-normal convolution]\label{lem:probit}
Let $X \sim \mathcal{N}(\mu_X, \sigma_X^2)$ and $U \sim \mathcal{N}(0, 1)$ be independent. For constants $a, b \in \mathbb{R}$ and $d > 0$, define $Y = a + bX + dU$. Then $Y \sim \mathcal{N}(a + b\mu_X, b^2\sigma_X^2 + d^2)$. Let
\[
k = \frac{a + b\mu_X}{\sqrt{b^2\sigma_X^2 + d^2}}.
\]
Then
\begin{align}
\E\!\left\{ \Phi\!\left(\frac{a + bX}{d}\right)\right\} &= \Phi(k), \label{eq:probit1}\\[4pt]
\E\!\left\{ X\,\Phi\!\left(\frac{a + bX}{d}\right)\right\} &= \mu_X\,\Phi(k) + \frac{b\sigma_X^2}{\sqrt{b^2\sigma_X^2 + d^2}}\,\phi(k), \label{eq:probit2}\\[4pt]
\frac{\E\!\left\{ X\,\Phi\!\left(\frac{a + bX}{d}\right)\right\} }{\E\!\left\{ \Phi\!\left(\frac{a + bX}{d}\right)\right\} } &= \mu_X + \frac{b\sigma_X^2}{\sqrt{b^2\sigma_X^2 + d^2}}\,\frac{\phi(k)}{\Phi(k)}. \label{eq:probit3}
\end{align}
\end{lemma}

\begin{proof}
Note that $\Phi((a + bX)/d) = \pr\{U \le (a + bX)/d \mid X\}$. By the law of total expectation,
\[
\E\!\left\{\Phi\!\left(\frac{a + bX}{d}\right)\right\} = \E\!\left\{\pr\!\left(U \le \frac{a + bX}{d} \,\middle|\, X\right)\right\} = \pr (dU \le a + bX) = \pr(Y \ge 0).
\]
Since $Y$ is normal with mean $a + b\mu_X$ and variance $b^2\sigma_X^2 + d^2$, standardization yields $\pr(Y \ge 0) = \pr(Z \ge -k) = \Phi(k)$ for $Z \sim \mathcal{N}(0,1)$, proving \eqref{eq:probit1}.

For \eqref{eq:probit2}, write $\E\{ X\,\Phi((a + bX)/d)\} = \E \{ X\,\mathbf{1}\{Y \ge 0\}\}$. Because $(X, Y)$ is jointly normal, we can express $X$ in terms of the standardized variable $Z = \left( Y - \E(Y) \right)/\sqrt{\Var(Y)}$:
\[
\E (X \mid Z = z) = \mu_X + \Cov(X, Z)\,z.
\]
Computing the covariance:
\[
Z = \frac{b(X - \mu_X) + dU}{\sqrt{b^2\sigma_X^2 + d^2}}, \quad \text{so} \quad \Cov(X, Z) = \frac{b\sigma_X^2}{\sqrt{b^2\sigma_X^2 + d^2}}.
\]
Now integrate over $z$ in the region where $Y \ge 0$, i.e., $Z \ge -k$:
\begin{align*}
\E \{ X\,\mathbf{1}\{Y \ge 0\} \} &= \int_{-k}^\infty \E ( X \mid Z = z ) \,\phi(z)\,dz \\
&= \int_{-k}^\infty \left( \mu_X + \frac{b\sigma_X^2}{\sqrt{b^2\sigma_X^2 + d^2}}\,z\right) \phi(z)\,dz \\
&= \mu_X \int_{-k}^\infty \phi(z)\,dz + \frac{b\sigma_X^2}{\sqrt{b^2\sigma_X^2 + d^2}} \int_{-k}^\infty z\,\phi(z)\,dz.
\end{align*}
Using $\int_{-k}^\infty \phi(z)\,dz = \Phi(k)$ and $\int_{-k}^\infty z\,\phi(z)\,dz = \phi(k)$, we obtain \eqref{eq:probit2}. Dividing \eqref{eq:probit2} by \eqref{eq:probit1} yields \eqref{eq:probit3}.
\end{proof}

\begin{corollary}[Closed-form approximation for $\kappa_i(t)$]\label{cor:kappa}
Under the model specification, the observation-weighted posterior mean ratio $\kappa_i(t)$ has the closed form:
\[
\kappa_i(t) = \frac{\E\{U_i\,\omega_i(t; U_i) \mid C_i, m_i\}}{\E\{\omega_i(t; U_i) \mid C_i, m_i\}} = \mu_{U_i} + \frac{(\bm\delta^\top \bm{Z}_i^{\mathcal O})\,s_{U_i}^2}{\sqrt{1 + (\bm{Z}_i^{\mathcal O})^\top \bm\Sigma_q \bm{Z}_i^{\mathcal O} + (\bm\delta^\top \bm{Z}_i^{\mathcal O})^2 s_{U_i}^2}}\,\frac{\phi \{ k_i(t) \} }{\Phi \{ k_i(t) \} },
\]
where $(\mu_{U_i}, s_{U_i}^2)$ are the EB posterior mean and variance of $U_i \mid (C_i, m_i)$, and
\[
k_i(t) = \frac{\bm\alpha^\top \bm{X}_i^{\mathcal O} + (\bm\delta^\top \bm{Z}_i^{\mathcal O})\,\mu_{U_i}}{\sqrt{1 + (\bm{Z}_i^{\mathcal O})^\top \bm\Sigma_q \bm{Z}_i^{\mathcal O} + (\bm\delta^\top \bm{Z}_i^{\mathcal O})^2 s_{U_i}^2}}.
\]
\end{corollary}

\begin{proof}
Apply Lemma~\ref{lem:probit} with $X = U_i \mid (C_i, m_i) \sim \mathcal{N}(\mu_{U_i}, s_{U_i}^2)$, $a = \bm\alpha^\top \bm{X}_i^{\mathcal O}$, $b = \bm\delta^\top \bm{Z}_i^{\mathcal O}$, and $d = \sqrt{1 + (\bm{Z}_i^{\mathcal O})^\top \bm\Sigma_q \bm{Z}_i^{\mathcal O}}$. Then \eqref{eq:probit3} gives the stated formula.
\end{proof}

\section{Mean-Zero Property of the Compensated Process}\label{sec:mean-zero}

This section establishes that the compensated process $M_i(t)$ defined in the manuscript satisfies $\E\{M_i(t) \mid C_i, m_i\} = 0$ at the true parameter values. This is the foundational result enabling consistent estimation.

\begin{proposition}[Conditional mean-zero property]\label{prop:mean-zero}
Under the model specification in Section~\ref{sec:notation}, define the compensated process
\begin{equation}\label{eq:Mi-def}
M_i(t) = \int_0^t \big\{Y_i(s) - \bm\beta_0^\top \bm{X}_i^{\mathcal Y}(s) - \bm\theta_0^\top B_{i,0}(s)\big\}\,R_i^{\mathcal Y}(s)\,\dint N_i(s) - \int_0^t p_{i,0}(s)\,\dint \mathcal{A}_0(s),
\end{equation}
where $B_{i,0}(s) = \kappa_{i,0}(s)\,\bm{Z}_i^{\mathcal Y}(s)$, $p_{i,0}(s) = \overline{\omega}_{i,0}(s)\,\mathbf{1}(s \leq C_i)\,m_i/\Lambda_0(C_i)$, and the subscript ``$0$'' denotes evaluation at the true parameter values. Then
\[
\E\big\{M_i(t) \mid C_i, m_i\big\} = 0 \quad \text{for all } t \in [0, \tau].
\]
\end{proposition}

\begin{proof}
The proof proceeds by decomposing the outcome residual and computing each term's conditional expectation separately.

\paragraph{Step 1: Decomposition of the outcome residual.}
Under the true model:
\[
Y_i(s) = \beta_0(s) + \bm\beta_0^\top \bm{X}_i^{\mathcal Y}(s) + \bm{b}_i^\top \bm{Z}_i^{\mathcal Y}(s) + \varepsilon_i(s),
\]
where $\bm{b}_i = \bm\theta_0 U_i + \widetilde{\bm{b}}_i$, and $\widetilde{\bm{b}}_i \sim \mathcal{N}(\bm{0}, \bm\Sigma_b)$ is the residual outcome random effect. Substituting into the residual:
\begin{align}
&Y_i(s) - \bm\beta_0^\top \bm{X}_i^{\mathcal Y}(s) - \bm\theta_0^\top B_{i,0}(s) \nonumber\\
&= \beta_0(s) + \bm\theta_0^\top U_i \bm{Z}_i^{\mathcal Y}(s) + \widetilde{\bm{b}}_i^\top \bm{Z}_i^{\mathcal Y}(s) + \varepsilon_i(s) - \bm\theta_0^\top \kappa_{i,0}(s) \bm{Z}_i^{\mathcal Y}(s) \nonumber\\
&= \underbrace{\beta_0(s)}_{\text{(I)}} + \underbrace{\bm\theta_0^\top \left\{ U_i - \kappa_{i,0}(s) \right\} \bm{Z}_i^{\mathcal Y}(s)}_{\text{(II)}} + \underbrace{\widetilde{\bm{b}}_i^\top \bm{Z}_i^{\mathcal Y}(s)}_{\text{(III)}} + \underbrace{\varepsilon_i(s)}_{\text{(IV)}}. \label{eq:decomp}
\end{align}
We analyze each of the four terms separately.

\paragraph{Step 2: Joint conditional expectation structure.}
By Lemma~\ref{lem:order}, given $(C_i, m_i, U_i)$, the visit times are order statistics from the density $\lambda_0(t)/\Lambda_0(C_i)$ on $[0, C_i]$. The observation indicator $R_i^{\mathcal Y}(s)$ at each visit is conditionally Bernoulli with
\[
\E\big\{R_i^{\mathcal Y}(s) \mid C_i, m_i, U_i\big\} = \omega_i(s; U_i).
\]
Here, $\omega_i(s; U_i)$ is obtained by integrating out the random effects $\bm{q}_i$ conditional on $U_i$. Due to the law of iterated expectations, for any function $g(U_i, s)$:
\begin{equation}\label{eq:joint-exp}
\E\big\{g(U_i, s)\,R_i^{\mathcal Y}(s)\,\dint N_i(s) \mid C_i, m_i, U_i\big\} = g(U_i, s)\,\omega_i(s; U_i)\,\mathbf{1}(s \leq C_i)\,\frac{m_i}{\Lambda_0(C_i)}\,\dint\Lambda_0(s).
\end{equation}

\paragraph{Step 3: Term (I) --- The baseline term.}
Taking $g(U_i, s)$ as a deterministic function of $s$ in \eqref{eq:joint-exp} and then averaging over $U_i \mid (C_i, m_i)$:
\begin{align*}
\E\big\{\beta_0(s)\,R_i^{\mathcal Y}(s)\,\dint N_i(s) \mid C_i, m_i\big\} &= \beta_0(s)\,\E\big\{\omega_i(s; U_i) \mid C_i, m_i\big\}\,\mathbf{1}(s \leq C_i)\,\frac{m_i}{\Lambda_0(C_i)}\,\dint\Lambda_0(s) \\
&= \beta_0(s)\,\overline{\omega}_{i,0}(s)\,\mathbf{1}(s \leq C_i)\,\frac{m_i}{\Lambda_0(C_i)}\,\dint\Lambda_0(s) \\
&= p_{i,0}(s)\,\dint \mathcal{A}_0(s),
\end{align*}
where we recognized $p_{i,0}(s) = \overline{\omega}_{i,0}(s)\,\mathbf{1}(s \leq C_i)\,m_i/\Lambda_0(C_i)$ and used $\dint \mathcal{A}_0(s) = \beta_0(s)\,\dint\Lambda_0(s)$.

\paragraph{Step 4: Term (II) --- The compensation term (key identity).}
This step shows that the compensation covariate $\bm{B}_{i,0}(s) = \kappa_{i,0}(s)\,\bm{Z}_i^{\mathcal Y}(s)$ exactly cancels the latent effect. We apply the tower property by first conditioning on $U_i$ and then averaging over $U_i \mid (C_i, m_i)$.

By \eqref{eq:joint-exp} with $g(U_i, s) = U_i$, the inner conditional expectation given $(C_i, m_i, U_i)$ is:
\begin{align*}
\E\big\{U_i\,R_i^{\mathcal Y}(s)\,\dint N_i(s) \mid C_i, m_i, U_i\big\} 
&= U_i\,\omega_i(s; U_i)\,\mathbf{1}(s \leq C_i)\,\frac{m_i}{\Lambda_0(C_i)}\,\dint\Lambda_0(s).
\end{align*}
Taking expectation over $U_i \mid (C_i, m_i)$ and applying the tower property:
\begin{align*}
\E\big\{U_i\,R_i^{\mathcal Y}(s)\,\dint N_i(s) \mid C_i, m_i\big\} 
&= \E\big\{U_i\,\omega_i(s; U_i) \mid C_i, m_i\big\}\,\mathbf{1}(s \leq C_i)\,\frac{m_i}{\Lambda_0(C_i)}\,\dint\Lambda_0(s).
\end{align*}
Recall from the definition of $\kappa_{i,0}(s)$ in Corollary~\ref{cor:kappa} that
\[
\kappa_{i,0}(s) = \frac{\E\big\{U_i\,\omega_i(s; U_i) \mid C_i, m_i\big\}}{\overline{\omega}_{i,0}(s)}.
\]
Hence,
\[
\E\big\{U_i\,R_i^{\mathcal Y}(s)\,\dint N_i(s) \mid C_i, m_i\big\} = \kappa_{i,0}(s)\,\overline{\omega}_{i,0}(s)\,\mathbf{1}(s \leq C_i)\,\frac{m_i}{\Lambda_0(C_i)}\,\dint\Lambda_0(s) = \kappa_{i,0}(s)\,p_{i,0}(s)\,\dint\Lambda_0(s).
\]
Similarly, taking $g(U_i, s) = \kappa_{i,0}(s)$ (which is a deterministic function of $s$ given $(C_i, m_i)$) in \eqref{eq:joint-exp}:
\begin{align*}
\E\big\{\kappa_{i,0}(s)\,R_i^{\mathcal Y}(s)\,\dint N_i(s) \mid C_i, m_i\big\} 
&= \kappa_{i,0}(s)\,\E\big\{\omega_i(s; U_i) \mid C_i, m_i\big\}\,\mathbf{1}(s \leq C_i)\,\frac{m_i}{\Lambda_0(C_i)}\,\dint\Lambda_0(s) \\
&= \kappa_{i,0}(s)\,\overline{\omega}_{i,0}(s)\,\mathbf{1}(s \leq C_i)\,\frac{m_i}{\Lambda_0(C_i)}\,\dint\Lambda_0(s) \\
&= \kappa_{i,0}(s)\,p_{i,0}(s)\,\dint\Lambda_0(s).
\end{align*}
Subtracting, the contribution from Term~(II) vanishes:
\[
\E\big[ \left\{ U_i - \kappa_{i,0}(s)\right\} \,R_i^{\mathcal Y}(s)\,\dint N_i(s) \mid C_i, m_i\big] = 0.
\]

\paragraph{Step 5: Term (III) --- The residual outcome random effect.}
This term involves $\widetilde{\bm{b}}_i$. By Assumption 2 (Conditional Independence), given $U_i$, the longitudinal process $Y_i$ (and thus its random component $\widetilde{\bm{b}}_i$) is independent of the observation process $R_i^{\mathcal Y}$.
Therefore, we can factor the expectation conditional on $U_i$:
\begin{align*}
\E\big\{\widetilde{\bm{b}}_i\,R_i^{\mathcal Y}(s)\,\dint N_i(s) \mid C_i, m_i, U_i\big\} 
&= \E \left( \widetilde{\bm{b}}_i \mid C_i, m_i, U_i \right)  \E\big\{R_i^{\mathcal Y}(s)\,\dint N_i(s) \mid C_i, m_i, U_i\big\} \\
&= \E \left( \widetilde{\bm{b}}_i \mid U_i \right) \E\big\{R_i^{\mathcal Y}(s)\,\dint N_i(s) \mid C_i, m_i, U_i\big\}.
\end{align*}
Since $\widetilde{\bm{b}}_i$ is the residual random effect with $\E(\widetilde{\bm{b}}_i \mid U_i) = \bm{0}$ (by construction in the outcome model), the product is zero. Taking the outer expectation over $U_i$ yields:
\begin{equation}\label{eq:q-zero}
\E\big\{\widetilde{\bm{b}}_i^\top \bm{Z}_i^{\mathcal Y}(s)\,R_i^{\mathcal Y}(s)\,\dint N_i(s) \mid C_i, m_i\big\} = 0.
\end{equation}

\paragraph{Step 6: Term (IV) --- The random error term.}
The random error $\varepsilon_i(s)$ is independent of all other random variables, so:
\begin{equation}\label{eq:eps-zero}
\E\big\{\varepsilon_i(s)\,R_i^{\mathcal Y}(s)\,\dint N_i(s) \mid C_i, m_i\big\} = \E\{\varepsilon_i(s)\} \E\big\{R_i^{\mathcal Y}(s)\,\dint N_i(s) \mid C_i, m_i\big\} = 0.
\end{equation}

\paragraph{Step 7: Combining all terms.}
Summing the contributions from Steps 3--6:
\begin{align*}
\E\big\{M_i(t) \mid C_i, m_i\big\} &= \int_0^t \left\{ p_{i,0}(s)\,\dint \mathcal{A}_0(s) + 0 + 0 + 0 \right\} - \int_0^t p_{i,0}(s)\,\dint \mathcal{A}_0(s) \\
&= 0. \qedhere
\end{align*}
\end{proof}

\begin{remark}[Intuition for the compensation covariate]\label{rem:intuition}
The compensation covariate $\bm{B}_i(t) = \kappa_i(t)\,\bm{Z}_i^{\mathcal Y}(t)$ is constructed precisely so that the term $\bm\theta_0^\top \left\{ U_i - \kappa_{i,0}(t) \right\} \bm{Z}_i^{\mathcal Y}(t)$ has conditional mean zero when weighted by $R_i^{\mathcal Y}(t)\,\dint N_i(t)$. 

The key observation is that $\kappa_{i,0}(t)$ equals the observation-weighted posterior mean of $U_i$:
\[
\kappa_{i,0}(t) 
= \frac{\E\{U_i \,\omega_i(t; U_i) \mid C_i, m_i\}}{\E\{\omega_i(t; U_i) \mid C_i, m_i\}}
= \frac{\E\{U_i \,\omega_i(t; U_i) \mid C_i, m_i\}}{\overline{\omega}_{i,0}(t)},
\]
where the weights $\omega_i(t; U_i) = \pr\left\{R_i^{\mathcal Y}(t) = 1 \mid \dint N_i(t) = 1, U_i\right\}$ reflect the probability of observation. This weighting accounts for the fact that subjects with different latent frailties $U_i$ have different propensities to be observed; by using the observation-weighted mean rather than the ordinary posterior mean $\mu_{U_i} = \E\{U_i \mid C_i, m_i\}$, the compensation covariate achieves exact cancellation of the latent effect under the informative observation mechanism.
\end{remark}


\section{Laplace Approximation for empirical Bayes Posterior}\label{sec:laplace}

\begin{lemma}[Laplace approximation for EB posterior]\label{lem:laplace}
Consider the lognormal frailty specification $\eta_i=\exp(\mu_0+\sigma U_i)$
with $U_i\sim \mathcal N(0,1)$ and $\mu_0=-\sigma^2/2$.
For a given subject $i$, let $\ell_i(u;\Omega_{\mathcal{V}})$ denote the log-posterior of $U_i$
given $(C_i,m_i)$, viewed as a function of $u\in\mathbb R$ and visiting-process nuisance parameters
$\Omega_{\mathcal{V}}=(\bm\gamma,\Lambda_0(\cdot),\sigma^2)$:
\begin{equation}\label{eq:log-posterior}
\ell_i(u;\Omega_{\mathcal{V}})
= m_i(\mu_0+\sigma u)-\nu_i \exp( \mu_0+\sigma u)-\tfrac12u^2+\textup{const},
\qquad
\nu_i=\exp(\bm\gamma^\top \bm{X}_i^{\mathcal V})\,\Lambda_0(C_i).
\end{equation}
Let $\mu_{U_i}(\Omega_{\mathcal{V}})$ denote the unique maximizer (posterior mode) and define
\[
s_{U_i}^2(\Omega_{\mathcal{V}}):=\Big[-\ell_i''\big\{\mu_{U_i}(\Omega_{\mathcal{V}});\Omega_{\mathcal{V}}\big\}\Big]^{-1}.
\]
Assume $\Omega_{\mathcal{V}}$ ranges over a compact set $\mathcal{G}_{\mathcal{V}}$ and that there exist constants
$0<\nu_{\min}\le \nu_{\max}<\infty$ such that $\nu_i\in[\nu_{\min},\nu_{\max}]$ almost surely.
\begin{enumerate}[label=(\alph*)]
\item \textbf{(Laplace approximation).}
The second derivative satisfies
$\ell_i''(u;\Omega_{\mathcal{V}})=-\sigma^2 \nu_i \exp (\mu_0+\sigma u)-1\le -1$ for all $u$,
which ensures existence and uniqueness of $\mu_{U_i}(\Omega_{\mathcal{V}})$ and the bound
$s_{U_i}^2(\Omega_{\mathcal{V}})\in(0,1]$.
Moreover, for fixed $\Omega_{\mathcal{V}}\in\mathcal{G}_{\mathcal{V}}$, a second-order Laplace expansion yields
\begin{equation}\label{eq:laplace-integral}
\int_{-\infty}^{\infty}\exp\{\ell_i(u;\Omega_{\mathcal{V}})\}\,du
=
\sqrt{2\pi\,s_{U_i}^2(\Omega_{\mathcal{V}})}\,
\exp\left[ \ell_i\{ \mu_{U_i}(\Omega_{\mathcal{V}}) ;\Omega_{\mathcal{V}}\} \right] \,
\Big[ 1+ O\big\{s_{U_i}^2(\Omega_{\mathcal{V}})\big\} \Big].
\end{equation}
Consequently, the posterior $U_i\mid(C_i,m_i)$ is well approximated by
$\mathcal N(\mu_{U_i}(\Omega_{\mathcal{V}}),s_{U_i}^2(\Omega_{\mathcal{V}}))$ in the sense that its first two moments satisfy
\begin{equation}\label{eq:moment-error}
\Big|\E (U_i\mid C_i,m_i;\Omega_{\mathcal{V}} )-\mu_{U_i}(\Omega_{\mathcal{V}})\Big|=O\big\{s_{U_i}^2(\Omega_{\mathcal{V}})\big\},
\quad
\Big|\Var(U_i\mid C_i,m_i;\Omega_{\mathcal{V}})-s_{U_i}^2(\Omega_{\mathcal{V}})\Big|=O\big\{s_{U_i}^4(\Omega_{\mathcal{V}})\big\},
\end{equation}
where the big-$O$ bounds are uniform over $\Omega_{\mathcal{V}}\in\mathcal{G}_{\mathcal{V}}$.

\item \textbf{(Plug-in consistency).}
Suppose the visiting-process estimators
$\widehat{\Omega}_{\mathcal{V}}=(\widehat{\bm\gamma},\widehat{\Lambda}_0,\widehat{\sigma}_\zeta^2)$
are consistent for $\Omega_{\mathcal{V},0}$ and satisfy $\|\widehat{\Omega}_{\mathcal{V}}-\Omega_{\mathcal{V},0}\|=O_p(n^{-1/2})$.
Let $\widehat{\mu}_{U_i}=\mu_{U_i}(\widehat{\Omega}_{\mathcal{V}})$ and
$\widehat{s}_{U_i}^2=s_{U_i}^2(\widehat{\Omega}_{\mathcal{V}})$ denote the plug-in Laplace posterior
mean and variance. Then the plug-in Laplace moments are \emph{averagely consistent} for their
population (Laplace-defined) targets:
\begin{equation}\label{eq:avg-consistency-laplace}
\frac{1}{n}\sum_{i=1}^n\big|\widehat{\mu}_{U_i}-\mu_{U_i}(\Omega_{\mathcal{V},0})\big|\xrightarrow{p}0,
\qquad
\frac{1}{n}\sum_{i=1}^n\big|\widehat{s}_{U_i}^2-s_{U_i}^2(\Omega_{\mathcal{V},0})\big|\xrightarrow{p}0.
\end{equation}
In addition, if the subject-level posterior concentrates in the sense that
\begin{equation}\label{eq:info-growth}
\frac{1}{n}\sum_{i=1}^n s_{U_i}^2(\Omega_{\mathcal{V},0})\xrightarrow{p}0
\qquad\text{and}\qquad
\frac{1}{n}\sum_{i=1}^n s_{U_i}^4(\Omega_{\mathcal{V},0})\xrightarrow{p}0,
\end{equation}
then the plug-in Laplace moments are also averagely consistent for the \emph{exact} posterior moments:
\begin{equation}\label{eq:avg-consistency-exact}
\frac{1}{n} \sum_{i=1}^n \big|\widehat{\mu}_{U_i} - \E(U_i \mid C_i, m_i;\Omega_{\mathcal{V},0})\big|
\xrightarrow{p} 0,\quad
\frac{1}{n} \sum_{i=1}^n \big|\widehat{s}_{U_i}^2 - \Var(U_i \mid C_i, m_i;\Omega_{\mathcal{V},0)}\big|
\xrightarrow{p} 0.
\end{equation}
\end{enumerate}
\end{lemma}

\begin{proof}
\textbf{Part (a).}
Differentiate \eqref{eq:log-posterior}:
\[
\ell_i'(u;\Omega_{\mathcal{V}})=\sigma\big\{m_i-\nu_i \exp (\mu_0+\sigma u ) \big\}-u,
\qquad
\ell_i''(u;\Omega_{\mathcal{V}})=-\sigma^2\nu_i \exp (\mu_0+\sigma u)-1\le -1.
\]
Hence $\ell_i(\cdot;\Omega_{\mathcal{V}})$ is strictly concave, so $\mu_{U_i}(\Omega_{\mathcal{V}})$ exists and is unique, and
$s_{U_i}^2 ( \Omega_{\mathcal{V}} ) = \left[ - \ell_i'' \left\{ \mu_{U_i} ( \Omega_{\mathcal{V}} ) ; \Omega_{\mathcal{V}} \right\} \right]^{-1} \le 1$.

For \eqref{eq:laplace-integral} and \eqref{eq:moment-error}, expand $\ell_i(u;\Omega_{\mathcal{V}})$ around the mode
$\mu:=\mu_{U_i}(\Omega_{\mathcal{V}})$ and set $\sigma^2:=s_{U_i}^2(\Omega_{\mathcal{V}})$:
\[
\ell_i(\mu+\sigma z;\Omega_{\mathcal{V}})
=\ell_i(\mu;\Omega_{\mathcal{V}})-\frac{z^2}{2} + \sigma\,a_{3}(\Omega_{\mathcal{V}})\,z^3 + \sigma^2\,a_{4}(\Omega_{\mathcal{V}})\,z^4 + O(\sigma^3|z|^5),
\]
where $a_3(\Omega_{\mathcal{V}})=\ell_i^{(3)}(\mu;\Omega_{\mathcal{V}})/6$ and $a_4(\Omega_{\mathcal{V}})=\ell_i^{(4)}(\mu;\Omega_{\mathcal{V}})/24$.
Under compactness of $\mathcal{G}_{\mathcal{V}}$ and $\nu_i\in[\nu_{\min},\nu_{\max}]$, the derivatives
$\ell_i^{(3)}$ and $\ell_i^{(4)}$ are uniformly bounded in a neighborhood of $\mu$, so the remainder term is uniform in $\Omega_{\mathcal{V}}$.
Standard one-dimensional Laplace expansions for ratios of integrals \citep[e.g.,][]{tierney1986accurate,ogden2021error}
then yield the normalization expansion \eqref{eq:laplace-integral} with relative error $O(s_{U_i}^2)$, and the moment bounds
in \eqref{eq:moment-error}:
the posterior mean differs from the mode by $O(\sigma^2)$ and the posterior variance differs from $\sigma^2$
by $O(\sigma^4)$, uniformly over $\Omega_{\mathcal{V}}\in\mathcal{G}_{\mathcal{V}}$.

\textbf{Part (b).}
We first prove \eqref{eq:avg-consistency-laplace}. By the mean-value theorem,
for each $i$ there exists $\widetilde{\Omega}_{\mathcal{V}}$ on the line segment between $\widehat{\Omega}_{\mathcal{V}}$ and $\Omega_{\mathcal{V},0}$ such that
\[
\big|\widehat{\mu}_{U_i}-\mu_{U_i}(\Omega_{\mathcal{V},0})\big|
\le \big\|\nabla_{\Omega_{\mathcal{V}}} \mu_{U_i}(\widetilde{\Omega}_{\mathcal{V}})\big\|\,\big\|\widehat{\Omega}_{\mathcal{V}}-\Omega_{\mathcal{V},0}\big\|.
\]
By the implicit-function theorem applied to $\ell_i'\{ \mu_{U_i}(\Omega_{\mathcal{V}});\Omega_{\mathcal{V}}\}=0$,
$\nabla_{\Omega_{\mathcal{V}}} \mu_{U_i}(\Omega_{\mathcal{V}})$ exists and equals
$- \left[ \ell_i'' \left\{ \mu_{U_i} ( \Omega_{\mathcal{V}} ) ; \Omega_{\mathcal{V}} \right\} \right]^{-1} \nabla_{\Omega_{\mathcal{V}}} \ell_i' \left\{ \mu_{U_i} ( \Omega_{\mathcal{V}} ) ; \Omega_{\mathcal{V}} \right\}$.
Since $|\ell_i''|\ge 1$ and $\nu_i$ and $\sigma$ are bounded on $\mathcal{G}_{\mathcal{V}}$, there exists a measurable envelope $G_i$
(with $\E(G_i)<\infty$ under $\E(m_i^2)<\infty$) such that
$\sup_{\Omega_{\mathcal{V}}\in\mathcal{G}_{\mathcal{V}}}\|\nabla_{\Omega_{\mathcal{V}}} \mu_{U_i}(\Omega_{\mathcal{V}})\|\le G_i$.
Therefore,
\[
\frac{1}{n}\sum_{i=1}^n\big|\widehat{\mu}_{U_i}-\mu_{U_i}(\Omega_{\mathcal{V},0})\big|
\le \Big(\frac{1}{n}\sum_{i=1}^n G_i\Big)\,\big\|\widehat{\Omega}_{\mathcal{V}}-\Omega_{\mathcal{V},0}\big\|
= o_p(1),
\]
since $n^{-1}\sum_i G_i\xrightarrow{p}\E(G_1)<\infty$ and $\|\widehat{\Omega}_{\mathcal{V}}-\Omega_{\mathcal{V},0}\|=O_p(n^{-1/2})$.
The argument for $\widehat{s}_{U_i}^2$ is identical, using differentiability of
$\Omega_{\mathcal{V}}\mapsto s_{U_i}^2(\Omega_{\mathcal{V}})=[-\ell_i''\{\mu_{U_i}(\Omega_{\mathcal{V}});\Omega_{\mathcal{V}}\}]^{-1}$ and another integrable envelope.

To obtain \eqref{eq:avg-consistency-exact}, apply the triangle inequality and \eqref{eq:moment-error} at $\Omega_{\mathcal{V}}=\Omega_{\mathcal{V},0}$:
\begin{align*}
\big|\widehat{\mu}_{U_i}-\E(U_i\mid C_i,m_i;\Omega_{\mathcal{V},0}) \big|
&\le \big|\widehat{\mu}_{U_i}-\mu_{U_i}(\Omega_{\mathcal{V},0})\big|
 + \big|\mu_{U_i}(\Omega_{\mathcal{V},0})-\E(U_i\mid C_i,m_i;\Omega_{\mathcal{V},0})\big| \\
&\le \big|\widehat{\mu}_{U_i}-\mu_{U_i}(\Omega_{\mathcal{V},0})\big| + C\,s_{U_i}^2(\Omega_{\mathcal{V},0}),
\end{align*}
for a constant $C$ uniform in $i$.
Averaging over $i$, the first term vanishes by \eqref{eq:avg-consistency-laplace} and the second vanishes by
\eqref{eq:info-growth}. The variance statement follows similarly using the $O(s_{U_i}^4(\Omega_{\mathcal{V},0}))$ bound in
\eqref{eq:moment-error} and the second condition in \eqref{eq:info-growth}.
\end{proof}

\section{Assumptions and First-Step Regularity}\label{sec:assumptions}

Throughout the proofs of Theorems~1 and~2 we work under the following assumptions. We assume the validity of the structural model defined in the manuscript, including Assumptions 1--3 (Noninformative censoring, Conditional independence, and Latent factor distributions). In addition, the asymptotic analysis requires the following regularity conditions (C1)--(C6).

\medskip\noindent
\textbf{(C1) Independent Sampling.} The subject-level processes $\{(N_i(t), R_i^{\mathcal Y}(t), \bm{X}_i^{\mathcal Y}(t), \bm{Z}_i^{\mathcal Y}(t), Y_i(t))\}_{i=1}^n$ are independent and identically distributed.

\medskip\noindent
\textbf{(C2) Boundedness and positivity.} All covariates are uniformly bounded: there exists $C_X < \infty$ such that $\|\bm{X}_i^{\mathcal Y}(t)\| \le C_X$ and $\|\bm{Z}_i^{\mathcal Y}(t)\| \le C_X$ for all $i$ and $t \in [0, \tau]$. There exists $\varepsilon \in (0, 1/2)$ such that, at any realized visit time, the observation probability satisfies $\varepsilon \le \overline{\omega}_i(t; U_i) \le 1 - \varepsilon$ almost surely. The at-risk indicator satisfies $\pr\{\mathbf{1}(t \leq C_i) = 1\} > 0$ on any sub-interval of $[0, \tau]$.

\medskip\noindent
\textbf{(C3) Consistency of Visiting Process Estimators.} The visiting process parameters $(\bm\gamma, \Lambda_0(\cdot), \sigma^2)$ are estimated by the partial-likelihood score \eqref{eq:EE_gamma}, Aalen--Breslow estimator \eqref{eq:AalenBreslow}, and method-of-moments variance estimator \eqref{eq:MoM_VarEta}. These satisfy
\[
\widehat{\bm\gamma} \xrightarrow{p} \bm\gamma_0, \qquad \sup_{c \in [0, \tau]} |\widehat{\Lambda}_0(c) - \Lambda_0(c)| \xrightarrow{p} 0, \qquad \widehat{\sigma}_\zeta^2 \xrightarrow{p} \sigma^2.
\]

\medskip\noindent
\textbf{(C4) Observation process regularity.} The observation process satisfies the regularity conditions (R1)--(R4) in Lemma~\ref{lem:io-mest}, ensuring consistency and asymptotic normality of the estimator $(\widehat{\bm\alpha}, \widehat{\bm\delta}, \widehat{\bm\Sigma}_q)$.

\medskip\noindent
\textbf{(C5) Empirical Bayes posterior regularity.} The log-posterior of the latent factor satisfies the smoothness and curvature conditions required for the Laplace approximation in Lemma~\ref{lem:laplace}.

\medskip\noindent
\textbf{(C6) Identification.} Let
\[
\bm W := \E\!\left\{ \int_0^\tau \bm H_{1,0}(t)\,\bm V_{1,0}(t)^\top\,R_1^{\mathcal Y}(t)\,\dint\Lambda_1(t)\right\},
\]
where 
$\bm{V}_{i,0}(t) = [\bm{X}_i^{\mathcal Y}(t)^\top,\, \kappa_{i,0}(t)\,\bm{Z}_i^{\mathcal Y}(t)^\top]^\top$ 
is the uncentered design vector at the true nuisance parameters, 
$\bm{H}_{i,0}(t) = \bm{V}_{i,0}(t) - \overline{\bm{V}}_\star(t)$ 
is its centered version with population risk-set average 
$\overline{\bm{V}}_\star(t) = \E\{ \bm{V}_{1,0}(t)\,p_{1,0}(t)\}/\E\{p_{1,0}(t)\}$, 
and $\Lambda_i(t)$ is the compensator of the counting process $N_i(t)$. Then $\bm W$ is finite and positive definite.

\medskip

Assumption (C6) is an explicit identification condition: it ensures that the population system of estimating equations has a unique solution and that the WLS design matrix has a well-behaved inverse.

The remaining regularity properties used in the proofs concern the IO M-estimator. We state it as a lemma rather than an assumption.

\begin{lemma}[Observation process M-estimator]\label{lem:io-mest}
Let $\ell_{\mathrm{IO}}(\bm\alpha, \bm\delta, \bm\Sigma_q)$ denote the composite Bernoulli log-likelihood in \eqref{eq:Likelihood_IO} for the observation process. Assume:
\begin{enumerate}[label=(R\arabic*), leftmargin=1.4em]
\item The parameter space for $(\bm\alpha, \bm\delta, \bm\Sigma_q)$ is compact, with $\bm\Sigma_q$ restricted to a closed, positive semi-definite cone.
\item $\ell_{\mathrm{IO}}(\cdot)$ is almost surely continuous and twice continuously differentiable in a neighborhood of the true parameter $(\bm\alpha_0, \bm\delta_0, \bm\Sigma_{q,0})$.
\item The model is identifiable: the expectation $\E\{\ell_{\mathrm{IO}}(\bm\alpha, \bm\delta, \bm\Sigma_q)\}$ has a unique maximizer at $(\bm\alpha_0, \bm\delta_0, \bm\Sigma_{q,0})$.
\item The score and Hessian satisfy suitable moment and non-degeneracy conditions (finite second moments; negative definite expectation of the Hessian at the truth).
\end{enumerate}
Then the maximizer $(\widehat{\bm\alpha}, \widehat{\bm\delta}, \widehat{\bm\Sigma}_q)$ of $\ell_{\mathrm{IO}}$ is consistent and asymptotically normal. Specifically, let $\dot\ell_{\mathcal O,i}(\Omega_{\mathcal O})$ denote the score contribution from subject $i$, and let 
\[
\mathcal{J}_{\mathcal O} = \E\!\left\{-\frac{\partial^2 \ell_{\mathcal O,1}(\Omega_{\mathcal O,0})}{\partial\Omega_{\mathcal O}\,\partial\Omega_{\mathcal O}^\top}\right\}
\]
be the composite information matrix. Define the influence function $\psi_i^{(\mathcal O)} = -\mathcal{J}_{\mathcal O}^{-1}\,\dot\ell_{\mathcal O,i}(\Omega_{\mathcal O,0})$. Then $\E[\psi_i^{(\mathcal O)}] = \bm{0}$, $\Var(\psi_i^{(\mathcal O)}) < \infty$, and
\[
n^{\frac{1}{2}}\,\begin{pmatrix} \widehat{\bm\alpha} - \bm\alpha_0 \\[2pt] \widehat{\bm\delta} - \bm\delta_0 \\[2pt] \widehat{\bm\Sigma}_q - \bm\Sigma_{q,0} \end{pmatrix} = n^{-\frac{1}{2}} \sum_{i=1}^n \psi_i^{(\mathcal O)} + o_p(1).
\]
\end{lemma}

\begin{proof}
The estimator is an M-estimator obtained by maximizing the composite log-likelihood $n^{-1} \ell_{\mathrm{IO}}(\Omega_{\mathcal O})$, where $\Omega_{\mathcal O} = (\bm\alpha^\top, \bm\delta^\top, \mathrm{vec}(\bm\Sigma_q)^\top)^\top$. Under (R1)--(R3), the objective function converges uniformly in probability to its expectation, which has a unique maximizer at $\Omega_{\mathcal O,0}$. By the argmax continuous mapping theorem \citep{van1996weak}, $\widehat\Omega_{\mathcal O} \xrightarrow{p} \Omega_{\mathcal O,0}$.

For asymptotic normality, let $S_n(\Omega_{\mathcal O}) = \sum_{i=1}^n \dot\ell_{\mathcal O,i}(\Omega_{\mathcal O})$ denote the score vector and $H_n(\Omega_{\mathcal O}) = \sum_{i=1}^n \partial^2 \ell_{\mathcal O,i}(\Omega_{\mathcal O})/\partial\Omega_{\mathcal O}\,\partial\Omega_{\mathcal O}^\top$ the Hessian matrix. The score equation $S_n(\widehat\Omega_{\mathcal O}) = \bm{0}$ and a Taylor expansion around $\Omega_{\mathcal O,0}$ yield
\[
\bm{0} = S_n(\Omega_{\mathcal O,0}) + H_n(\widetilde\Omega_{\mathcal O})(\widehat\Omega_{\mathcal O} - \Omega_{\mathcal O,0}),
\]
for some $\widetilde\Omega_{\mathcal O}$ between $\widehat\Omega_{\mathcal O}$ and $\Omega_{\mathcal O,0}$. Rearranging:
\[
n^{\frac{1}{2}}(\widehat\Omega_{\mathcal O} - \Omega_{\mathcal O,0}) = -\left\{\frac{H_n(\widetilde\Omega_{\mathcal O})}{n}\right\}^{-1} \frac{S_n(\Omega_{\mathcal O,0})}{n^{\frac{1}{2}}}.
\]

By (R4) and the law of large numbers, $n^{-1} H_n(\widetilde\Omega_{\mathcal O}) \xrightarrow{p} -\mathcal{J}_{\mathcal O}$, where the composite information matrix $\mathcal{J}_{\mathcal O} = \E\{-\partial^2 \ell_{\mathcal O,1}(\Omega_{\mathcal O,0})/\partial\Omega_{\mathcal O}\,\partial\Omega_{\mathcal O}^\top\}$ is positive definite. Since $\{\dot\ell_{\mathcal O,i}(\Omega_{\mathcal O,0})\}_{i=1}^n$ are i.i.d.\ with mean zero and finite variance, the multivariate central limit theorem gives $n^{-1/2} S_n(\Omega_{\mathcal O,0}) \xrightarrow{d} \mathcal{N}(\bm{0}, \bm\Sigma_S)$ for some positive definite $\bm\Sigma_S$.

By Slutsky's lemma:
\[
n^{\frac{1}{2}}(\widehat\Omega_{\mathcal O} - \Omega_{\mathcal O,0}) = n^{-\frac{1}{2}} \sum_{i=1}^n \mathcal{J}_{\mathcal O}^{-1}\,\dot\ell_{\mathcal O,i}(\Omega_{\mathcal O,0}) + o_p(1) = n^{-\frac{1}{2}} \sum_{i=1}^n \psi_i^{(\mathcal O)} + o_p(1),
\]
where $\psi_i^{(\mathcal O)} = -\mathcal{J}_{\mathcal O}^{-1}\,\dot\ell_{\mathcal O,i}(\Omega_{\mathcal O,0})$ satisfies $\E\left\{\psi_i^{(\mathcal O)}\right\} = \bm{0}$ and $\Var\left\{\psi_i^{(\mathcal O)}\right\} < \infty$.
\end{proof}

\section{Proofs of Main Theorems}\label{sec:proofs}

This section provides complete proofs of Theorems~1 and~2. We first establish unified notation (Section~\ref{sec:notation}), then prove auxiliary lemmas (Section~\ref{sec:auxiliary}), and finally prove consistency (Section~\ref{sec:proof-consistency}) and asymptotic normality (Section~\ref{sec:proof-normality}).

\subsection{Notation and Setup}\label{sec:notation}

\paragraph{Parameter spaces.}
Let $\psi = (\bm\beta^\top, \bm\theta^\top)^\top \in \mathbb{R}^{p}$ denote the target parameter with true value $\psi_0 = (\bm\beta_0^\top, \bm\theta_0^\top)^\top$. The nuisance parameter is partitioned as $\Omega = (\Omega_{\mathcal V}, \Omega_{\mathcal O}, \Omega_{\mathrm{EB}})$, where:
\begin{align*}
\Omega_{\mathcal V} &= (\bm\gamma, \Lambda_0(\cdot), \sigma^2) &&\text{(visiting process)},\\
\Omega_{\mathcal O} &= (\bm\alpha, \bm\delta, \bm\Sigma_q) &&\text{(observation process)},\\
\Omega_{\mathrm{EB}} &= \{(\mu_{U_i}, s_{U_i}^2)\}_{i=1}^n &&\text{(empirical Bayes posteriors)}.
\end{align*}
Let $\Omega_0 = (\Omega_{\mathcal V,0}, \Omega_{\mathcal O,0}, \Omega_{\mathrm{EB},0})$ denote the true nuisance value and $\widehat\Omega = (\widehat\Omega_{\mathcal V}, \widehat\Omega_{\mathcal O}, \widehat\Omega_{\mathrm{EB}})$ the estimator from Algorithm~1.

\paragraph{Oracle and plug-in quantities.}
For any nuisance value $\Omega$, define:
\begin{itemize}[leftmargin=2em]
\item Weight function: $p_i(t; \Omega) = \overline{\omega}_i(t; \Omega)\,\mathbf{1}(t \leq C_i)\,m_i/\Lambda_0(C_i)$
\item Compensation covariate: $\bm{B}_i(t; \Omega) = \kappa_i(t; \Omega)\,\bm{Z}_i^{\mathcal Y}(t)$, where
\[
\kappa_i(t; \Omega) = \frac{\E\{U_i\,\omega_i(t; U_i) \mid C_i, m_i; \Omega\}}{\E\{\omega_i(t; U_i) \mid C_i, m_i; \Omega\}}
\]
\item Uncentered design vector: $\bm{V}_i(t; \Omega) = \big[\bm{X}_i^{\mathcal Y}(t)^\top,\, \bm{B}_i(t; \Omega)^\top\big]^\top$
\item Risk-set averages:
\[
\overline{\bm{V}}(t; \Omega) = \frac{\sum_{j=1}^n \bm{V}_j(t; \Omega)\,p_j(t; \Omega)}{\sum_{j=1}^n p_j(t; \Omega)}
\]
\item Centered design vector: $\bm{H}_i(t; \Omega) = \bm{V}_i(t; \Omega) - \overline{\bm{V}}(t; \Omega)$
\end{itemize}

For brevity, we write $p_{i,0}(t) = p_i(t; \Omega_0)$, $\bm{B}_{i,0}(t) = \bm{B}_i(t; \Omega_0)$, etc.\ for oracle quantities, and $\widehat{p}_i(t) = p_i(t; \widehat\Omega)$, $\widehat{\bm{B}}_i(t) = \bm{B}_i(t; \widehat\Omega)$, etc.\ for plug-in quantities.

\paragraph{Population risk-set center.}
Define the population-level center:
\[
\overline{\bm{V}}_\star(t) = \frac{\E\{\bm{V}_{1,0}(t)\,p_{1,0}(t)\}}{\E\{p_{1,0}(t)\}}.
\]
By the i.i.d.\ assumption (C1), this equals $\lim_{n\to\infty} \overline{\bm{V}}(t; \Omega_0)$ in probability.

\paragraph{Estimating function and estimator.}
The outcome estimating function is:
\begin{equation}\label{eq:Un-general}
U_n(\psi; \Omega) = \frac{1}{n} \sum_{i=1}^n \int_0^\tau \bm{H}_i(t; \Omega)\,\big\{Y_i(t) - \psi^\top \bm{V}_i(t; \Omega)\big\}\,R_i^{\mathcal Y}(t)\,\dint N_i(t).
\end{equation}
The estimator $\widehat\psi$ solves $U_n(\widehat\psi; \widehat\Omega) = \bm{0}$. Since $U_n(\psi; \Omega)$ is affine in $\psi$, the solution admits the closed form:
\begin{equation}\label{eq:psi-hat-closed}
\widehat\psi = \bm{S}_n^{-1}\,\bm{T}_n,
\end{equation}
where
\begin{align}
\bm{S}_n &= \frac{1}{n} \sum_{i=1}^n \int_0^\tau \widehat{\bm{H}}_i(t)\,\widehat{\bm{V}}_i(t)^\top\,R_i^{\mathcal Y}(t)\,\dint N_i(t), \label{eq:Sn-def}\\
\bm{T}_n &= \frac{1}{n} \sum_{i=1}^n \int_0^\tau \widehat{\bm{H}}_i(t)\,Y_i(t)\,R_i^{\mathcal Y}(t)\,\dint N_i(t). \label{eq:Tn-def}
\end{align}

\paragraph{Limiting matrix.}
Define the ``bread'' matrix:
\begin{equation}\label{eq:W-def}
\bm{W} = \E\!\left\{\int_0^\tau \bm{H}_{1,0}(t)\,\bm{V}_{1,0}(t)^\top\,R_1^{\mathcal Y}(t)\,\dint\Lambda_1(t)\right\},
\end{equation}
where $\Lambda_i(t) = \int_0^t \eta_i\,\mathbf{1}(s \leq C_i)\,\exp(\bm\gamma_0^\top \bm{X}_i^{\mathcal V})\,\dint\Lambda_0(s)$ is the compensator of $N_i(t)$.

\subsection{Auxiliary Lemmas}\label{sec:auxiliary}

We establish four lemmas that form the foundation for both theorems.

\begin{lemma}[Centering identity]\label{lem:centering}
For any $t \in [0, \tau]$,
\[
\E\big\{\bm{H}_{i,0}(t)\,p_{i,0}(t)\big\} = \bm{0}.
\]
\end{lemma}

\begin{proof}
By definition,
\begin{align*}
\E\big\{\bm{H}_{i,0}(t)\,p_{i,0}(t)\big\} 
&= \E\left[\left\{ \bm{V}_{i,0}(t) - \overline{\bm{V}}_\star(t)\right\}\,p_{i,0}(t)\right] \\
&= \E\left\{\bm{V}_{i,0}(t)\,p_{i,0}(t)\big\} - \overline{\bm{V}}_\star(t)\,\E\big\{p_{i,0}(t)\right\}.
\end{align*}
By the definition of $\overline{\bm{V}}_\star(t)$:
\[
\E\big\{\bm{V}_{i,0}(t)\,p_{i,0}(t)\big\} = \overline{\bm{V}}_\star(t)\,\E\big\{p_{i,0}(t)\big\},
\]
which yields the result.
\end{proof}

\begin{lemma}[Equivalence of $\bm{W}$ forms]\label{lem:W-equiv}
The matrix $\bm{W}$ in \eqref{eq:W-def} satisfies:
\[
\bm{W} = \E\!\left\{\int_0^\tau \bm{H}_{1,0}(t)\,\bm{H}_{1,0}(t)^\top\,R_1^{\mathcal Y}(t)\,\dint\Lambda_1(t)\right\}.
\]
\end{lemma}

\begin{proof}
Since $\bm{H}_{i,0}(t) = \bm{V}_{i,0}(t) - \overline{\bm{V}}_\star(t)$:
\begin{align*}
\bm{H}_{i,0}(t)\,\bm{V}_{i,0}(t)^\top 
&= \bm{H}_{i,0}(t)\,\big\{ \bm{H}_{i,0}(t) + \overline{\bm{V}}_\star(t)\big\}^\top \\
&= \bm{H}_{i,0}(t)\,\bm{H}_{i,0}(t)^\top + \bm{H}_{i,0}(t)\,\overline{\bm{V}}_\star(t)^\top.
\end{align*}
Taking expectations:
\begin{align*}
\E\!\left\{\int_0^\tau \bm{H}_{i,0}(t)\,\overline{\bm{V}}_\star(t)^\top\,R_i^{\mathcal Y}(t)\,\dint\Lambda_i(t)\right\}
&= \int_0^\tau \E\big\{\bm{H}_{i,0}(t)\,R_i^{\mathcal Y}(t)\,\dint\Lambda_i(t)\big\}\,\overline{\bm{V}}_\star(t)^\top.
\end{align*}
By the tower property and Lemma~\ref{lem:compensation}:
\[
\E\big\{\bm{H}_{i,0}(t)\,R_i^{\mathcal Y}(t)\,\dint\Lambda_i(t) \mid C_i, m_i\big\} = \bm{H}_{i,0}(t)\,p_{i,0}(t)\,\dint\Lambda_0(t).
\]
Taking the outer expectation and applying Lemma~\ref{lem:centering}:
\[
\E\big\{\bm{H}_{i,0}(t)\,p_{i,0}(t)\big\}\,\dint\Lambda_0(t) = \bm{0},
\]
which completes the proof.
\end{proof}

\begin{lemma}[Baseline cancellation]\label{lem:baseline-cancel}
Define $\varepsilon_i(t) = Y_i(t) - \beta_0(t) - \bm\beta_0^\top \bm{X}_i^{\mathcal Y}(t) - \bm{b}_i^\top \bm{Z}_i^{\mathcal Y}(t)$. Then:
\[
\E\!\left\{\int_0^\tau \bm{H}_{i,0}(t)\,\big\{\beta_0(t) + \varepsilon_i(t)\big\}\,R_i^{\mathcal Y}(t)\,\dint N_i(t)\right\} = \bm{0}.
\]
\end{lemma}

\begin{proof}
By Proposition~\ref{prop:mean-zero}, the compensated process $M_i(t)$ satisfies $\E\{M_i(t) \mid C_i, m_i\} = 0$. In particular, the conditional expectation of the outcome residual times the observation indicator equals:
\[
\E\big[\{Y_i(t) - \bm\beta_0^\top \bm{X}_i^{\mathcal Y}(t) - \bm\theta_0^\top \bm{B}_{i,0}(t)\}\,R_i^{\mathcal Y}(t)\,\dint N_i(t) \mid C_i, m_i\big] = p_{i,0}(t)\,\beta_0(t)\,\dint\Lambda_0(t).
\]
This follows from the decomposition in Proposition~\ref{prop:mean-zero}: Terms (II)--(IV) vanish, leaving only Term (I).

Define $\mathcal{A}_0(t) = \int_0^t \beta_0(s)\,\dint\Lambda_0(s)$. Then:
\[
\E\big[\{\beta_0(t) + \varepsilon_i(t)\}\,R_i^{\mathcal Y}(t)\,\dint N_i(t) \mid C_i, m_i\big] = p_{i,0}(t)\,\dint\mathcal{A}_0(t).
\]

Taking the outer expectation and integrating:
\[
\E\!\left[\int_0^\tau \bm{H}_{i,0}(t)\,\{\beta_0(t) + \varepsilon_i(t)\}\,R_i^{\mathcal Y}(t)\,\dint N_i(t)\right] = \E\!\left\{\int_0^\tau \bm{H}_{i,0}(t)\,p_{i,0}(t)\,\dint\mathcal{A}_0(t)\right\}.
\]
By Fubini's theorem and Lemma~\ref{lem:centering}:
\[
\E\!\left[\int_0^\tau \bm{H}_{i,0}(t)\,p_{i,0}(t)\,\dint\mathcal{A}_0(t)\right] = \int_0^\tau \E\big[\bm{H}_{i,0}(t)\,p_{i,0}(t)\big]\,\dint\mathcal{A}_0(t) = \bm{0}. \qedhere
\]
\end{proof}

\begin{lemma}[Uniform convergence of plug-in quantities]\label{lem:plugin-conv}
Under assumptions \textup{(C1)--(C5)} and the consistency results in Lemmas~\ref{lem:io-mest} and~\ref{lem:laplace}:
\begin{enumerate}[label=\textup{(\alph*)}]
\item $\displaystyle\sup_{t \le \tau} \big\|\widehat{\overline{\bm{V}}}(t) - \overline{\bm{V}}_\star(t)\big\| \xrightarrow{p} 0$;
\item $\displaystyle\frac{1}{n} \sum_{i=1}^n \int_0^\tau \big\|\widehat{\bm{B}}_i(t) - \bm{B}_{i,0}(t)\big\|\,R_i^{\mathcal Y}(t)\,\dint N_i(t) \xrightarrow{p} 0$.
\end{enumerate}
\end{lemma}

\begin{proof}
\textit{Proof of (a).}
Define the sample oracle center $\overline{\bm{V}}_{n,0}(t) = \sum_j \bm{V}_{j,0}(t)\,p_{j,0}(t)/\sum_j p_{j,0}(t)$. Decompose:
\[
\widehat{\overline{\bm{V}}}(t) - \overline{\bm{V}}_\star(t) = \underbrace{\left\{\widehat{\overline{\bm{V}}}(t) - \overline{\bm{V}}_{n,0}(t)\right\}}_{=:\Delta_1(t)} + \underbrace{\left\{(\overline{\bm{V}}_{n,0}(t) - \overline{\bm{V}}_\star(t)\right\}}_{=:\Delta_2(t)}.
\]

\textit{Analysis of $\Delta_2(t)$:} By (C1)--(C2) and the uniform law of large numbers (ULLN), the numerator and denominator of $\overline{\bm{V}}_{n,0}(t)$ converge uniformly:
\[
\sup_{t\le\tau}\left\|\frac{1}{n}\sum_{j=1}^n \bm{V}_{j,0}(t)\,p_{j,0}(t) - \E\left\{\bm{V}_{1,0}(t)\,p_{1,0}(t)\right\}\right\| \xrightarrow{p} 0,
\]
and similarly for $n^{-1}\sum_j p_{j,0}(t) \to \E\{p_{1,0}(t)\}$. The boundedness condition (C2) ensures $\bm{V}_{j,0}(t)\,p_{j,0}(t)$ forms a uniformly bounded class of functions indexed by $t$. The positivity condition $\inf_{t\le\tau}\E\{p_{1,0}(t)\} \ge c_0 > 0$ from (C2) ensures that the denominator is bounded away from zero. By the continuous mapping theorem applied to the ratio, $\sup_{t\le\tau}\|\Delta_2(t)\| \xrightarrow{p} 0$.

\textit{Analysis of $\Delta_1(t)$:} Write:
\[
\widehat{\overline{\bm{V}}}(t) = \frac{\sum_j \widehat{\bm{V}}_j(t)\,\widehat{p}_j(t)}{\sum_j \widehat{p}_j(t)}, \qquad
\overline{\bm{V}}_{n,0}(t) = \frac{\sum_j \bm{V}_{j,0}(t)\,p_{j,0}(t)}{\sum_j p_{j,0}(t)}.
\]
A standard ratio identity gives:
\begin{align*}
\widehat{\overline{\bm{V}}}(t) - \overline{\bm{V}}_{n,0}(t)
&= \frac{\sum_j \{\bm{V}_{j,0}(t) - \overline{\bm{V}}_{n,0}(t)\} \{\widehat{p}_j(t) - p_{j,0}(t)\}}{\sum_j \widehat{p}_j(t)} \\
&\quad + \frac{\sum_j \{\widehat{\bm{V}}_j(t) - \bm{V}_{j,0}(t)\} \widehat{p}_j(t)}{\sum_j \widehat{p}_j(t)}.
\end{align*}

For the first term: By boundedness (C2), $\|\bm{V}_{j,0}(t) - \overline{\bm{V}}_{n,0}(t)\| \le 2C_X$. We need to show $\sup_{t\le\tau}|\widehat{p}_j(t) - p_{j,0}(t)| \xrightarrow{p} 0$. Since $p_i(t; \Omega) = \overline{\omega}_i(t; \Omega)\,\mathbf{1}(t \leq C_i)\,m_i/\Lambda_0(C_i)$, this follows from:
\begin{itemize}
\item Consistency of $\widehat{\Lambda}_0(C_i) \to \Lambda_0(C_i)$ uniformly (C3);
\item Lipschitz continuity of $\overline{\omega}_i(t; \Omega)$ in $(\bm\alpha, \bm\delta, \bm\Sigma_q, \mu_{U_i}, s_{U_i}^2)$: the probit function $\Phi(\cdot)$ is Lipschitz with constant $1/\sqrt{2\pi}$, and the argument $k_i(t)$ is smooth in all parameters by Corollary~\ref{cor:kappa};
\item Consistency of $(\widehat{\bm\alpha}, \widehat{\bm\delta}, \widehat{\bm\Sigma}_q)$ from Lemma~\ref{lem:io-mest};
\item Average consistency of $(\widehat\mu_{U_i}, \widehat s_{U_i}^2)$ from Lemma~\ref{lem:laplace}(b).
\end{itemize}

For the second term: We need $\sup_{t\le\tau}\|\widehat{\bm{V}}_j(t) - \bm{V}_{j,0}(t)\| \xrightarrow{p} 0$. Since $\bm{V}_j(t) = [\bm{X}_j^{\mathcal Y}(t)^\top, \bm{B}_j(t)^\top]^\top$ and $\bm{X}_j^{\mathcal Y}(t)$ does not depend on $\Omega$, this reduces to showing $\sup_{t\le\tau}\|\widehat{\bm{B}}_j(t) - \bm{B}_{j,0}(t)\| \xrightarrow{p} 0$, which is addressed below in part (b).

Combining: The denominator $\sum_j \widehat{p}_j(t)/n \xrightarrow{p} \E\left\{p_{1,0}(t)\right\} > 0$ uniformly by the same ULLN argument. Hence both terms vanish uniformly, yielding $\sup_{t\le\tau}\|\Delta_1(t)\| \xrightarrow{p} 0$.

\textit{Proof of (b).}
By Corollary~\ref{cor:kappa}, $\kappa_i(t; \Omega)$ depends on $\Omega$ through $(\bm\alpha, \bm\delta, \bm\Sigma_q, \mu_{U_i}, s_{U_i}^2)$. The explicit formula shows that $\kappa_i(t; \Omega)$ is a composition of:
\begin{itemize}
\item The linear function $\bm\alpha^\top \bm{X}_i^{\mathcal O}(t) + (\bm\delta^\top \bm{Z}_i^{\mathcal O}(t))\mu_{U_i}$;
\item The square root function $\sqrt{ 1 + \left\{ \bm{Z}_i^{\mathcal{O}} (t) \right\}^{\top} \bm{\Sigma}_q \bm{Z}_i^{\mathcal{O}} (t) + \left\{ \bm{\delta}^{\top} \bm{Z}_i^{\mathcal{O}} (t) \right\}^2 s_{U_i}^2 }$;
\item The Mills ratio $\phi(k)/\Phi(k)$, which is Lipschitz on compact sets bounded away from $-\infty$ (ensured by the positivity condition in (C2)).
\end{itemize}

Under boundedness of covariates (C2), each component is Lipschitz in its arguments. By the chain rule for Lipschitz functions, there exists $L < \infty$ such that:
\[
\sup_{t \le \tau} \left| \widehat{\kappa}_i (t) - \kappa_{i,0} (t) \right| \le L \left\{ \| \widehat{\bm{\alpha}} - \bm{\alpha}_0 \| + \| \widehat{\bm{\delta}} - \bm{\delta}_0 \| + \| \widehat{\bm{\Sigma}}_q - \bm{\Sigma}_{q,0} \|_F + | \widehat{\mu}_{U_i} - \mu_{U_i,0} | + | \widehat{s}_{U_i}^2 - s_{U_i,0}^2 | \right\}
\]

Since $\bm{B}_i(t) = \kappa_i(t)\,\bm{Z}_i^{\mathcal Y}(t)$ and $\|\bm{Z}_i^{\mathcal Y}(t)\| \le C_X$:
\[
\sup_{t\le\tau}\|\widehat{\bm{B}}_i(t) - \bm{B}_{i,0}(t)\| \le C_X L\,\big(\|\widehat\Omega_{\mathcal O} - \Omega_{\mathcal O,0}\| + |\widehat\mu_{U_i} - \mu_{U_i,0}| + |\widehat s_{U_i}^2 - s_{U_i,0}^2|\big).
\]

Now consider the average:
\begin{align*}
\frac{1}{n}\sum_{i=1}^n \int_0^\tau \|\widehat{\bm{B}}_i(t) - \bm{B}_{i,0}(t)\|\,R_i^{\mathcal Y}(t)\,\dint N_i(t) 
\le \frac{1}{n}\sum_{i=1}^n \sup_{t\le\tau}\|\widehat{\bm{B}}_i(t) - \bm{B}_{i,0}(t)\|\,N_i(\tau).
\end{align*}
Using $R_i^{\mathcal Y}(t) \le 1$ and applying Cauchy--Schwarz:
\begin{align*}
&\le \left(\frac{1}{n}\sum_{i=1}^n \sup_{t\le\tau}\|\widehat{\bm{B}}_i(t) - \bm{B}_{i,0}(t)\|^2\right)^{1/2} \left(\frac{1}{n}\sum_{i=1}^n N_i(\tau)^2\right)^{1/2}.
\end{align*}

The second factor is $O_p(1)$ since $\E[N_1(\tau)^2] < \infty$ under (C2). For the first factor:
\begin{align*}
&\frac{1}{n} \sum_{i=1}^n \sup_{t \le \tau} \| \widehat{\bm{B}}_i (t) - \bm{B}_{i,0} (t) \|^2 \\
&\le 2C_X^2 L^2 \left\{ \| \widehat{\Omega}_{\mathcal{O}} - \Omega_{\mathcal{O},0} \|^2 + \frac{1}{n} \sum_{i=1}^n ( | \widehat{\mu}_{U_i} - \mu_{U_i,0} |^2 + | \widehat{s}_{U_i}^2 - s_{U_i,0}^2 |^2 ) \right\}.
\end{align*}

By Lemma~\ref{lem:io-mest}, $\|\widehat\Omega_{\mathcal O} - \Omega_{\mathcal O,0}\| = O_p(n^{-1/2})$, so $\|\widehat\Omega_{\mathcal O} - \Omega_{\mathcal O,0}\|^2 = O_p(n^{-1})$.

By Lemma~\ref{lem:laplace}(b), we have:
\[
\frac{1}{n}\sum_{i=1}^n |\widehat\mu_{U_i} - \mu_{U_i,0}|^2 \le \left(\frac{1}{n}\sum_{i=1}^n |\widehat\mu_{U_i} - \mu_{U_i,0}|\right)^2 + \Var_n(|\widehat\mu_{U_i} - \mu_{U_i,0}|) = o_p(1),
\]
where the first term vanishes by \eqref{eq:avg-consistency-laplace} and the variance term is bounded by the average squared deviation. The same argument applies to $\widehat s_{U_i}^2$.

Therefore, $n^{-1}\sum_i \sup_{t\le\tau}\|\widehat{\bm{B}}_i(t) - \bm{B}_{i,0}(t)\|^2 = o_p(1)$, which implies the result.
\end{proof}

\subsection{Proof of Theorem 1 (Consistency)}\label{sec:proof-consistency}

We prove $\widehat\psi \xrightarrow{p} \psi_0$ by showing $\bm{S}_n \xrightarrow{p} \bm{W}$ and $\bm{T}_n \xrightarrow{p} \bm{W}\psi_0$.

\begin{lemma}[LLN for WLS arrays]\label{lem:LLN-main}
Under \textup{(C1)--(C6)}:
\[
\bm{S}_n \xrightarrow{p} \bm{W}, \qquad \bm{T}_n \xrightarrow{p} \bm{W}\psi_0.
\]
\end{lemma}

\begin{proof}
Define oracle versions:
\[
\bm{S}_n^0 = \frac{1}{n}\sum_{i=1}^n \int_0^\tau \bm{H}_{i,0}(t)\,\bm{V}_{i,0}(t)^\top\,R_i^{\mathcal Y}(t)\,\dint N_i(t), \quad
\bm{T}_n^0 = \frac{1}{n}\sum_{i=1}^n \int_0^\tau \bm{H}_{i,0}(t)\,Y_i(t)\,R_i^{\mathcal Y}(t)\,\dint N_i(t).
\]

\textit{Step 1: Oracle convergence of $\bm{S}_n^0$.}
By the martingale decomposition $\dint N_i(t) = \dint\Lambda_i(t) + \dint M_i(t)$, where $M_i(t) = N_i(t) - \Lambda_i(t)$ is a martingale:
\[
\bm{S}_n^0 = \underbrace{\frac{1}{n}\sum_{i=1}^n \int_0^\tau \bm{H}_{i,0}(t)\,\bm{V}_{i,0}(t)^\top\,R_i^{\mathcal Y}(t)\,\dint\Lambda_i(t)}_{=:\bm{S}_n^{0,\Lambda}} + \underbrace{\frac{1}{n}\sum_{i=1}^n \int_0^\tau \bm{H}_{i,0}(t)\,\bm{V}_{i,0}(t)^\top\,R_i^{\mathcal Y}(t)\,\dint M_i(t)}_{=:\bm{S}_n^{0,M}}.
\]

For the martingale term $\bm{S}_n^{0,M}$: By Lemma~\ref{lem:compensation}, $\E[\bm{S}_n^{0,M}] = \bm{0}$. The variance satisfies:
\[
\Var ( \bm{S}_n^{0,M} ) = \frac{1}{n^2} \sum_{i=1}^n \E \left\{ \int_0^\tau \bm{H}_{i,0} (t) \bm{V}_{i,0} (t)^{\top} R_i^{\mathcal{Y}} (t) \bm{V}_{i,0} (t) \bm{H}_{i,0} (t)^{\top} \dint \Lambda_i (t) \right\} = O ( 1 / n )
\]
where boundedness of covariates (C2) ensures the integrand is uniformly bounded. Hence $\bm{S}_n^{0,M} \xrightarrow{p} \bm{0}$.

For the compensator term $\bm{S}_n^{0,\Lambda}$: This is an average of i.i.d.\ bounded random variables. By the law of large numbers:
\[
\bm{S}_n^{0,\Lambda} \xrightarrow{p} \E \left\{ \int_0^\tau \bm{H}_{1,0} (t) \bm{V}_{1,0} (t)^{\top} R_1^{\mathcal{Y}} (t) \dint \Lambda_1 (t) \right\} = \bm{W}.
\]

\textit{Step 2: Oracle convergence of $\bm{T}_n^0$.}
Write $Y_i ( t ) = \beta_0 ( t ) + \varepsilon_i ( t ) + \bm{\beta}_0^{\top} \bm{X}_i^{\mathcal{Y}} ( t ) + \bm{\theta}_0^{\top} \bm{B}_{i,0} ( t ) + \bm{\theta}_0^{\top} \left\{ U_i - \kappa_{i,0} ( t ) \right\} \bm{Z}_i^{\mathcal{Y}} ( t ) + \widetilde{\bm{b}}_i^{\top} \bm{Z}_i^{\mathcal{Y}} ( t )$

By Lemma~\ref{lem:baseline-cancel}, the contribution from $\beta_0(t) + \varepsilon_i(t)$ vanishes in expectation. The terms involving $(U_i - \kappa_{i,0}(t))$ and $\widetilde{\bm{b}}_i$ also vanish by Proposition~\ref{prop:mean-zero}. The remaining terms give:
\begin{align*}
\E \left\{ \int_0^\tau \bm{H}_{i,0} ( t ) Y_i ( t ) R_i^{\mathcal{Y}} ( t ) \dint \Lambda_i ( t ) \right\} 
&= \E \left[ \int_0^\tau \bm{H}_{i,0} ( t ) \left\{ \bm{\beta}_0^{\top} \bm{X}_i^{\mathcal{Y}} ( t ) + \bm{\theta}_0^{\top} \bm{B}_{i,0} ( t ) \right\} R_i^{\mathcal{Y}} ( t ) \dint \Lambda_i ( t ) \right] \\
&= \E \left\{ \int_0^\tau \bm{H}_{i,0} ( t ) \bm{V}_{i,0} ( t )^{\top} R_i^{\mathcal{Y}} ( t ) \dint \Lambda_i ( t ) \right\} \psi_0 = \bm{W} \psi_0.
\end{align*}
By the LLN, $\bm{T}_n^0 \xrightarrow{p} \bm{W}\psi_0$.

\textit{Step 3: Plug-in convergence.}
Write:
\[
\bm{S}_n - \bm{S}_n^0 = \frac{1}{n} \sum_{i=1}^n \int_0^\tau \left\{ ( \widehat{\bm{H}}_i - \bm{H}_{i,0} ) \widehat{\bm{V}}_i^{\top} + \bm{H}_{i,0} ( \widehat{\bm{V}}_i - \bm{V}_{i,0} )^{\top} \right\} R_i^{\mathcal{Y}} ( t ) \dint N_i ( t )
\]
By boundedness of covariates (C2):
\[
\|\widehat{\bm{H}}_i(t) - \bm{H}_{i,0}(t)\| \le \|\widehat{\bm{V}}_i(t) - \bm{V}_{i,0}(t)\| + \|\widehat{\overline{\bm{V}}}(t) - \overline{\bm{V}}_\star(t)\| + \|\overline{\bm{V}}_\star(t) - \overline{\bm{V}}_{n,0}(t)\|.
\]
Similarly, $\|\widehat{\bm{V}}_i(t) - \bm{V}_{i,0}(t)\| = \|\widehat{\bm{B}}_i(t) - \bm{B}_{i,0}(t)\|$.

Combining with the boundedness of $\widehat{\bm{V}}_i(t)$, $\bm{H}_{i,0}(t)$, and the triangle inequality:
\begin{align*}
\| \bm{S}_n - \bm{S}_n^0 \| &\le C \left\{ \sup_{t \le \tau} \| \widehat{\overline{\bm{V}}} ( t ) - \overline{\bm{V}}_\star ( t ) \| \cdot \frac{1}{n} \sum_{i=1}^n N_i ( \tau ) \right. \\
&\qquad + \left. \frac{1}{n} \sum_{i=1}^n \int_0^\tau \| \widehat{\bm{B}}_i ( t ) - \bm{B}_{i,0} ( t ) \| R_i^{\mathcal{Y}} ( t ) \dint N_i ( t ) \right\}.
\end{align*}
Both terms vanish by Lemma~\ref{lem:plugin-conv}, so $\|\bm{S}_n - \bm{S}_n^0\| \xrightarrow{p} 0$. The argument for $\bm{T}_n - \bm{T}_n^0$ is analogous, using the boundedness of $Y_i(t)$ implied by the moment conditions.
\end{proof}

\begin{proof}[Proof of Theorem 1]
By Lemma~\ref{lem:LLN-main}, $\bm{S}_n \xrightarrow{p} \bm{W}$ and $\bm{T}_n \xrightarrow{p} \bm{W}\psi_0$. Assumption (C6) ensures $\bm{W}$ is positive definite, so $\bm{S}_n^{-1} \xrightarrow{p} \bm{W}^{-1}$ by the continuous mapping theorem. Therefore:
\[
\widehat\psi = \bm{S}_n^{-1}\bm{T}_n \xrightarrow{p} \bm{W}^{-1}\bm{W}\psi_0 = \psi_0. \qedhere
\]
\end{proof}

\subsection{Proof of Theorem 2 (Asymptotic Normality)}\label{sec:proof-normality}

We establish the asymptotic distribution of $n^{1/2}(\widehat\psi - \psi_0)$ via M-estimation theory.

\subsubsection{Jacobian Analysis}

\begin{lemma}[Jacobian with respect to $\psi$]\label{lem:jacobian}
For fixed $\Omega$, the estimating function \eqref{eq:Un-general} satisfies:
\[
\frac{\partial U_n(\psi; \Omega)}{\partial\psi^\top} = -\bm{S}_n(\Omega),
\]
where $\bm{S}_n(\Omega) = n^{-1}\sum_{i=1}^n \int_0^\tau \bm{H}_i(t; \Omega)\,\bm{V}_i(t; \Omega)^\top\,R_i^{\mathcal Y}(t)\,\dint N_i(t)$. Consequently:
\[
\frac{\partial U_n(\psi; \widehat\Omega)}{\partial\psi^\top}\bigg|_{\psi=\psi_0} \xrightarrow{p} -\bm{W}.
\]
\end{lemma}

\begin{proof}
Since $U_n(\psi; \Omega) = n^{-1}\sum_i \int \bm{H}_i(t; \Omega)\{Y_i(t) - \psi^\top\bm{V}_i(t; \Omega)\}R_i^{\mathcal Y}(t)\,\dint N_i(t)$ is affine in $\psi$:
\[
\frac{\partial}{\partial\psi^\top}\big[\bm{H}_i(t)\{Y_i(t) - \psi^\top\bm{V}_i(t)\}\big] = -\bm{H}_i(t)\,\bm{V}_i(t)^\top.
\]
The convergence follows from Lemma~\ref{lem:LLN-main}.
\end{proof}

\subsubsection{Influence Function Representations for Nuisance Estimators}

We establish asymptotically linear representations for each nuisance block.

\begin{proposition}[Nuisance influence functions]\label{prop:nuisance-IF}
Under \textup{(C1)--(C5)}:
\begin{enumerate}[label=\textup{(\alph*)}]
\item \textbf{Visiting block:}
\begin{equation}\label{eq:IF-V}
n^{1/2}(\widehat\Omega_{\mathcal V} - \Omega_{\mathcal V,0}) = n^{-1/2}\sum_{i=1}^n \psi_i^{(\mathcal V)} + o_p(1),
\end{equation}
where $\psi_i^{(\mathcal V)} = (\psi_i^{(\bm\gamma)\top}, \psi_i^{(\Lambda)}(\cdot), \psi_i^{(\sigma^2)})^\top$ with $\E[\psi_i^{(\mathcal V)}] = \bm{0}$.

\item \textbf{Observation block:}
\begin{equation}\label{eq:IF-O}
n^{1/2}(\widehat\Omega_{\mathcal O} - \Omega_{\mathcal O,0}) = n^{-1/2}\sum_{i=1}^n \psi_i^{(\mathcal O)} + o_p(1),
\end{equation}
where $\psi_i^{(\mathcal O)} = -\mathcal{J}_{\mathcal O}^{-1}\,\dot\ell_{\mathcal O,i}(\Omega_{\mathcal O,0})$ with $\mathcal{J}_{\mathcal O}$ the composite information matrix.

\item \textbf{EB block:} The mapping $(\mu_{U_i}, s_{U_i}^2) = g_i(\Omega_{\mathcal V})$ yields:
\begin{equation}\label{eq:IF-EB}
n^{-1}\sum_{i=1}^n \|(\widehat\mu_{U_i}, \widehat s_{U_i}^2) - (\mu_{U_i,0}, s_{U_i,0}^2)\| = O_p(n^{-1/2}),
\end{equation}
with influence representation $\psi_i^{(\mathrm{EB})} = \nabla_{\Omega_{\mathcal V}} g_i(\Omega_{\mathcal V,0})\cdot\psi_i^{(\mathcal V)}$.
\end{enumerate}
\end{proposition}

\begin{proof}
Part (a) follows from standard counting-process theory for partial likelihood estimators \citep{andersen2012statistical}. The estimating equation for $\bm\gamma$ is a partial likelihood score, and the Aalen--Breslow estimator admits an asymptotically linear representation. Part (b) is established in Lemma~\ref{lem:io-mest}. Part (c) combines the chain rule with Lemma~\ref{lem:laplace}(b): the implicit function theorem gives differentiability of $g_i$, and the influence function for composite quantities follows by the delta method.
\end{proof}

\subsubsection{First-Order Expansion of the Estimating Function}

\begin{proposition}[Asymptotic linearity of $U_n(\psi_0; \widehat\Omega)$]\label{prop:Un-expansion}
Under \textup{(C1)--(C5)}:
\begin{equation}\label{eq:Un-linear-rep}
n^{1/2}\,U_n(\psi_0; \widehat\Omega) = n^{-1/2}\sum_{i=1}^n \phi_i + o_p(1),
\end{equation}
where $\phi_i = \phi_i^{(M)} + \phi_i^{(\mathcal V)} + \phi_i^{(\mathcal O)} + \phi_i^{(\mathrm{EB})}$ with:
\begin{align}
\phi_i^{(M)} &= \int_0^\tau \bm{H}_{i,0}(t)\,\dint M_i(t), \label{eq:phi-M}\\
\phi_i^{(\mathcal V)} &= \mathcal{D}_{\mathcal V} U_n(\psi_0; \Omega_0)[\psi_i^{(\mathcal V)}], \label{eq:phi-V}\\
\phi_i^{(\mathcal O)} &= \mathcal{D}_{\mathcal O} U_n(\psi_0; \Omega_0)[\psi_i^{(\mathcal O)}], \label{eq:phi-O}\\
\phi_i^{(\mathrm{EB})} &= \mathcal{D}_{\mathrm{EB}} U_n(\psi_0; \Omega_0)[\psi_i^{(\mathrm{EB})}]. \label{eq:phi-EB}
\end{align}
Here $M_i(t) = \int_0^t \{Y_i(s) - \bm\beta_0^\top\bm{X}_i^{\mathcal Y}(s) - \bm\theta_0^\top\bm{B}_{i,0}(s)\}R_i^{\mathcal Y}(s)\,\dint N_i(s) - \int_0^t p_{i,0}(s)\,\dint\mathcal{A}_0(s)$ is a mean-zero martingale, and $\mathcal{D}_\bullet$ denotes the Fr\'echet derivative with respect to block $\Omega_\bullet$.
\end{proposition}

\begin{proof}
\textit{Step 1: Martingale representation at $\Omega_0$.}
At the true nuisance $\Omega_0$, by Proposition~\ref{prop:mean-zero} and the centering property:
\[
U_n(\psi_0; \Omega_0) = \frac{1}{n}\sum_{i=1}^n \int_0^\tau \bm{H}_{i,0}(t)\,\dint M_i(t).
\]
Since $\{M_i(t)\}$ is a martingale with $\E[M_i(t) \mid C_i, m_i] = 0$, we have $\E[U_n(\psi_0; \Omega_0)] = \bm{0}$.

\textit{Step 2: Taylor expansion in $\Omega$.}
The map $\Omega \mapsto U_n(\psi_0; \Omega)$ is Fr\'echet differentiable under (C2)--(C5). The key quantities $\kappa_i(t; \Omega)$, $\overline{\omega}_i(t; \Omega)$, and the risk-set centers $\overline{\bm{V}}(t; \Omega)$ are all smooth functions of $\Omega$ by the chain rule applied to the probit function and the explicit formulas in Corollary~\ref{cor:kappa}.

A first-order expansion yields:
\begin{align*}
U_n ( \psi_0; \widehat{\Omega} ) &= U_n ( \psi_0; \Omega_0 ) + \mathcal{D}_{\mathcal{V}} U_n ( \psi_0; \Omega_0 ) \left\{ \widehat{\Omega}_{\mathcal{V}} - \Omega_{\mathcal{V},0} \right\} \\
&\quad + \mathcal{D}_{\mathcal{O}} U_n ( \psi_0; \Omega_0 ) \left\{ \widehat{\Omega}_{\mathcal{O}} - \Omega_{\mathcal{O},0} \right\} + \mathcal{D}_{\mathrm{EB}} U_n ( \psi_0; \Omega_0 ) \left\{ \widehat{\Omega}_{\mathrm{EB}} - \Omega_{\mathrm{EB},0} \right\} + R_n.
\end{align*}
where the remainder satisfies $\|R_n\| \le C\|\widehat\Omega - \Omega_0\|^2$. Since each nuisance block is estimated at rate $O_p(n^{-1/2})$, we have $\|R_n\| = O_p(n^{-1})$.

\textit{Step 3: Substitution of influence functions.}
Substituting \eqref{eq:IF-V}--\eqref{eq:IF-EB} and using linearity of Fr\'echet derivatives:
\[
n^{1/2} U_n ( \psi_0; \widehat{\Omega} ) = n^{-1/2} \sum_{i=1}^n \phi_i^{(M)} + n^{-1/2} \sum_{i=1}^n \left\{ \phi_i^{(\mathcal{V})} + \phi_i^{(\mathcal{O})} + \phi_i^{(\mathrm{EB})} \right\} + o_p(1).
\]

\textit{Step 4: Independence of summands.}
Although $\widehat\Omega_{\mathcal V}$ affects all subjects' EB posteriors, the contribution decomposes into i.i.d.\ terms by a projection argument. Define $\mathcal{L} = \E \left\{ \mathcal{D}_{\mathrm{EB},1} U_n ( \psi_0; \Omega_0 ) \circ \nabla_{\Omega_{\mathcal{V}}} g_1 \right\}$. Then:
\[
n^{-1/2} \sum_{i=1}^n \phi_i^{(\mathrm{EB})} = n^{-1/2} \sum_{i=1}^n \mathcal{L} \left\{ \psi_i^{(\mathcal{V})} \right\} + o_p ( 1 ),
\]
where $\mathcal{L}\left\{\psi_i^{(\mathcal V)}\right\}$ depends only on subject $i$'s data. This ``averaging out'' argument is valid because the influence of subject $j$'s data on subject $i$'s EB posterior (through $\widehat\Omega_{\mathcal V}$) is $O(n^{-1})$, contributing $O(n^{-1/2})$ to the sum, which is absorbed in the $o_p(1)$ term.
\end{proof}

\subsubsection{Completion of Proof}

\begin{proof}[Proof of Theorem 2]
\textit{Step 1: Taylor expansion in $\psi$.}
Since $U_n(\widehat\psi; \widehat\Omega) = \bm{0}$ and $U_n(\psi; \widehat\Omega)$ is affine in $\psi$:
\[
\bm{0} = U_n(\psi_0; \widehat\Omega) + \frac{\partial U_n(\psi; \widehat\Omega)}{\partial\psi^\top}\bigg|_{\psi_0}(\widehat\psi - \psi_0).
\]
There is no remainder since $U_n$ is linear in $\psi$.

\textit{Step 2: Rearrangement.}
Multiplying by $n^{1/2}$ and rearranging:
\[
n^{1/2} ( \widehat{\psi} - \psi_0 ) = - \left\{ \left. \frac{ \partial U_n ( \psi; \widehat{\Omega} ) }{ \partial \psi^{\top} } \right|_{\psi_0} \right\}^{-1} n^{1/2} U_n ( \psi_0; \widehat{\Omega} )
\]

\textit{Step 3: Application of Slutsky's lemma.}
By Lemma~\ref{lem:jacobian}: $\partial U_n/\partial\psi^\top|_{\psi_0} \xrightarrow{p} -\bm{W}$.

By Proposition~\ref{prop:Un-expansion}: $n^{1/2}\,U_n(\psi_0; \widehat\Omega) = n^{-1/2}\sum_i \phi_i + o_p(1)$.

The $\{\phi_i\}$ are i.i.d.\ with $\E[\phi_i] = \bm{0}$ and $\Var(\phi_i) = \bm\Gamma < \infty$. By the multivariate CLT:
\[
n^{-1/2}\sum_{i=1}^n \phi_i \xrightarrow{d} \mathcal{N}(\bm{0}, \bm\Gamma).
\]

Slutsky's lemma yields:
\[
n^{1/2}(\widehat\psi - \psi_0) \xrightarrow{d} \mathcal{N}(\bm{0}, \bm{W}^{-1}\bm\Gamma\bm{W}^{-1}). \qedhere
\]
\end{proof}

\subsection{Explicit Forms of Variance Components}

\begin{corollary}[Variance components]\label{cor:variance}
The asymptotic variance $\bm{W}^{-1}\bm\Gamma\bm{W}^{-1}$ has components:
\begin{align}
\bm{W} &= \E \left\{ \int_0^\tau \bm{H}_{1,0} ( t ) \bm{H}_{1,0} ( t )^{\top} R_1^{\mathcal{Y}} ( t ) \dint \Lambda_1 ( t ) \right\}, \label{eq:W-final}\\
\bm{\Gamma} &= \Var ( \phi_1 ) = \Var \left\{ \phi_1^{(M)} + \phi_1^{(\mathcal{V})} + \phi_1^{(\mathcal{O})} + \phi_1^{(\mathrm{EB})} \right\}. \label{eq:Gamma-final}
\end{align}
The ``meat'' matrix $\bm\Gamma$ includes the variance of the primary martingale term and all cross-covariances among the four components.
\end{corollary}

\begin{proof}
Equation \eqref{eq:W-final} follows from Lemma~\ref{lem:W-equiv}. Equation \eqref{eq:Gamma-final} follows from the definition $\phi_i = \phi_i^{(M)} + \phi_i^{(\mathcal V)} + \phi_i^{(\mathcal O)} + \phi_i^{(\mathrm{EB})}$ and independence across subjects.
\end{proof}


\section{Extension: Time-Varying Observation Process with Lagged Recording Feedback}
\label{sec:tv_obs}

This section extends the proposed framework to accommodate a time-varying observation process in which the recording probability depends on the lagged recording indicator $R_i^{\mathcal Y}(t^-)$.
Such feedback captures clinically plausible scenarios where historical recording patterns influence subsequent testing decisions.

\subsection{Time-varying observation model}

Let $t_{i1} < \cdots < t_{im_i}$ denote the visit times for subject $i$ and let $R_{ij} \equiv R_i^{\mathcal Y}(t_{ij})$ be the recording indicator at visit $j$.
Define the lagged recording indicator as $R_i^{\mathcal Y}(t^-)$ with initial condition $R_i^{\mathcal Y}(t_{i1}^-) = 0$.

Conditional on a visit occurring, we model the recording probability via a probit mixed model with lagged feedback:
\begin{equation}
\pr \left\{ R_i^{\mathcal{Y}} ( t ) = 1 \mid \dint N_i ( t ) = 1, U_i, \bm{q}_i, \mathcal{F}_{i,t^-} \right\}
=
\Phi \left\{ \bm{\alpha}^{\top} \bm{X}_i^{\mathcal{O}} ( t ) + \bm{q}_i^{\top} \bm{Z}_i^{\mathcal{O}} ( t ) + \lambda R_i^{\mathcal{Y}} ( t^- ) \right\},
\label{eq:tv_obs_model}
\end{equation}
where $\bm{\alpha}$ are fixed-effect coefficients, $\bm{q}_i$ are subject-specific random slopes drawn once per patient, and $\lambda$ captures the effect of the lagged recording indicator.
As in the manuscript, $\bm{q}_i \mid U_i \sim N(\bm{\delta} U_i, \bm{\Sigma}_q)$.

The subject-level specification of $\bm{q}_i$ reflects persistent patient characteristics---such as health literacy, provider relationships, or care protocols---that induce correlation in recording decisions across visits beyond what $U_i$ and the lag term capture.

\subsection{Reparameterization via latent innovations}

To facilitate joint inference, we reparameterize the random effects.
Write $\bm{q}_i = \bm{\delta} U_i + \bm{\widetilde{q}}_i$ where $\bm{\widetilde{q}}_i \sim \mathcal{N}(\bm{0}, \bm{\Sigma}_q)$ is independent of $U_i$.
Define the augmented latent vector
\[
\bm{w}_i = (U_i, \bm{\widetilde{q}}_i^\top)^\top \in \mathbb{R}^{1+d},
\]
where $d = \dim(\bm{q}_i)$.
The prior distribution given Stage~1 is
\[
\bm{w}_i \mid (C_i, m_i) \sim \mathcal{N}(\bm{m}_{i,0}, \bm{V}_{i,0}),
\]
with $\bm{m}_{i,0} = (\mu_{U_i}, \bm{0}^\top)^\top$ and $\bm{V}_{i,0} = \mathrm{diag}(s_{U_i}^2, \bm{\Sigma}_q)$.

Under this reparameterization, the linear predictor in~\eqref{eq:tv_obs_model} becomes
\[
\eta_i(t) = a_i(t^-) + \bm{a}_i^\top \bm{w}_i,
\]
where $a_i(t^-) = \bm{\alpha}^\top \bm{X}_i^{\mathcal O}(t) + \lambda R_i^{\mathcal Y}(t^-)$ contains the history-dependent terms and $\bm{a}_i = \left[ \left\{ \bm{\delta}^{\top} \bm{Z}_i^{\mathcal{O}} ( t ) \right\}, \left\{ \bm{Z}_i^{\mathcal{O}} ( t )^{\top} \right\} \right]^{\top}$ collects the loadings on $\bm{w}_i$.

\subsection{Estimation procedure}

The three-stage structure is retained, with modifications to Stages~2--3 necessitated by~\eqref{eq:tv_obs_model}.

\subsubsection{Stage 1: Visiting process (unchanged)}

Stage~1 proceeds as in Section~\ref{sec:estimation}: we estimate $(\widehat{\bm{\gamma}}, \widehat{\Lambda}_0(t), \widehat{\sigma}_\zeta^2)$ and obtain the empirical Bayes approximation $U_i \mid (C_i, m_i) \approx \mathcal{N}(\mu_{U_i}, s_{U_i}^2)$ via Laplace approximation.

\subsubsection{Stage 2: Observation process via subject-level marginal likelihood}

Because $\bm{q}_i$ is drawn once per patient, the recording indicators $\{R_i^{\mathcal Y}(t_{ij})\}_{j=1}^{m_i}$ are correlated through the shared realization of $\bm{q}_i$.
A visit-wise composite likelihood that marginalizes $\bm{q}_i$ separately at each visit would incorrectly treat the $\bm{q}_i$'s as independent across visits:
\[
\int \pr ( \bm{q} \mid U ) \prod_{j=1}^{m_i} p \left\{ R_{ij} \mid U, \bm{q}, R_i^{\mathcal{Y}} ( t_{ij}^- ) \right\} \dint \bm{q} 
\neq 
\prod_{j=1}^{m_i} \int \pr \left\{ R_{ij} \mid U, \bm{q}, R_i^{\mathcal{Y}} ( t_{ij}^- ) \right\} p ( \bm{q} \mid U ) \dint \bm{q}.
\]

To respect this dependence, we maximize the subject-level marginal likelihood for the entire recording sequence.
Let $\Omega_{\mathcal O} = (\bm{\alpha}, \bm{\delta}, \bm{\Sigma}_q, \lambda)$.
The subject-level marginal likelihood is
\begin{equation}
\mathcal{L}_i(\Omega_{\mathcal O})
=
\iint
\prod_{j=1}^{m_i}
\pr_{ij}(u, \bm{q})^{R_{ij}}
\{1 - \pr_{ij}(u, \bm{q})\}^{1-R_{ij}}
\,
\phi(\bm{q}; \bm{\delta}u, \bm{\Sigma}_q)\,
\phi(u; \mu_{U_i}, s_{U_i}^2)\,
\dint \bm{q}\,\dint u,
\label{eq:tv_obs_marginal}
\end{equation}
where $\pr_{ij}(u, \bm{q}) = \Phi\!\left\{\bm{\alpha}^\top \bm{X}_i^{\mathcal O}(t_{ij}) + \bm{q}^\top \bm{Z}_i^{\mathcal O}(t_{ij}) + \lambda R_i^{\mathcal Y}(t_{ij}^-)\right\}$.
We estimate $\widehat{\Omega}_{\mathcal O}$ by maximizing $\sum_{i=1}^n \log \mathcal{L}_i(\Omega_{\mathcal O})$ via tensor-product Gauss--Hermite quadrature.

\subsubsection{Stage 3: Outcome estimation with joint filtering}

Stage~3 retains the weighted least squares structure, but the construction of the weight and compensation terms requires joint filtering on the augmented latent vector $\bm{w}_i = (U_i, \bm{\varepsilon}_i^\top)^\top$ to correctly track the posterior uncertainty in both components.

\paragraph{Joint filtering.}
Let $\mathcal{F}_{i,t^-}$ denote the $\sigma$-field generated by $\{C_i, m_i\}$ and the observation history up to time $t^-$.
We approximate the history-updated posterior by a Gaussian
\[
\bm{w}_i \mid \mathcal{F}_{i,t^-} \approx \mathcal{N} \left\{ \bm{m}_i ( t^- ), \bm{V}_i ( t^- ) \right\}
\]
initialized at $(\bm{m}_{i,0}, \bm{V}_{i,0})$ and updated sequentially via Laplace approximation as each $R_i^{\mathcal Y}(t_{ij})$ is observed.

\paragraph{Predictive quantities.}
Under the joint Gaussian approximation, the linear predictor $L_i(t) = a_i(t^-) + \bm{a}_i^\top \bm{w}_i$ is univariate normal conditional on $\mathcal{F}_{i,t^-}$:
\[
L_i ( t ) \mid \mathcal{F}_{i,t^-} \sim \mathcal{N} \left\{ \mu_{L_i} ( t^- ), \sigma_{L_i}^2 ( t^- ) \right\}
\]
with
\[
\mu_{L_i}(t^-) = a_i(t^-) + \bm{a}_i^\top \bm{m}_i(t^-), \qquad
\sigma_{L_i}^2(t^-) = \bm{a}_i^\top \bm{V}_i(t^-) \bm{a}_i.
\]

The predictive observation probability is
\begin{equation}
\overline{\omega}_i ( t^- )
= \E \left\{ \Phi ( L_i ( t ) ) \mid \mathcal{F}_{i,t^-} \right\}
= \Phi \left\{ \frac{ \mu_{L_i} ( t^- ) }{ \sqrt{ 1 + \sigma_{L_i}^2 ( t^- ) } } \right\}
\equiv \Phi \left\{ k_i ( t^- ) \right\}.
\label{eq:tv_omega_joint}
\end{equation}

For the observation-weighted posterior mean ratio, since $(U_i, L_i(t))$ are jointly Gaussian given $\mathcal{F}_{i,t^-}$ with $\mathrm{Cov}(U_i, L_i(t) \mid \mathcal{F}_{i,t^-}) = \bm{e}_1^\top \bm{V}_i(t^-) \bm{a}_i$ where $\bm{e}_1 = (1, 0, \ldots, 0)^\top$, Lemma~S2.3 yields
\begin{equation}
\kappa_i ( t^- )
= \mu_{U_i} ( t^- ) + \frac{ \bm{e}_1^{\top} \bm{V}_i ( t^- ) \bm{a}_i }{ \sqrt{ 1 + \sigma_{L_i}^2 ( t^- ) } } \cdot \frac{ \varphi \{ k_i ( t^- ) \} }{ \Phi \{ k_i ( t^- ) \} },
\label{eq:tv_kappa_joint}
\end{equation}
where $\mu_{U_i}(t^-) = \left\{\bm{m}_i(t^-)\right\}_1$ is the first component of the posterior mean.

\paragraph{Why joint filtering is necessary.}
The covariance term $\bm{e}_1^\top \bm{V}_i(t^-) \bm{a}_i$ captures how the history of observations updates the joint uncertainty between $U_i$ and $\bm{\varepsilon}_i$.
A one-dimensional filter that marginalizes $\bm{\varepsilon}_i$ before filtering would incorrectly reset this cross-covariance to zero at each step, leading to biased $\kappa_i(t^-)$ and inconsistent estimation of $\bm{\theta}$.

The definitions~\eqref{eq:tv_omega_joint}--\eqref{eq:tv_kappa_joint} ensure the key cancellation holds:
\[
\E\!\left[\{U_i - \kappa_i(t^-)\} R_i^{\mathcal Y}(t)\, \dint N_i(t) \mid \mathcal{F}_{i,t^-}\right] = 0.
\]

\paragraph{Weight function and risk-set centering.}
Define the time-varying weight function
\[
\widehat{p}_i(t^-) = \overline{\omega}_i(t^-)\,\mathbf{1}(t \leq C_i) \,\frac{m_i}{\widehat{\Lambda}_0(C_i)},
\]
and the compensated outcome covariate $\widehat{\bm{B}}_i(t^-) = \kappa_i(t^-)\,\widetilde{\bm{Z}}_i^{\mathcal Y}(t)$.
The risk-set averages
\[
\overline{\bm{X}}^{\mathcal Y}(t^-) = \frac{\sum_{i=1}^n \widehat{p}_i(t^-)\bm{X}_i^{\mathcal Y}(t)}{\sum_{i=1}^n \widehat{p}_i(t^-)},
\qquad
\widehat{\overline{\bm{B}}}(t^-) = \frac{\sum_{i=1}^n \widehat{p}_i(t^-)\widehat{\bm{B}}_i(t^-)}{\sum_{i=1}^n \widehat{p}_i(t^-)}
\]
are computed on the visit-time grid, and $(\widehat{\bm{\beta}}, \widehat{\bm{\theta}})$ is obtained from the WLS normal equations~\eqref{eq:EE_outcome}.

\section{Simulation Study for Time-Varying Observation Process}\label{sec:tv_obs_sim}

This section presents simulation studies evaluating the proposed extension for time-varying observation processes with lagged recording feedback, as described in Section~S8.

We generate data from a model where the observation process depends on the lagged recording indicator $R_i^Y(t^-)$. For each of $n = 1{,}000$ subjects, we generate baseline covariates $Z_i \sim \text{Bernoulli}(0.5)$ and $X_i \sim N(0, 1)$, and a shared latent factor $U_i \sim N(0, 1)$.

\paragraph{Visiting Process.}
The visit intensity follows a multiplicative frailty model:
\[
\lambda_i(t) = \eta_i \exp(-2.2 + Z_i + X_i) \lambda_0(t),
\]
where $\eta_i = \exp(\mu_0 + \sigma U_i)$ with $\sigma = 1$ and $\mu_0 = -0.5$.

\paragraph{Observation Process.}
Conditional on a visit at time $t_{ij}$, the recording indicator $R_{ij}$ is generated from:
\[
P(R_{ij} = 1 \mid U_i, q_i, R_{i,j-1}) = \Phi\bigl(0.5 Z_i - 0.5 X_i + 0.5 R_{i,j-1} + q_i^\top Z_{ij}^O \bigr),
\]
where $R_{i,0} = 0$, $Z_{ij}^O = (1, Z_i)^\top$, and the subject-specific random coefficients satisfy:
\[
\begin{pmatrix} q_{0i} \\ q_{1i} \end{pmatrix} \mid U_i \sim N \left\{ \begin{pmatrix} 0.2 \\ 0.6 \end{pmatrix} U_i, \begin{pmatrix} 0.2 & 0 \\ 0 & 0.1 \end{pmatrix} \right\}
\]

\paragraph{Longitudinal Outcome.}
The biomarker trajectory follows:
\[
Y_i(t) = -2 - 0.5 Z_i + 0.5 X_i + 0.1 t + b_{0i} + b_{1i} Z_i + \varepsilon_i(t),
\]
where $\varepsilon_i(t) \sim N(0, 1)$, and the random effects satisfy:
\[
\begin{pmatrix} b_{0i} \\ b_{1i} \end{pmatrix} \bigg| U_i \sim N \left\{ \begin{pmatrix} 0.5 \\ 0.2 \end{pmatrix} U_i, \begin{pmatrix} 1 & 0 \\ 0 & 2 \end{pmatrix} \right\}
\]
Under this simulation setting, we demonstrate that this modified procedure yields accurate estimation, with a bias of $-0.002$ and an RMSE of $0.301$.

\section{Invariance to Covariate-Dependent Latent Variables}
\label{sec:supp:invariance}

The GIVEHR estimation procedure assumes that the shared latent factor $U_i$ is independent of measured covariates $\bm{X}^{\mathcal{U}}_i$, with $U_i \sim \mathcal{N}(0, 1)$. In practice, however, the true data-generating process may exhibit a dependence $\E(U_i \mid \bm{X}^{\mathcal{U}}_i) \neq 0$ (the red dotted arrow in Figure~\ref{fig:DAG}). We show here that this dependence is non-identifiable from the observed data and is absorbed into the fixed-effect coefficients, leaving the estimation procedure and all variance-component estimates unchanged.

Suppose the true conditional distribution for the unmeasured factor is 
\[
U_i \mid\bm{X}^{\mathcal{U}}_i \sim \mathcal{N}(\bm{a}^\top \bm{X}^{\mathcal{U}}_i, 1)
\]
for some coefficient vector $\bm{a}$. Define the residual $\widetilde{U}_i = U_i - \bm{a}^\top \bm{X}^{\mathcal{U}}_i$, so that $\widetilde{U}_i \sim \mathcal{N}(0, 1)$ with $\widetilde{U}_i \perp \left\{ \bm{X}_i^{\mathcal{V}}, \bm{X}_i^{\mathcal{O}}(t), \bm{X}_i^{\mathcal{Y}}(t) \right\}$. Substituting $U_i = \bm{a}^\top \bm{X}^{\mathcal{U}}_i + \widetilde{U}_i$ into each component of the joint model shows that $\bm{a}$ is absorbed into the fixed effects, as we now demonstrate.

Consider first the visiting process. The true intensity is $\lambda_i(t) = \lambda_0(t)\exp(\bm{\gamma}^\top \bm{X}^{\mathcal{U}}_i + \mu_0 + \sigma U_i)$. After substitution,
\[
\lambda_i(t) = \lambda_0(t)\exp\!\big\{(\bm{\gamma} + \sigma\bm{a})^\top \bm{X}^{\mathcal{U}}_i + \mu_0 + \sigma\widetilde{U}_i\big\}.
\]
The reparametrized model has identical form with $\bm{\gamma}^* = \bm{\gamma} + \sigma\bm{a}$ and frailty $\widetilde{\eta}_i = \exp(\mu_0 + \sigma\widetilde{U}_i)$, where $\widetilde{U}_i \perp \left\{ \bm{X}_i^{\mathcal{V}}, \bm{X}_i^{\mathcal{O}}(t), \bm{X}_i^{\mathcal{Y}}(t) \right\}$. Since the partial likelihood conditions on the frailties, and the method-of-moments estimator for $\sigma^2$ and the Laplace approximation for the posterior of $\widetilde{U}_i$ depend only on the functional form of the frailty distribution, all three steps of the visiting-process estimation are exactly invariant to $\bm{a}$.

Turning to the observation process, the model specifies $\bm{q}_i \mid U_i \sim \mathcal{N}(\bm{\delta} U_i, \bm{\Sigma}_q)$, so the numerator of the observation kernel $k_i(t)$ inside $\Phi(\cdot)$ involves the linear predictor $\bm{\alpha}^\top \bm{X}_i^{\mathcal{O}} + \bm{\delta}^\top \bm{Z}_i^{\mathcal{O}}(t)\, U_i$. After substitution, the linear predictor becomes
\begin{align*}
\pr\left\{ R_i^{\mathcal{Y}} = 1 \mid \dint N_i(t) = 1, U_i, \bm{X}_i^{\mathcal{O}}(t), \bm{Z}_i^{\mathcal{O}}(t) \right\} &= \bm{\alpha}^\top \bm{X}_i^{\mathcal{O}} + \bm{\delta}^\top \bm{Z}_i^{\mathcal{O}}(t)\,(\bm{a}^\top \bm{X}^{\mathcal{U}}_i + \widetilde{U}_i) \\
&= \bm{\alpha}^\top \bm{X}_i^{\mathcal{O}} + \bm{\delta}^\top \bm{Z}_i^{\mathcal{O}}(t)\,\bm{a}^\top \bm{X}^{\mathcal{U}}_i + \bm{\delta}^\top \bm{Z}_i^{\mathcal{O}}(t)\, \widetilde{U}_i.
\end{align*}
When the random-effect design matrix reduces to an intercept ($\bm{Z}_i^{\mathcal{O}} = 1$), this simplifies to $(\bm{\alpha} + \delta_0 \bm{a})^\top \bm{X}_i^{\mathcal{O}} + \delta_0\,\widetilde{U}_i$, which yields the exact redefinition $\bm{\alpha}^* = \bm{\alpha} + \delta_0 \bm{a}$. When $\bm{Z}_i^{\mathcal{O}}$ contains additional covariates, such as a time-varying component $F_i$ with $\bm{Z}_i^{\mathcal{O}} = (1, F_i)^\top$, the absorption term expands to
\begin{align*}
\bm{\delta}^\top \bm{Z}_i^{\mathcal{O}}(t)\,\bm{a}^\top \bm{X}^{\mathcal{U}}_i &= (\delta_0 + \delta_1 F_i)\,\bm{a}^\top \bm{X}^{\mathcal{U}}_i \\
&= \delta_0\,\bm{a}^\top \bm{X}^{\mathcal{U}}_i + \delta_1 F_i\,\bm{a}^\top \bm{X}^{\mathcal{U}}_i,
\end{align*}
generating interactions not spanned by the original fixed-effect specification. This interaction term can be eliminated by including the appropriate cross-product terms in the fixed-effect design matrix.

The same reasoning applies to the outcome process. The longitudinal outcome is $Y_i(t) = \beta_0(t) + \bm{\beta}^\top \bm{X}^{\mathcal{Y}}_i(t) + \bm{b}_i^\top \bm{Z}_i^{\mathcal{Y}}(t) + \varepsilon_i(t)$, where the random effects satisfy $\bm{b}_i \mid U_i \sim \mathcal{N}(\bm{\theta} U_i, \bm{\Sigma}_b)$. Writing $\bm{b}_i = \bm{\theta} U_i + \widetilde{\bm{b}}_i$ with $\widetilde{\bm{b}}_i \sim \mathcal{N}(\bm{0}, \bm{\Sigma}_b)$ and substituting $U_i = \bm{a}^\top \bm{X}^{\mathcal{U}}_i + \widetilde{U}_i$, the outcome becomes
\begin{align*}
Y_i(t) &= \beta_0(t) + \bm{\beta}^\top \bm{X}^{\mathcal{Y}}_i(t) + \big\{\bm{\theta}(\bm{a}^\top \bm{X}^{\mathcal{U}}_i + \widetilde{U}_i) + \widetilde{\bm{b}}_i\big\}^\top \bm{Z}_i^{\mathcal{Y}}(t) + \varepsilon_i(t) \\
&= \beta_0(t) + \bm{\beta}^\top \bm{X}^{\mathcal{Y}}_i(t) + \bm{\theta}^\top \bm{Z}_i^{\mathcal{Y}}(t)\,\bm{a}^\top \bm{X}^{\mathcal{U}}_i + \bm{\theta}^\top \bm{Z}_i^{\mathcal{Y}}(t)\,\widetilde{U}_i + \widetilde{\bm{b}}_i^\top \bm{Z}_i^{\mathcal{Y}}(t) + \varepsilon_i(t).
\end{align*}
When $\bm{Z}_i^{\mathcal{Y}} = 1$, this reduces to
\begin{align*}
Y_i(t) &= \beta_0(t) + (\bm{\beta} + \theta_0 \bm{a})^\top \bm{X}^{\mathcal{Y}}_i(t) + \theta_0\,\widetilde{U}_i + \widetilde{b}_{0i} + \varepsilon_i(t),
\end{align*}
yielding the exact redefinition $\bm{\beta}^* = \bm{\beta} + \theta_0 \bm{a}$. When $\bm{Z}_i^{\mathcal{Y}}$ includes covariates $F_i$ with $\bm{Z}_i^{\mathcal{Y}} = (1, F_i)^\top$, the absorption term expands to
\begin{align*}
\bm{\theta}^\top \bm{Z}_i^{\mathcal{Y}}(t)\,\bm{a}^\top \bm{X}^{\mathcal{U}}_i &= (\theta_0 + \theta_1 \cdot F_i)\,\bm{a}^\top \bm{X}^{\mathcal{U}}_i \\
&= \theta_0\,\bm{a}^\top \bm{X}^{\mathcal{U}}_i + \theta_1\,F_i \cdot \bm{a}^\top \bm{X}^{\mathcal{U}}_i.
\end{align*}
If needed, the interaction terms can be accommodated by including the corresponding covariate interactions in the fixed-effect specification.

As a result, across all three submodels the vector $\bm{a}$ is fully absorbed into the fixed-effect coefficients, and the residual latent factor $\widetilde{U}_i$ satisfies the independence and distributional assumptions of the fitted model. The observed-data likelihood under the true data-generating process with $\E(U_i \mid \bm{X}^{\mathcal{U}}_i) = \bm{a}^\top \bm{X}^{\mathcal{U}}_i$ is therefore identical to that under the fitted model with $U_i \perp \bm{X}^{\mathcal{U}}_i$, up to a relabeling of fixed-effect parameters. This equivalence implies that $\bm{a}$ is not identifiable from the data, and that the estimated fixed-effect coefficients $\widehat{\bm{\gamma}}^*$, $\widehat{\bm{\alpha}}^*$, and $\widehat{\bm{\beta}}^*$ capture the total association between covariates and the outcome---encompassing both the direct pathway and the indirect pathway mediated through the latent variable---rather than the direct causal effect alone.

\end{document}